\documentclass{article}
\usepackage{amsmath}
\usepackage{amssymb}\usepackage{amsthm}
\usepackage{amsfonts}
\begin{document}
\theoremstyle{plain}
\newtheorem{thm}{Theorem}[section]
\newtheorem{lem}[thm]{Lemma}
\newtheorem{prop}[thm]{Proposition}
\newtheorem{cor}[thm]{Corollary}
\theoremstyle{definition}
\newtheorem{assum}[thm]{Assumption}
\newtheorem{notation}[thm]{Notation}
\newtheorem{defn}[thm]{Definition}
\newtheorem{clm}[thm]{Claim}
\newtheorem{ex}[thm]{Example}
\theoremstyle{remark}
\newtheorem{rem}[thm]{Remark}
\newcommand{\unit}{\mathbb I}
\newcommand{\ali}[1]{{\mathfrak A}_{[ #1 ,\infty)}}
\newcommand{\alm}[1]{{\mathfrak A}_{(-\infty, #1 ]}}
\newcommand{\nn}[1]{\lV #1 \rV}
\newcommand{\br}{{\mathbb R}}
\newcommand{\dm}{{\rm dom}\mu}
\newcommand{\lb}{l_{\bb}(n,n_0,k_R,k_L,\lal,\bbD,\bbG,Y)}
\newcommand{\Ad}{\mathop{\mathrm{Ad}}\nolimits}
\newcommand{\Proj}{\mathop{\mathrm{Proj}}\nolimits}
\newcommand{\RRe}{\mathop{\mathrm{Re}}\nolimits}
\newcommand{\RIm}{\mathop{\mathrm{Im}}\nolimits}
\newcommand{\Wo}{\mathop{\mathrm{Wo}}\nolimits}
\newcommand{\Prim}{\mathop{\mathrm{Prim}_1}\nolimits}
\newcommand{\Primz}{\mathop{\mathrm{Prim}}\nolimits}
\newcommand{\ClassA}{\mathop{\mathrm{ClassA}}\nolimits}
\newcommand{\Class}{\mathop{\mathrm{Class}}\nolimits}
\def\qed{{\unskip\nobreak\hfil\penalty50
\hskip2em\hbox{}\nobreak\hfil$\square$
\parfillskip=0pt \finalhyphendemerits=0\par}\medskip}
\def\proof{\trivlist \item[\hskip \labelsep{\bf Proof.\ }]}
\def\endproof{\null\hfill\qed\endtrivlist\noindent}
\def\proofof[#1]{\trivlist \item[\hskip \labelsep{\bf Proof of #1.\ }]}
\def\endproofof{\null\hfill\qed\endtrivlist\noindent}
\newcommand{\caA}{{\mathcal A}}
\newcommand{\caB}{{\mathcal B}}
\newcommand{\caC}{{\mathcal C}}
\newcommand{\caD}{{\mathcal D}}
\newcommand{\caE}{{\mathcal E}}
\newcommand{\caF}{{\mathcal F}}
\newcommand{\caG}{{\mathcal G}}
\newcommand{\caH}{{\mathcal H}}
\newcommand{\caI}{{\mathcal I}}
\newcommand{\caJ}{{\mathcal J}}
\newcommand{\caK}{{\mathcal K}}
\newcommand{\caL}{{\mathcal L}}
\newcommand{\caM}{{\mathcal M}}
\newcommand{\caN}{{\mathcal N}}
\newcommand{\caO}{{\mathcal O}}
\newcommand{\caP}{{\mathcal P}}
\newcommand{\caQ}{{\mathcal Q}}
\newcommand{\caR}{{\mathcal R}}
\newcommand{\caS}{{\mathcal S}}
\newcommand{\caT}{{\mathcal T}}
\newcommand{\caU}{{\mathcal U}}
\newcommand{\caV}{{\mathcal V}}
\newcommand{\caW}{{\mathcal W}}
\newcommand{\caX}{{\mathcal X}}
\newcommand{\caY}{{\mathcal Y}}
\newcommand{\caZ}{{\mathcal Z}}
\newcommand{\bbA}{{\mathbb A}}
\newcommand{\bbB}{{\mathbb B}}
\newcommand{\bbC}{{\mathbb C}}
\newcommand{\bbD}{{\mathbb D}}
\newcommand{\bbE}{{\mathbb E}}
\newcommand{\bbF}{{\mathbb F}}
\newcommand{\bbG}{{\mathbb G}}
\newcommand{\bbH}{{\mathbb H}}
\newcommand{\bbI}{{\mathbb I}}
\newcommand{\bbJ}{{\mathbb J}}
\newcommand{\bbK}{{\mathbb K}}
\newcommand{\bbL}{{\mathbb L}}
\newcommand{\bbM}{{\mathbb M}}
\newcommand{\bbN}{{\mathbb N}}
\newcommand{\bbO}{{\mathbb O}}
\newcommand{\bbP}{{\mathbb P}}
\newcommand{\bbQ}{{\mathbb Q}}
\newcommand{\bbR}{{\mathbb R}}
\newcommand{\bbS}{{\mathbb S}}
\newcommand{\bbT}{{\mathbb T}}
\newcommand{\bbU}{{\mathbb U}}
\newcommand{\bbV}{{\mathbb V}}
\newcommand{\bbW}{{\mathbb W}}
\newcommand{\bbX}{{\mathbb X}}
\newcommand{\bbY}{{\mathbb Y}}
\newcommand{\bbZ}{{\mathbb Z}}
\newcommand{\str}{^*}
\newcommand{\lv}{\left \vert}
\newcommand{\rv}{\right \vert}
\newcommand{\lV}{\left \Vert}
\newcommand{\rV}{\right \Vert}
\newcommand{\la}{\left \langle}
\newcommand{\ra}{\right \rangle}
\newcommand{\ltm}{\left \{}
\newcommand{\rtm}{\right \}}
\newcommand{\lcm}{\left [}
\newcommand{\rcm}{\right ]}
\newcommand{\ket}[1]{\lv #1 \ra}
\newcommand{\bra}[1]{\la #1 \rv}
\newcommand{\lmk}{\left (}
\newcommand{\rmk}{\right )}
\newcommand{\al}{{\mathcal A}}
\newcommand{\md}{M_d({\mathbb C})}
\newcommand{\Tr}{\mathop{\mathrm{Tr}}\nolimits}
\newcommand{\Ran}{\mathop{\mathrm{Ran}}\nolimits}
\newcommand{\Ker}{\mathop{\mathrm{Ker}}\nolimits}
\newcommand{\spn}{\mathop{\mathrm{span}}\nolimits}
\newcommand{\Mat}{\mathop{\mathrm{M}}\nolimits}
\newcommand{\UT}{\mathop{\mathrm{UT}}\nolimits}
\newcommand{\DT}{\mathop{\mathrm{DT}}\nolimits}
\newcommand{\GL}{\mathop{\mathrm{GL}}\nolimits}
\newcommand{\spa}{\mathop{\mathrm{span}}\nolimits}
\newcommand{\supp}{\mathop{\mathrm{supp}}\nolimits}
\newcommand{\rank}{\mathop{\mathrm{rank}}\nolimits}
\newcommand{\idd}{\mathop{\mathrm{id}}\nolimits}
\newcommand{\ran}{\mathop{\mathrm{Ran}}\nolimits}
\newcommand{\dr}{ \mathop{\mathrm{d}_{{\mathbb R}^k}}\nolimits} 
\newcommand{\dc}{ \mathop{\mathrm{d}_{\cc}}\nolimits} \newcommand{\drr}{ \mathop{\mathrm{d}_{\rr}}\nolimits} 
\newcommand{\zin}{\mathbb{Z}}
\newcommand{\rr}{\mathbb{R}}
\newcommand{\cc}{\mathbb{C}}
\newcommand{\ww}{\mathbb{W}}
\newcommand{\nan}{\mathbb{N}}\newcommand{\bb}{\mathbb{B}}
\newcommand{\aaa}{\mathbb{A}}\newcommand{\ee}{\mathbb{E}}
\newcommand{\pp}{\mathbb{P}}
\newcommand{\wks}{\mathop{\mathrm{wk^*-}}\nolimits}
\newcommand{\he}{\hat {\mathbb E}}
\newcommand{\ikn}{{\caI}_{k,n}}
\newcommand{\mk}{{\Mat_k}}
\newcommand{\mnz}{\Mat_{n_0}}
\newcommand{\mn}{\Mat_{n}}
\newcommand{\mkk}{\Mat_{k_R+k_L+1}}
\newcommand{\mnzk}{\mnz\otimes \mkk}
\newcommand{\hbb}{H^{k,\bb}_{m,p,q}}
\newcommand{\gb}[1]{\Gamma^{(R)}_{#1,\bb}}
\newcommand{\cgv}[1]{\caG_{#1,\vv}}
\newcommand{\gv}[1]{\Gamma^{(R)}_{#1,\vv}}
\newcommand{\gvt}[1]{\Gamma^{(R)}_{#1,\vv(t)}}
\newcommand{\gbt}[1]{\Gamma^{(R)}_{#1,\bb(t)}}
\newcommand{\cgb}[1]{\caG_{#1,\bb}}
\newcommand{\cgbt}[1]{\caG_{#1,\bb(t)}}
\newcommand{\gvp}[1]{G_{#1,\vv}}
\newcommand{\gbp}[1]{G_{#1,\bb}}
\newcommand{\gbpt}[1]{G_{#1,\bb(t)}}
\newcommand{\Pbm}[1]{\Phi_{#1,\bb}}
\newcommand{\Pvm}[1]{\Phi_{#1,\bb}}
\newcommand{\mb}{m_{\bb}}
\newcommand{\E}[1]{\widehat{\mathbb{E}}^{(#1)}}
\newcommand{\lal}{{\boldsymbol\lambda}}
\newcommand{\rar}{{\boldsymbol r}}
\newcommand{\oo}{{\boldsymbol\omega}}
\newcommand{\vv}{{\boldsymbol v}}
\newcommand{\bbm}{{\boldsymbol m}}
\newcommand{\kl}[1]{{\mathcal K}_{#1}}
\newcommand{\wb}[1]{\widehat{B_{\mu^{(#1)}}}}
\newcommand{\ws}[1]{\widehat{\psi_{\mu^{(#1)}}}}
\newcommand{\wsn}[1]{\widehat{\psi_{\nu^{(#1)}}}}
\newcommand{\wv}[1]{\widehat{v_{\mu^{(#1)}}}}
\newcommand{\wbn}[1]{\widehat{B_{\nu^{(#1)}}}}
\newcommand{\wo}[1]{\widehat{\omega_{\mu^{(#1)}}}}
\newcommand{\dist}{\dc}
\newcommand{\hpu}{\hat P^{(n_0,k_R,k_L)}_R}
\newcommand{\hpd}{\hat P^{(n_0,k_R,k_L)}_L}
\newcommand{\pu}{ P^{(k_R,k_L)}_R}
\newcommand{\pd}{ P^{(k_R,k_L)}_L}
\newcommand{\puuz}{P_{R}^{(n_0-1,n_0-1)}\otimes P^{(k_R,k_L)}_R}
\newcommand{\pddz}{P_{L}^{(n_0-1,n_0-1)}\otimes P^{(k_R,k_L)}_L}
\newcommand{\puu}{\tilde P_R}
\newcommand{\pdd}{\tilde P_L}
\newcommand{\qu}[1]{ Q^{(k_R,k_L)}_{R, #1}}
\newcommand{\qd}[1]{ Q^{(k_R,k_L)}_{L,#1}}
\newcommand{\hqu}[1]{ \hat Q^{(n_0,k_R,k_L)}_{R, #1}}
\newcommand{\hqd}[1]{ \hat Q^{(n_0,k_R,k_L)}_{L,#1}}
\newcommand{\eij}[1] {E^{(k_R,k_L)}_{#1}}
\newcommand{\eijz}[1] {E^{(n_0-1,n_0-1)}_{#1}}
\newcommand{\heij}[1] {\hat E^{(k_R,k_L)}_{#1}}
\newcommand{\cn}{\mathop{\mathrm{CN}(n_0,k_R,k_L)}\nolimits}
\newcommand{\ghd}[1]{\mathop{\mathrm{GHL}(#1,n_0,k_R,k_L,\bbG)}\nolimits}
\newcommand{\ghu}[1]{\mathop{\mathrm{GHR}(#1,n_0,k_R,k_L,\bbD)}\nolimits}
\newcommand{\ghdb}[1]{\mathop{\mathrm{GHL}(#1,n_0,k_R,k_L,\bbG)}\nolimits}
\newcommand{\ghub}[1]{\mathop{\mathrm{GHR}(#1,n_0,k_R,k_L,\bbD)}\nolimits}
\newcommand{\hfu}[1]{{\mathfrak H}_{#1}^R}
\newcommand{\hfd}[1]{{\mathfrak H}_{#1}^L}
\newcommand{\hfui}[1]{{\mathfrak H}_{#1,1}^R}
\newcommand{\hfdi}[1]{{\mathfrak H}_{#1,1}^L}
\newcommand{\hfuz}[1]{{\mathfrak H}_{#1,0}^R}
\newcommand{\hfdz}[1]{{\mathfrak H}_{#1,0}^L}
\newcommand{\CN}{\overline{\hpd}\lmk\mnzk \rmk\overline{\hpu}}
\newcommand{\cnz}[1] {\chi_{#1}^{(n_0)}}
\newcommand{\eu}{\eta_{R}^{(k_R,k_L)}}
\newcommand{\ezu}{\eta_{R}^{(n_0-1,n_0-1)}}
\newcommand{\ed}{\eta_{L}^{(k_R,k_L)}}
\newcommand{\ezd}{\eta_{L}^{(n_0-1,n_0-1)}}
\newcommand{\fii}[1]{f_{#1}^{(k_R,k_L)}}
\newcommand{\fiir}[1]{f_{#1}^{(k_R,0)}}
\newcommand{\fiil}[1]{f_{#1}^{(0,k_L)}}
\newcommand{\fiz}[1]{f_{#1}^{(n_0-1,n_0-1)}}
\newcommand{\zeij}[1] {e_{#1}^{(n_0)}}
\newcommand{\CL}{\ClassA}
\newcommand{\CLn}{\Class_2(n,n_0,k_R,k_L)}
\newcommand{\braket}[2]{\left\langle#1,#2\right\rangle}
\newcommand{\abs}[1]{\left\vert#1\right\vert}
\newtheorem{nota}{Notation}[section]
\def\qed{{\unskip\nobreak\hfil\penalty50
\hskip2em\hbox{}\nobreak\hfil$\square$
\parfillskip=0pt \finalhyphendemerits=0\par}\medskip}
\def\proof{\trivlist \item[\hskip \labelsep{\bf Proof.\ }]}
\def\endproof{\null\hfill\qed\endtrivlist\noindent}
\def\proofof[#1]{\trivlist \item[\hskip \labelsep{\bf Proof of #1.\ }]}
\def\endproofof{\null\hfill\qed\endtrivlist\noindent}
\title{A class of asymmetric gapped Hamiltonians on quantum spin chains and its characterization I}
\author{
{\sc Yoshiko Ogata}\footnote{Supported in part by
the Grants-in-Aid for
Scientific Research, JSPS.}\\
{\small Graduate School of Mathematical Sciences}\\
{\small The University of Tokyo, Komaba, Tokyo, 153-8914, Japan}
}

\maketitle{}

\begin{abstract}
We introduce a class of gapped Hamiltonians on quantum spin chains,
which allows asymmetric edge ground states.
This class is an asymmetric generalization of the class of Hamiltonians in 
\cite{Fannes:1992vq}. It can be characterized by
five qualitative physical properties of ground state structures.
In this Part I, we introduce the models and investigate their properties.
\end{abstract}

\section{Introduction}
Recently, gapped ground state phases attract a lot of attention. It is related to various fields of mathematics and physics, including condensed matter physics,
quantum information, spectral theory, and topology, and studied widely from  many different points of view.
Even if we restrict our attention to quantum spin systems,  many interesting facts have been found in the last decades.
One of the most famous discoveries is the area law  in one dimensional quantum spin system proven by Hastings \cite{area}.
In the paper, he showed and used the fact that a unique ground state of a gapped Hamiltonian 
can be approximated by a product of three localized operators. This fact holds even in more general setting, see \cite{hmns}.  
Furthermore, the exponential decay of correlations of gapped ground states was proven in \cite{hk}, \cite{ns}.
In a word, the existence of the gap guarantees us to have a good control on the spectral projection of ground state spaces.

Such a nice control was used to show the automorphic equivalence of 
the ground state structures in \cite{Bachmann:2011kw}, in the classification problem of gapped Hamiltonians.
Here, two Hamiltonians are defined to be in the same class, if they are the endpoints of a $C^1$-path of Hamiltonians along which the spectral gap above the ground state energy does not close. 
(See \cite{bo} for a more formal definition.)
We would like to emphasize that in  \cite{Bachmann:2011kw},  finite volume Hamiltonians with open boundary conditions are considered.
Therefore, we call this classification, the $C^1$-classification with open boundary condition. One benefit of considering open boundary condition
is that it possesses the information of edge states, as well as the bulk one. Another more technical advantage is that it is convenient when we use the martingaele method 
introduced in \cite{Nachtergaele:1996vc} to show the spectral gap.

As we have seen, one can derive strong results under quite general setting, {\it if we assume the existence of the spectral gap},
because of the nice control of the spectral projections.
However, 
{\it to prove the existence of the gap itself} turns out to be a much more difficult problem, especially in more than one dimensional systems.
For one dimensional systems,
a recipe to construct Hamiltonians out of $n$-tuple of matrices in some auxiliary systems is known. (See  Subsection \ref{subsec:ph}.)
We would like to call
the Hamiltonians given by this recipe, the MPS (matrix product state) Hamiltonians.
Just a random choice of the $n$-tuples does not guarantee the spectral gap. In \cite{Fannes:1992vq},
a sufficient condition, which we call the injectivity condition, to guarantee the gap was introduced. (See Remark \ref{rem:inj}.)
In this setting of \cite{Fannes:1992vq}, the bulk ground state turns out to be a pure finitely correlated state (or MPS).
The class of Hamiltonians given in this way with the injectivity condition covers the AKLT model \cite{Affleck:1988vr}.
Classification of this class of Hamiltonians was studied in \cite{Chen:2010gb}, \cite{Chen:2011iq}, \cite{Schuch:2011ve},
\cite {bo}.

Having this nice recipe, one natural and basic question is how much of the gapped ground state phase in quantum spin chains
is covered by the Hamiltonians given by MPS recipe with the injectivity condition.
More precisely, for each equivalence class of the $C^1$-classification,
is there a representative
given by such a Hamiltonian?

The answer is no: If we assume the injectivity condition, the edge states (namely, the ground state space on left/right half infinite chains) have to be symmetric \cite{bo}.
However, in ~\cite{Bachmann:2012uu, Bachmann:2012bf},
a particular family of gapped models
called PVBS models, with asymmetric edge states was introduced.
PVBS model still can be given by the MPS recipe, but the matrices do not satisfy the injectivity condition.
In this model, the bulk ground state is
 a pure product state.
It is then natural to explore a general condition on $n$-tuple of matrices which guarantees the spectral gap,
but still allows the asymmetric edge ground states.

This is what we do in this series of papers. 
In Part I, we introduce a new class of MPS gapped Hamiltonians which covers 
Hamiltonians in \cite{Fannes:1992vq}. 
This class allows asymmetric edge ground states. 
We investigate the properties of this Hamiltonian in this paper (Theorem \ref{thm:asymmetric}).
This new class is not a mere generalization of the known models. In Part II, we will show that
this class has a characterization in terms of five qualitative physical conditions (which corresponds to (i)-(v) and (viii)) of Theorem \ref{thm:asymmetric})  .
More precisely, we will show if a (not necessarily MPS) Hamiltonian satisfies these five conditions, there is an MPS Hamiltonian from our class
satisfying the followings: The ground state spaces  of the two Hamiltonians on the infinite intervals coincide.
In the finite intervals, the spectral projections onto the ground state space of the original Hamiltonian on each intervals
are well approximated by that of the MPS Hamiltonian. This last property has a corollary to the $C^1$-classification
problem i.e., the classification problem  of Hamiltonians satisfying these five properties is reduced to the
classification problem of our generalized class of MPS Hamiltonians.
The benefit is that the latter one has a concrete structure  which allows us to handle the spectral gap.

We should emphasis two points. Firstly, it is well known that any vector state on a finite interval can be represented as an MPS.
However, naive representation would require the size of the auxiliary systems to grow very fast, as the length of the intervals goes to infinity.
If  it is the unique ground state of a gapped Hamiltonian or even just by having the exponential decay of the correlation functions,
the growth can be reduced significantly \cite{area}, \cite{bh}.
What we do in this paper is however, of  different nature. We would like to have the dimension of the auxiliary system to be fixed.
As a cost, we have to assume that our Hamiltonian is frustration free.
The second point is that if we care only about the ground state in the bulk, it is already known that any ground state of frustration free Hamiltonian
with uniformly bounded degeneracy
is  an MPS \cite{Matsui1}, \cite{Matsui2}.
The difference here is that we would like to care about the spectral gap, the $C^1$-classification, and the edge states.
As a result, we have to represent not only bulk ground state, but also left/right edge ground states and  the  ground states on the finite intervals,
simultaneously, using  the {\it  same}  auxiliary system. 
This requires the detailed analysis of the auxiliary system.

We will use the notations listed in Appendix \ref{sec:nota} freely.

\subsection{Hamiltonians and ground state structures}
For $\bbN\ni n\geq 2$, let $\caA$ be the finite dimensional C*-algebra $\caA = \Mat_n$, the algebra of $n\times n$ matrices. 
Throughout this article, this $n$ is fixed as the dimension of the spin under consideration, and we fix an orthonormal basis $\{\psi_\mu\}_{\mu=1}^n$ of $\cc^n$.
We denote the set of all finite subsets in $\Gamma\subset{\bbZ}$ by ${\mathfrak S}_{\Gamma}$.
The number of elements in a finite set $\Lambda\subset {\bbZ}$ is denoted by
$|\Lambda|$.
When we talk about intervals in $\bbZ$, $[a,b]$ for $a\le b$,
means the interval in $\bbZ$, i.e., $[a,b]\cap \bbZ$.
We denote the set of all finite intervals in $\Gamma$
by ${\mathfrak I}_{\Gamma}$.
For each $z\in\bbZ$, we let $\caA_{\{z\}}$ be an isomorphic copy of $\caA$ and for any finite subset $\Lambda\subset\bbZ$, $\caA_{\Lambda} = \otimes_{z\in\Lambda}\caA_{\{z\}}$ is the local algebra of observables. 
For finite $\Lambda$, the algebra $\caA_{\Lambda} $ can be regarded as the set of all bounded operators acting on
a Hilbert space $\otimes_{z\in\Lambda}{\bbC}^n$.
We use this identification freely.
If $\Lambda_1\subset\Lambda_2$, the algebra $\caA_{\Lambda_1}$ is naturally embedded in $\caA_{\Lambda_2}$ by tensoring its elements with the identity. Finally,
for an infinite subset $\Gamma$ of $\bbZ$, the algebra $\caA_{\Gamma}$ 
is given as the inductive limit of the algebras $\caA_{\Lambda}$ with $\Lambda\in{\mathfrak S}_{\Gamma}$. In particular,
$\caA_{\bbZ}$  is the chain algebra.
We denote the set of local observables in $\Gamma$ by $\caA_{\Gamma}^{\rm loc}=\bigcup_{\Lambda\in{\mathfrak S}_{\Gamma}}\caA_{\Lambda}
$.
For $\omega$ a state on $\caA_{\Gamma}$ and
each finite $\Lambda\subset \Gamma$, we denote by $D_{\omega\vert_{\caA_\Lambda}}$ the density matrix of the restriction $\omega\vert_{\caA_\Lambda}$.

For any $x\in\bbZ$, let $\tau_x$ be the shift operator by $x$ on $\caA_\bbZ$. 
An interaction is a map $\Phi$ from 
${\mathfrak S}_{\bbZ}$ into ${\caA}_{{\mathbb Z}}^{\rm loc}$ such
that $\Phi(X) \in {\caA}_{X}$ 
and $\Phi(X) = \Phi(X)^*$
for $X \in {\mathfrak S}_{\bbZ}$. 
An interaction $\Phi$ is translation invariant if
$
\Phi(X+j)=\tau_j\lmk
\Phi(X)\rmk,
$
for all $ j\in{\mathbb Z}$ and $X\in  {\mathfrak S}_{\bbZ}$.
Furthermore, it is of finite range if there exists an $m\in {\mathbb N}$ such that
$
\Phi(X)=0$,
for $X$ with diameter larger than $m$.
In this case, 
we say that the interaction length of $\Phi$ is less than or equal to $m$.
A natural number $m\in\nan$ and an element $h\in\caA_{[0,m-1]}$,
define an interaction $\Phi_h$ by
\begin{align}\label{hamdef}
\Phi_h(X):=\left\{
\begin{gathered}
\tau_x\lmk h\rmk,\quad \text{if}\quad  X=[x,x+m-1] \quad \text{for some}\quad  x\in\bbZ\\
0,\quad\text{otherwise}
\end{gathered}\right.
\end{align}
for $X\in {\mathfrak S}_{\bbZ}$.

A Hamiltonian associated with $\Phi$ is a net of self-adjoint operators $H_{\Phi}:=\left((H_{\Phi })_\Lambda\right)_{\Lambda\in{\mathfrak I}_{\bbZ}}$ such that 
\begin{equation}\label{GenHamiltonian}
\lmk H_{\Phi}\rmk_{\Lambda}:=\sum_{X\subset{\Lambda}}\Phi(X).
\end{equation}
Note that $(H_\Phi)_{\Lambda}\in {\caA}_{\Lambda}$. Without loss of generality we consider positive interactions i.e.,  $\Phi(X)\geq 0$
for any $X\in{\mathfrak S}_{\bbZ}$, throughout this article.
Let us specify what we mean by gapped  with respect to the open boundary conditions Hamiltonian in this paper.
\begin{defn}
A Hamiltonian $H:=\left(H_\Lambda\right)_{\Lambda\in{\mathfrak I}_{\bbZ}}$
associated with a positive translation invariant finite range interaction is \emph{gapped with respect to the open boundary conditions}  
if there exist $\gamma>0$ and $N_0\in\nan$ 
such that the difference 
between the smallest and the next-smallest eigenvalue of $H_\Lambda$, is bounded below by 
$\gamma$,  for all finite intervals $\Lambda\subset\bbZ$ with $|\Lambda|\ge N_0$. 
\end{defn}
In this definition, the smallest eigenvalue can be degenerated in general.
Let $H=(H_{\Lambda})$ be a Hamiltonian associated with some positive 
translation invariant finite range interaction.
For a finite interval $\Lambda$, a ground state of
$H_{\Lambda}$ means a state on $\caA_{\Lambda}$
with support in the smallest eigenvalue space of $H_{\Lambda}$.
We denote the set of all ground states of $H_{\Lambda}$
on $\caA_{\Lambda}$ by $\caS_{\Lambda}(H)$.
For $\Lambda\in{\mathfrak I}_{\Gamma}$, any of the elements in 
 $\caS_{\Lambda}(H)$ can be extended to a 
state on $\caA_{\Gamma}$, and there exit weak-$*$
accumulation points of such extensions,
in the thermodynamical limit $\Lambda\to\Gamma$.
We denote the set of all such accumulation points by
$\caS_{\Gamma}(H)$.
\begin{defn}
We call the quadruplet $\left(\left\{\caS(H)\right\}_{I\in{\mathfrak I}_{\Gamma}},
\caS_{(-\infty-1]}(H),\caS_{[0,\infty)}(H),\caS_{\bbZ}(H)\right)$ 
the ground state structure of the Hamiltonian $H$.
\end{defn}

\subsection{The parent Hamiltonians}\label{subsec:ph}
The Hamiltonians we introduce in this paper are parent Hamiltonians of
sequence of subspaces satisfying the intersection property.
We say that a sequence of subspaces $\{\caD_N\}_{N\in \nan}$, $\caD_N\subset \bigotimes_{i=0}^{N-1}\cc^n$, $N\in\nan$, satisfies the {\it intersection property},
if there exists an $m\in\nan$, such that the relation 
\begin{equation}\label{intersection property}
\caD_N = \bigcap_{x=0}^{N-m} (\bbC^{n})^{\otimes x}\otimes \caD_{m}\otimes (\bbC^{n})^{\otimes N-m-x}, 
\end{equation} holds for all $N\ge m$.
In order to specify the number $m\in \nan$, we will say that $\{\caD_N\}_{N\in\nan}$ satisfies  Property~(I,$m$) 
when (\ref{intersection property}) holds for $m$
and all $N\ge m$.
Note that Property~(I,$m$) implies Property~(I,$m'$)  for all $m'\ge m$.
\begin{defn}
For a sequence of subspaces $\{\caD_N\}_{N\in \nan}$, $\caD_N\subset \bigotimes_{i=0}^{N-1}\cc^n$, $N\in\nan$,
we define $\bbm_{\{\caD_N\}}\in\nan\cup\{\infty\}$ by
\[
\bbm_{\{\caD_N\}}:=\inf\left\{
m\mid
\{\caD_N\}_N\text{ satisfies Property~(I,$m$)} 
\right\}.
\]
\end{defn}
Let $\{\caD_N\}$ be a sequence of nonzero spaces satisfying Property~(I,$m$).
Let $Q_m$ be the orthogonal projection onto the orthogonal complement of $\caD_m$ in $\otimes_{i=0}^{m-1}\cc^n$,
and consider the interaction $\Phi_{Q_m}$ associated with $Q_m$.
Then, by (\ref{intersection property}), we see that $\ker \lmk H_{\Phi_{Q_m}}\rmk _{[0,N-1]}=\caD_N$
for all $N\ge m$.
Namely, the ground state spaces
of the Hamiltonian $H_{\Phi_{Q_m}}$
are given by $\{\caD_N\}$.
We shall refer to that particular Hamiltonian as 
the \emph{parent Hamiltonian}
of $\{\caD_N\}$,
and denote this $\Phi_{Q_m}$ by $\Phi_{m,\{\caD_N\}}$.

The matrix product formalism gives a way to define sequences of subspaces.
Let $k\in \nan$, and 
$\vv=(v_1,\cdots,v_n)\in \lmk \mk\rmk^{\times n}$, an
$n$-tuple of elements in $\mk$.
For $l\in\nan$ and $\mu^{(l)}=(\mu_0,\mu_1,\ldots,\mu_{l-1})\in\{1,\cdots,n\}^{\times l}$,
we use the notation
\begin{align}
\widehat{v_{\mu^{(l)}}}:=
v_{\mu_0}v_{\mu_1}\cdots v_{\mu_{l-1}}\in\mk,\quad
\widehat{\psi_{\mu^{(l)}}}:=
\bigotimes_{i=0}^{l-1}\psi_{\mu_i}\in\bigotimes_{i=0}^{l-1}\cc^n.
\end{align}
(Recall that $\{\psi_{\mu}\}_{\mu}$ is the fixed CONS of $\cc^n$.)For an invertible $R\in \mk$ and $\vv=(v_1,\cdots,v_n)\in \lmk \mk\rmk^{\times n}$,
we denote by $R\vv R^{-1}$ the $n$-tuple given by
$R\vv R^{-1}=(Rv_1R^{-1},\cdots,Rv_nR^{-1})$.
We say $R \vv R^{-1}$ is similar to $\vv$ in this case.
\begin{defn}
Let $k\in \nan$, and 
$\vv=(v_1,\cdots,v_n)\in \mk^{\times n}$.
For each $l\in\nan$,
define
$\Gamma^{(R)}_{l,\vv}\;:\; \mk\to\bigotimes_{i=0}^{l-1}\cc^n$
by
\begin{align}\label{eq:gbdef}
\Gamma^{(R)}_{l,\vv}\lmk X\rmk
=\sum_{\mu^{(l)}\in \{1,\cdots,n\}^{\times l}}\lmk\Tr X \lmk\widehat{v_{\mu^{(l)}}} \rmk^*\rmk\widehat{\psi_{\mu^{(l)}}},\quad
X\in\mk,
\end{align}
and set
$
\cgv{l}:=\Ran\Gamma^{(R)}_{l,\vv}\subset \bigotimes_{i=0}^{l-1}\cc^n.
$
Furthermore, we denote by
$G_{l,\vv}$
the orthogonal projection onto $\cgv{l}$ in $\bigotimes_{i=0}^{l-1}\cc^n$.
We set $h_{m,\vv}:=1-G_{m,\vv}$ and $\Phi_{m,\bb}:=\Phi_{h_{m,\bb}}$.
For the simplicity of the terminology, 
we use the symbol 
$\bbm_{\vv}$
to denote $\bbm_{\{\cgv{N}\}_N}$.
\end{defn}

With a random choice of $\vv$, the sequence of subspaces $\{\cgv{l}\}_l$ would not satisfy 
the intersection property. Furthermore, even if the intersection property is satisfied,
the parent Hamiltonian may not be gapped.
We need to require some additional conditions to guarantee those properties.

We introduce several notations about $n$-tuples of matrices, which we use throughout this paper. For $k,l\in\nan$ and $\vv\in \mk^{\times n}$,
let ${\mathcal K}_l(\vv)$ be the following span of monomials of degree $l$ in the $v_\mu$'s,
\begin{equation}\label{KB}
{\mathcal K}_l(\vv) :=\spn\left\{v_{\mu_0}v_{\mu_{1}}\ldots v_{\mu_{l-1}}\mid
(\mu_0,\mu_1,\ldots,\mu_{l-1})\subset\{1,\ldots n\}^{\times l}\right\}.
\end{equation}
For $k\in\nan$ and $\vv\in \mk^{\times n}$,
we define the completely positive map $T_{\vv}\;:\; \mk\to\mk$ by
\begin{align}
T_{\vv}(X):=\sum_{\mu=1}^n v_{\mu}Xv_{\mu}^*,\quad X\in\mk.
\end{align}
The spectral property of $T_{\vv}$ will play an important role in the analysis.
For $\vv\in \Mat_{k}^{\times n}$, 
and a projection $p\in\mk$,
we denote by $\vv_p$ the $n$-tuple given by
$\vv_p=(pv_{1}p,pv_{2}p,\ldots,pv_{n}p)$.
For $p, q\in\caP(\mk)$, we also set 
\begin{align}\label{eq:mvpq}
M_{\vv,p,q}:=\inf\{M\in\nan\mid \gv{N} \;\text{is injective on} \; p\mk q  \; \text{for all } N\ge M\}.
\end{align}
We set $\cgv{N}^{p,q}:=\gv{N}(p\mk q)$ for $N\in\nan$.

\subsection{$\ClassA$}
In this section we introduce a class of $n$-tuples of matrices  which we consider in this paper.
First we introduce several notations we use repeatedly.
\begin{defn}\label{def:class}
For $k_R,k_L\in\nan\cup\{0\}$,
we denote by $\Wo(k_R,k_L)$ the set of all
${\boldsymbol\lambda}=(\lambda_{-k_R},\ldots,\lambda_{-1},\lambda_0,\lambda_1,\ldots,\lambda_{k_L})\in
\cc^{k_R+k_L+1}$
satisfying
\begin{align*}
&\lambda_0=1,\\
&0<\lv\lambda_{-k_R}\rv\le\lv\lambda_{-k_R+1}\rv\le\cdots
\le\lv\lambda_{-1}\rv<1,\\
&0<\lv\lambda_{k_L}\rv\le\lv\lambda_{k_L-1}\rv\le\cdots
\le\lv\lambda_{1}\rv<1.
\end{align*}
\end{defn}
For $\lal\in \cc^{k_R+k_L+1}$, we define
the diagonal matrix $\Lambda_{\lal}:=\sum_{i=-k_R}^{k_L}\lambda_i E^{(k_R,k_L)}_{ii}$.
\begin{defn}
Let $n,n_0\in\nan$ with $n\ge 2$,
and $\oo=(\omega_1,\ldots,\omega_n)\in \mnz^{\times n}$.
We say $\oo$ is primitive if 
\[
\l_{\oo}:=\inf\{l\in\nan\;:\; \kl{l'}(\oo)=\mnz \text{ for all } l'\ge l\}<\infty.
\]
We denote by $\Primz(n,n_0)$ the set of 
all primitive $n$-tuples $\oo$ of $n_0\times n_0$ matrices.
We also denote by $\Prim(n,n_0)$ the set of 
all primitive $n$-tuples $\oo$ of $n_0\times n_0$ matrices 
with $r_{T_{\oo}}=1$.
Furthermore, we denote by $\Primz_u(n,n_0)$ the set of 
all primitive $n$-tuples $\oo$ of $n_0\times n_0$ matrices
such that $T_\oo$ is a unital CP map.
\end{defn}
\begin{defn}\label{def:dg}
Let $k_R\in\nan$.
We define $\caC^R(k_R)$
by the set of $k_R$-tuples $\bbD=(D_1,\ldots,D_{k_R})$ of $\UT_{0,k_R+1}$
satisfying the following conditions.
\begin{enumerate}
\item
$D_aE_{00}^{(k_R,0)}=E_{-a,0}^{(k_R,0)}$
\item
The linear span of 
$\{D_a\}_{a=1}^{k_R}$ is a subalgebra of $\UT_{0,k_R+1}$
\end{enumerate}
Similarly, for  $k_L\in\nan$, we define $\caC^L(k_L)$
by the set of $k_L$-tuples $\bbG=(G_1,\ldots,G_{k_L})$ of $\UT_{0,k_L+1}$
satisfying the following conditions.:\begin{enumerate}
\item
$E_{00}^{(0,k_L)}G_b=E_{0,b}^{(0,k_L)}$
\item
The linear span of 
$\{G_b\}_{b=1}^{k_L}$ is a subalgebra of $\UT_{0,k_L+1}$.
\end{enumerate}
\end{defn}
\begin{defn}[$\caT(k_R,k_L)$]
Let $ (k_R,k_L)\in \lmk \nan\cup\{0\}\rmk^{\times 2}$.
If $k_R,k_L\in\nan$,
we denote by $\caT(k_R,k_L)$
the set of all quadruplets
$(\lal,\bbD,\bbG,Y)$
with $\lal\in\Wo(k_R,k_L)$, $\bbD\in \caC^R(k_R)$,
$\bbG\in \caC^L(k_L)$, and $Y\in\UT_{0,k_L+k_R+1}$ such that
for any $1\le a\le k_R$ and $1\le b\le k_L$,
\begin{align}
&{\Lambda_{\lal}}I_R^{(k_R,k_L)}(D_{a})=\lambda_{-a}I_R^{(k_R,k_L)}(D_{a}){\Lambda_{\lal}},\quad
a=1,\ldots,k_R,\label{eq:ld}\\
&{\Lambda_{\lal}} I_L^{(k_R,k_L)}(G_b)=\lambda_{b}^{-1}I_L^{(k_R,k_R)}(G_{b}){\Lambda_{\lal}},\quad
b=1,\ldots,k_L\label{eq:lg},\\
&Y\Lambda_{\lal}=\Lambda_{\lal}Y,\quad
\pu Y\pd=0,\label{eq:ly}
\end{align}
and
\begin{align}
&\lmk\Lambda_{\lal}(1+Y)\rmk^lI_R^{(k_R,k_L)}(D_a)
=\sum_{a'=1}^{k_R}\braket{f_{-a'}^{(k_R,k_L)}}{\lmk \Lambda_{\lal}(1+Y)\rmk^lf_{-a}^{(k_R,k_L)}}
I_R^{(k_R,k_L)}(D_{a'})\lmk\Lambda_{\lal}(1+Y)\rmk^l,\label{eq:dygy1}\\
&I_L^{(k_R,k_L)}(G_b)\lmk\Lambda_{\lal}(1+Y)\rmk^l
=\sum_{b'=1}^{k_L}\braket{f_{b}^{(k_R,k_L)}}{\lmk \Lambda_{\lal}(1+Y)\rmk^lf_{b'}^{(k_R,k_L)}}
\lmk\Lambda_{\lal}(1+Y)\rmk^lI_L^{(k_R,k_L)}(G_{b'}),\label{eq:dygy2}
\end{align}
for all $l\in\nan$.If $k_R=0$ and $k_L\in\nan$, $\caT(0,k_L)$ is the set of all triples $(\lal,\bbG,Y)$
with $\lal\in\Wo(k_R,k_L)$, 
$\bbG\in \caC^L(k_L)$, and $Y\in\UT_{0,k_L+k_R+1}$, satisfying (\ref{eq:lg}), (\ref{eq:ly}), and (\ref{eq:dygy2}).
If $k_R\in\nan$ and $k_L=0$, $\caT(k_R,0)$ is the set of all triples $(\lal,\bbD,Y)$
with  $\lal\in\Wo(k_R,k_L)$, $\bbD\in \caC^R(k_R)$, and $Y\in\UT_{0,k_L+k_R+1}$, satisfying (\ref{eq:ld}), (\ref{eq:ly}), and (\ref{eq:dygy1}).
If $k_R=k_L=0$, $\caT(0,0)$ consists of a single point $\lambda_0=1\in\Wo(0,0)$.
\end{defn}
\begin{rem}
As $\Lambda_{\lal}(1+Y)$ is an upper triangular matrix,
the summations on the right hand side of (\ref{eq:dygy1}) and (\ref{eq:dygy2}) are actually over
$a\le a'$ and $b\le b'$.
\end{rem}
\begin{rem}\label{rm:dgd}
For the simplicity of the statements, we write $\bbD$ or $\bbG$ even when $k_R=0$ or $k_L=0$,
although they are not defined in these cases.
The readers should discard them suitably in each statements.
\end{rem}
\begin{rem}\label{rem:dkk}
Let $\bbD\in\caC^R(k_R)$ and $\bbG\in\caC^L(k_L)$.
Note that $I_R^{(k_R,k_L)}(D_a)=\overline{P_L^{(k_R,k_L)}}I_R^{(k_R,k_L)}(D_a)P_R^{(k_R,k_L)}$,
 $I_L^{(k_R,k_L)}(G_b)={P_L^{(k_R,k_L)}}I_L^{(k_R,k_L)}(G_b)\overline{P_R^{(k_R,k_L)}}$, because
$D_a$ and $G_b$ are upper triangular matrices with zero diagonal elements.
Therefore, the linear space
\begin{align}
\caD(k_R,k_L,\bbD,\bbG):=\spa\lmk \unit, \{I_R^{(k_R,k_L)}(D_a)\}_{a=1}^{k_R} \cup\{I_L^{(k_R,k_L)}(G_b)\}_{b=1}^{k_L}\cup
\left\{E_{-a,b}^{(k_R,k_L)}\right\}_{a=1,\ldots,k_R,b=1,\ldots,k_L}\rmk
\end{align}
is a subalgebra of $\UT_{k_L+k_R+1}$.
(Here, as mentioned in Remark \ref{rm:dgd},  discard 
$ \{I_R^{(k_R,k_L)}(D_a)\}_{a=1}^{k_R}$, 
$\left\{E_{-a,b}^{(k_R,k_L)}\right\}_{a=1,\ldots,k_R,b=1,\ldots,k_L}$,
on the right hand side if $k_R=0$.)
\end{rem}

\begin{rem}\label{rem:din}
By (\ref{eq:dygy1}), (\ref{eq:dygy2}) and the fact that $\Lambda_{\lal}(1+Y)$ is an invertible
matrix, we observe that
\begin{align}
&I_R^{(k_R,k_L)}(D_a)\lmk\Lambda_{\lal}(1+Y)\rmk^l
=\sum_{a'=1}^{k_R}\braket{f_{-a'}^{(k_R,k_L)}}{\lmk \Lambda_{\lal}(1+Y)\rmk^{-l}f_{-a}^{(k_R,k_L)}}
\lmk\Lambda_{\lal}(1+Y)\rmk^{l}I_R^{(k_R,k_L)}(D_{a'}),\nonumber\\
&\lmk\Lambda_{\lal}(1+Y)\rmk^lI_L^{(k_R,k_L)}(G_b)
=\sum_{b'=1}^{k_L}\braket{f_{b}^{(k_R,k_L)}}{\lmk \Lambda_{\lal}(1+Y)\rmk^{-l}f_{b'}^{(k_R,k_L)}}
I_L^{(k_R,k_L)}(G_{b'})\lmk\Lambda_{\lal}(1+Y)\rmk^l.
\end{align}
This implies
\begin{align}\label{eq:adl}
\lmk \Lambda_{\lal}\lmk 1+Y\rmk\rmk^{-x}
\caD(k_R,k_L,\bbD,\bbG)
\lmk \Lambda_{\lal}\lmk 1+Y\rmk\rmk^x
=\caD(k_R,k_L,\bbD,\bbG),\quad
x\in\bbZ.
\end{align}
\end{rem}
\begin{defn}
Let $n,n_0\in\nan$ with $n\ge 2$, $k_R,k_L\in \nan\cup\{0\}$ and
$(\lal,\bbD,\bbG,Y)\in\caT(k_R,k_L)$.
We denote by ${\mathfrak B}(n,n_0,k_R,k_L,\lal,\bbD,\bbG,Y)$
the set of all $n$-tuples $\bb=(B_1,\ldots,B_n)\in\mnz\otimes\lmk
 \caD(k_R,k_L,\bbD,\bbG)\Lambda_{\lal}\lmk 1+Y\rmk
\rmk$
satisfying 
\begin{align}\label{eq:lblb}
l_\bb=
l_{\bb}(n,n_0,k_R,k_L,\lal,\bbD,\bbG,Y)
:=\inf\left\{l\mid
\caK_{l'}(\bbB)=
\mnz\otimes\lmk
 \caD(k_R,k_L,\bbD,\bbG)\lmk \Lambda_{\lal}\lmk 1+Y\rmk\rmk^{l'}
\rmk
\text{ for all } l'\ge l
\right\}
<\infty,
\end{align}
and $r_{T_{\bb}}=1$.
\end{defn}
\begin{defn}[$\ClassA$]
Let $n\ge 2$.
We define 
\[
\ClassA:=\bigcup
\left\{ {\mathfrak B}(n,n_0,k_R,k_L,\lal,\bbD,\bbG,Y)
\mid n_0\in\nan, k_R,k_L\in \nan\cup\{0\}, (\lal,\bbD,\bbG,Y)\in\caT(k_R,k_L)
\right\}.
\]
\end{defn}
\begin{rem}\label{rem:inj}
We say $\bb$ belongs to $\ClassA$
with respect to $(n_0,k_R,k_L,\lal, \bbD,\bbG,Y)$,
when we would like to make them explicitly.
\end{rem}
\begin{rem}
Recall that the parents Hamiltonians considered in \cite{Fannes:1992vq}
are generated by $\vv\in \Primz(n,n_0)$,
i.e., the condition $\caK_{l}(\vv)=\mnz$ for large $l$, is required.
This is what we call the injectivity condition.
This corresponds to the case $k_L=k_R=0$
in our setting.
\end{rem}
\begin{ex}
Let $n_0\in\nan$, and $k_R,k_L\in \nan\cup\{0\}$.
We fix $0<\kappa<1$, and set $\lal=(\lambda_i)_{i=-k_R}^{k_L}$ and  $\rar=(r_\alpha)_{\alpha=1}^{n_0}$ by
\begin{align*}
&r_{\alpha}:=\kappa ^{\alpha-1},\quad
\text{for }\alpha=1,\ldots,n_0,\notag\\
&\lambda_{j}:=\kappa ^{|j|n_0},\quad
\text{for }j=-k_R,\ldots, -1,0,1,\ldots,k_L.
\end{align*}
We also set
\begin{align*}
V_R:=\sum_{j=-k_R}^{-1}E_{j,j+1}^{(k_R,0)},\quad 
V_L:=\sum_{j=0}^{k_L-1}E_{j,j+1}^{(0,k_L)},
\end{align*}
and define $\bbD$ and $\bbG$ by
\begin{align*}
D_a:=V_R^a,\quad a=1,\ldots,k_R,\quad
G_b:=V_L^b,\quad b=1,\ldots, k_L.
\end{align*}
It is easy to check $(\lal, \bbD,\bbG,0)\in \caT(k_R,k_L)$.
Set
\begin{align*}
\eta:=\sum_{\alpha=2}^{(n_0)}\chi^{(n_0)}_\alpha.
\end{align*}
We define $\bbB=(B_\mu)$ by
\begin{align*}
&B_1:=\Lambda_{\rar}\otimes \Lambda_\lal,\\
&B_2:= \lmk\ket{\chi_1^{(n_0)}}\bra{\eta}+\ket{\eta}\bra{\chi_1^{(n_0)}}\rmk\otimes \Lambda_\lal
+\Lambda_\rar\otimes \lmk I_R^{(k_R,k_L)}(V_R)+I_L^{(k_R,k_L)}(V_L)\rmk\Lambda_\lal,\\
&B_\mu:=0,\quad \mu\ge 3.
\end{align*}
We claim  $\bbB\in {\mathfrak B}(n,n_0,k_R,k_L,\lal,\bbD,\bbG,0)$.
From the definition, we have $B_\mu\in\mnz\otimes\lmk
 \caD(k_R,k_L,\bbD,\bbG)\Lambda_{\lal}
\rmk$, for each $\mu=1,\ldots, n$.
This fact and Remark \ref{rem:dkk}, Remark \ref{rem:din} imply that
$\caK_l(\bbB)\subset \mnz\otimes\lmk
 \caD(k_R,k_L,\bbD,\bbG)\Lambda_{\lal}^l
\rmk$ holds for all $l\in\nan$.
Now we prove the opposite inclusion for $l$ large enough.
We use the argument in Lemma 5.2 \cite{bo}.
Let $l\in\nan$ and $0\le j\le l-1$. By the expansion, we see that
\begin{align*}
B_1^{l-1-j}B_2 B_1^j
=\sum_{\alpha=2}^{n_0}\lmk r_\alpha^j\zeij{1\alpha}+r_\alpha^{l-1-j}\zeij{\alpha1}\rmk\otimes \Lambda_\lal^{l}
+\kappa^{(j-l+1)n_0}\Lambda_\rar^{l}\otimes I_L^{(k_R,k_L)}(V_L)\Lambda^{l}_\lal
+\kappa^{(l-1-j)n_0}\Lambda_\rar^{l}\otimes I_R^{(k_R,k_L)}(V_R)\Lambda^{l}_\lal.
\end{align*}
As the numbers
$\{r_\alpha=\kappa^{\alpha-1}\}_{\alpha=2}^{n_0}$, $\{r_\alpha^{-1}=\kappa^{-\alpha+1}\}_{\alpha=2}^{n_0}$,
$\{\kappa^{n_0}\}$ $\{\kappa^{-n_0}\}$
are distinct, by the argument of Lemma 5.2 \cite{bo}, we conclude that for $l$ large enough, all of
\begin{align*}
\zeij{1\alpha}\otimes \Lambda_\lal^{l}, \quad\zeij{\alpha1}\otimes \Lambda_\lal^{l},\quad
\Lambda_\rar^l\otimes I_L^{(k_R,k_L)}(V_L)\Lambda^{l}_\lal,\quad
\Lambda_\rar^l\otimes  I_R^{(k_R,k_L)}(V_R)\Lambda^{l}_\lal
\end{align*}
belong to $\caK_l(\bbB)$.
Multiplying these terms each other, we conclude that
$\mnz\otimes\lmk
 \caD(k_R,k_L,\bbD,\bbG)\Lambda_{\lal}^l\rmk\subset \caK_l(\bbB)$ holds for all $l$ large enough.
\end{ex}

\subsection{The main result}
In this part I, we study the ground state structure of the Hamiltonians given by $\bbB\in\ClassA$.
Our Hamiltonian is the Hamiltonian $H_{\Phi_{m,\bb}}$ given by 
the subspaces $\cgb{m}$ via the formula
(\ref{hamdef}) with $h=1-G_{m,\bb}$ for $\bb\in\ClassA$.
\begin{thm}\label{thm:asymmetric}
Let $n\in\nan$ with $n\ge 2$, and $\bb\in\ClassA$
with respect to $(n_0,k_R,k_L,\lal, \bbD,\bbG,Y)$.
Then $\bbm_\bb\le 2 l_\bb(n,n_0,k_R,k_L,\lal,\bbD,\bbG,Y)<\infty$ and for any 
$m\ge \max\{2l_{\bb}(n,n_0,k_R,k_L,\lal,\bbD,\bbG,Y),\frac{\log\lmk n_0^2(k_L+1)(k_R+1)+1\rmk}{\log n}\}$,
we have the followings.
\begin{enumerate}
\item[(i)]$\ker\lmk H_{\Phi_{m,\bb}}\rmk_{[0,N-1]}=\cgb{N}$
and $\dim \ker\lmk H_{\Phi_{m,\bb}}\rmk_{[0,N-1]}=n_0^2(k_L+1)(k_R+1)$, for $N\ge m$.
\item[(ii)]  There exist $\gamma_m>0$ and $N_m\in\nan$ such that
\[
\gamma_m\lmk 1-G_{N,\bb}\rmk
\le \lmk H_{\Phi_{m,\bb}}\rmk_{[0,N-1]},\text{ for all } N\ge N_m.
\] 
\item[(iii)]$\caS_{\bbZ}(H_{\Phi_{m,\bb}})$ consists of  a unique state $\omega_{\bb,\infty}$ on $\caA_{\bbZ}$.
\item[(iv)]There exist $0<C_{\bb}$, $0<s_{\bb}<1$, $N_{\bb}\in\nan$, and states
$\omega_{R,\bb}\in \caS_{[0,\infty)}(H_{\Phi_{m,\bb}})$, and 
$\omega_{L,\bb}\in \caS_{(-\infty,-1]}(H_{\Phi_{m,\bb}})$, such that
\begin{align}
&\lv
\frac{\Tr_{[0,N-1]}\lmk G_{N,\bb} A\rmk}{\Tr_{[0,N-1]}\lmk G_{N,\bb}\rmk}-
\omega_{R,\bb}(A)
\rv\le C_{\bb} s_{\bb}^{N-l}\lV A\rV,\notag\\
&\lv
\frac{\Tr_{[0,N-1]}\lmk G_{N,\bb}\tau_{N-l}\lmk A\rmk\rmk}{\Tr_{[0,N-1]}\lmk G_{N,\bb}\rmk}-
\omega_{L,\bb}\circ\tau_{-l}(A)
\rv\le C_{\bb} s_{\bb}^{N-l}\lV A\rV,
\end{align}
for all $l\in\nan$, $A\in\caA_{[0,l-1]}$, and $N\ge \max\{l, N_{\bb}\}$,
and 
\begin{align}
\inf\left\{ \sigma\lmk\omega_{R,\bb}\vert_{\caA_{[0,l-1]}}\rmk\setminus \{0\}\mid l\in\nan\right\}>0,\notag\\
\inf\left\{ \sigma\lmk\omega_{L,\bb}\vert_{\caA_{[-l,-1]}}\rmk\setminus \{0\}\mid l\in\nan\right\}>0.
\end{align}
\item[(v)]For any $\psi\in\caS_{[0,\infty)}(H_{\Phi_{m,\bb}})$ (resp. 
$\psi\in\caS_{(-\infty,-1]}(H_{\Phi_{m,\bb}})$), there exists an  $l_{\psi}\in \nan$
such that $\lV\psi-\psi\circ\tau_{l_\psi}\rV<2$
(resp.  $\lV\psi-\psi\circ\tau_{-l_\psi}\rV<2$).
\item[(vi)] 
$\caS_{(-\infty,-1]}(H_{\Phi_{m,\bb}})$ and 
$\caS_{[0,+\infty)}(H_{\Phi_{m,\bb}})$ are convex sets and
there exist affine bijections 
\[
\Xi_L:{\mathfrak E}_{n_0(k_L+1)}\to \caS_{(-\infty,-1]}(H_{\Phi_{m,\bb}}),\qquad
\Xi_R:{\mathfrak E}_{n_0(k_R+1)}\to \caS_{[0,+\infty)}(H_{\Phi_{m,\bb}}).
\]
\item[(vii)]
There exist $C_{\bb}'>0$, and $0<s_{\bb}'<1$ such that
\begin{align}
&\lv
\psi\circ\tau_{N}(A)-\omega_{\bb,\infty}(A)
\rv\le C_{\bb}'{s_{\bb}'}^N\lV A\rV,\notag\\
&(resp. \lv
\psi\circ\tau_{-N}(A)-\omega_{\bb,\infty}(A)
\rv\le C_{\bb}'{s_{\bb}'}^N\lV A\rV,)
\end{align}
for all $A\in \caA_{[0,\infty)}$ and $\psi\in\caS_{[0,\infty)}(H_{\Phi_{m,\bb}})$
(resp.
for all $A\in \caA_{(-\infty,-1]}$ and
$\psi\in\caS_{(-\infty,-1]}(H_{\Phi_{m,\bb}})$), and $N\in\nan$.
\item[(viii)] Any element in $\caS_{(-\infty,-1]}(H_{\Phi_{m,\bb}})$ or $\caS_{[0,+\infty)}(H_{\Phi_{m,\bb}})$ is a factor
state.
\item[(ix)] $\omega_{\bb,\infty}$ satisfies the exponential decay of correlations. 
\end{enumerate}
\end{thm}
This article is organized as follows:
In Section \ref{general}, we introduce a sufficient condition
for an $n$-tuple 
$\vv$ to have a gapped parent Hamiltonian, in rather general setting.
Applying the results of Section \ref{general},
we study the properties of $H_{\Phi_{m,\bb}}$ for $\bb\in\ClassA$ in Section \ref{sec:gss}.
\section{The intersection property and the spectral gap}\label{general}
In this section, we introduce a sufficient condition
for the MPS Hamiltonians to be gapped.
We first introduce a set of conditions on sequences of subspaces.
\begin{defn}[\it Condition 1]
Let  $n\in\nan$.
Let 
$\caD_l$ be a nonzero subspace of 
$\bigotimes_{i=0}^{l-1}\cc^n$ given for each $l\in\nan$.
We denote by $G_l$ the orthogonal projection
 onto $\caD_l$ in $\bigotimes_{i=0}^{l-1}\cc^n$ for each $l\in\nan$.
 We say the sequence of subspaces $\{\caD_l\}_{l\in\nan}$
 satisfies the {\it Condition 1} if  
 there exists $m_0,l_0\in \nan$ such that
\begin{enumerate}
\item[(i)] 
$\{{\caD}_l\}_{l\in\nan}$ satisfies Property~(I,$m_0$), and
 \item[(ii)]
for any $l_0\le l$, there exists $0<\varepsilon_l<\frac{1}{\sqrt l}$
such that
\[
\lV\lmk 1_{[0,N-l]}\otimes G_l\rmk\lmk G_N\otimes 1_{\{N\}}-G_{N+1}\rmk\rV
<\varepsilon_l,
\]
for all $N\ge 2l$.
\end{enumerate}
 We say $\{\caD_l\}_{l\in\nan}$
satisfies the {\it Condition 1} for $(m_0,l_0)$
when we would like to specify the numbers.
\end{defn}

The following theorem is a special version of Theorem 3 in ~\cite{Nachtergaele:1996vc}.
\begin{thm}[\cite{Nachtergaele:1996vc}]\label{bruno}Let  $\{{\caD}_l\}_{l\in\nan}$  be a sequence of nonzero vector spaces such that 
$\caD_l\subset \bigotimes_{i=0}^{l-1}\cc^n$, and $G_l$ the orthogonal projections onto $\caD_l$
in $\otimes_{i=0}^{l-1}\cc^n$. 
For $m,N\in\nan$ with $m\le N$, we set
\[
H_N^m:=\sum_{0\le x\le N-m}\tau_x(1-G_m).
\]
Suppose that $\{\caD_l\}_l$ satisfies the {\it Condition 1} for $(m_0,l_0)$.
Then for all $m_0\le m$, 
\begin{enumerate}
\item[(1)]$\ker H_N^m=\caD_N$, for all $N\ge m$, and
\item[(2)]
\[
\frac{\gamma_{l,m}}{l+2}\lmk 1-\varepsilon_l\sqrt l\rmk^2\lmk 1-G_N\rmk
\le
H_N^m,
\]
for all $l$ with $\max\{l_0,m\}\le l$, 
and $N$ with $2l+1\le N$,
where
\[
\gamma_{l,m}=\min\left\{\drr\lmk\sigma(H_l^m)\setminus \{0\},0\rmk,\drr\lmk\sigma(H_{2l}^m)\setminus \{0\},0\rmk\right\}.
\]
\end{enumerate}
\end{thm}

In this section, we give a criterion for this in the MPS formalism.
We introduce a set of conditions.
Recall the definition (\ref{KB})
of $\caK_l(\vv)$.
\begin{defn}[\it Condtion 2]
For $n,k\in\nan$, projections $p,q\in\caP(\mk)$, and  $n$-tuple of
$k\times k$ matrices $\vv\in\mk^{\times n}$,
we say that the pentad $(n,k,p,q,\vv)$ satisfies 
{\it Condition 2} if the following conditions are satisfied.
\begin{description}
\item[(i)]
$pq=qp\neq 0$.
\item[(ii)]
$v_{\mu}p=pv_{\mu}p$ for all $\mu=1,\ldots,n$.
\item[(iii)]
$qv_{\mu}=qv_{\mu}q$ for all $\mu=1,\ldots,n$.
\item[(iv)]
There exist constants $c_{pq}>0$, $0<s_{pq}<1$, 
a positive linear functional $\varphi_{pq}$ on $\mk$
and a positive element $e_{pq}\in \mk$ with $s(\varphi_{pq})=s(e_{pq})=pq$,
such that
\[
\lV T^N_{\vv_{pq}}\lmk A\rmk-\varphi_{pq}(A)e_{pq}\rV\le
c_{pq}s_{pq}^N\lV A\rV,\quad
\text{for all} \; A\in \mk,\text{and}\;N\in\nan.
\]
\item[(v)]
There exist constants $c_{\bar q}>0$, $0<s_{\bar q}<1$
such that
\[
\lV T_{\vv_{\bar q}}^N(A)\rV\le c_{\bar q}s_{\bar q}^N\lV A\rV,\quad
\text{for all }\; A\in \mk \;\text{and}\;
N\in\nan.
\]
\item[(vi)]
There exist constants $c_{\bar p}>0$, $0<s_{\bar p}<1$
such that
\[
\lV T_{\vv_{\bar p}}^N(A)\rV\le c_{\bar p}s_{\bar p}^N\lV A\rV,\quad
\text{for all }\; A\in \mk \;\text{and}\;
N\in\nan.
\]
\item[(vii)]
For any $\eta\in q\cc^k$ with $pq\eta\neq 0$,
there exists $l_{\eta}\in\nan$ such that $\lmk \caK_{l_{\eta}}(\vv_{q})\rmk ^*\eta=q\cc^k$.
\item[(viii)]
For any $\xi\in p\cc^k$ with $pq\xi\neq 0$,
there exists $l_{\xi}'\in\nan$ such that 
$\caK_{l_{\xi}'}(\vv_{p})\xi=p\cc^k$.
\end{description}
\end{defn}
\begin{defn}[\it Condition 3]
Let $n,k\in\nan$, $p,q\in\caP(\mk)$ and $\vv\in\mk^{\times n}$.
We  say that the pentad $(n,k,p,q,\vv)$ satisfies
the Condition 3 for $m_1\in\nan$ if 
for all $N\ge m_1$,
$\dim\caK_{N}(\vv)=(\rank p)(\rank q)$.
\end{defn}
\begin{defn}[\it Condition 4]
Let $n,k\in\nan$,  and $\vv\in\mk^{\times n}$.
For $m_2,m_3\in\nan$, we  say that the triple $(n,k,\vv)$ satisfies
the Condition 4 for $(m_2,m_3)$, if 
there exists an invertible $X_{m_2}\in\caK_{m_2}(\vv)$
such that $X_{m_2}^{-1}\caK_{N+m_2}(\vv)\subset \caK_{N}(\vv)$
for all $N\ge m_3$.
When we would like to specify $X_{m_2}$,
we say the triple $(n,k,\vv)$ satisfies
the Condition 4 for $(m_2,m_3)$ 
with respect to $X_{m_2}$.
 \end{defn}
Here is the main Proposition of this section.
\begin{prop}\label{prop:maingen}
Let $n,k,m_1,m_2,m_3\in\nan$, $p,q\in\caP(\mk)$, and 
 $\vv=(v_{\mu})_{\mu=1}^n\in \mk^{\times n}$.
Suppose that
the pentad $(n,k,p,q,\vv)$ satisfies the {\it Condition 2}, and
the {\it  Condition 3} for $m_1$.
Furthermore, assume that the triple $(n,k,\vv)$
satisfies the {\it Condition 4} for $(m_2,m_3)$.
Then $
M_{\vv,p,q}<\infty$
and $\{\caG_{l,\vv}\}_{l\in\nan}$
satisfies 
the {\it Condition 1}.
Here, 
$m_0\in\nan$ of  {\it Condition 1}
can be taken $m_0=m_2+m_3$.
\end{prop}
 As we mentioned in the introduction, the spectral property of $T_{\vv}$ is important for the analysis of
 the spectral gap. In Subsection \ref{cpgen}, we consider the spectral property of
 $T_{\vv}$ when $(n,k,p,q,\vv)$ satisfies {\it Condition 2}.
The intersection property
is considered in subsection \ref{intergen}.
The {\it Condition 2} and {\it Condition 3}
imply the bijectivity of $\gv{N}\vert_{p\mk q}$,
which is proven in subsection \ref{bijecgen}.
The important input to show the spectral gap is an estimate of
overlaps of spectral projections. This is done in
subsection \ref{gapgen}.


\subsection{Spectral analysis of CP maps}\label{cpgen}
\begin{defn}
Let $T$ be a CP map on $\mk$.
Let $0<s<1$,
$\varphi$ a state on $\mk$, and $e\in\mk_+$.
We say $T$ satisfies the Spectral Property II
with respect to $(s,e,\varphi)$ if 
\begin{description}
\item[(1)]
$r_{T}=1$ and $1$ is a non-degenerate eigenvalue of $T$,
\item[(2)]
$\sigma(T)\setminus \{1\}\subset \caB_{s}(0)$,
\item[(3)]
$e$  is a $T$-invariant positive  element and $e=P_{\{1\}}^{T}(1 )$,
\item[(4)]
$\varphi$ is $T$-invariant,
\item[(5)]
$P_{\{1\}}^{T}(\cdot )=\varphi(\cdot)e$,
\item[(6)]
for any $s\le s'<1$, we have
\[
\lV T^N(A)-\varphi(A)e\rV
\le
(s')^{N+1}\sup_{|z|=s'}\lV(z-T)^{-1}\rV\lV A\rV,\;\;\text{for all}\;\;
N\in \nan ,A\in\mk.
\]
\end{description}
\end{defn}
\begin{rem}
The conditions above are redundant.
For example,
(6) follows from (1)-(5). 
However, we leave them for the convenience.
\end{rem}
In this subsection we prove the following Lemma.
\begin{lem}\label{cpmain}
Let $n,k\in\nan$, $p,q\in\caP(\mk)$, and  $\vv\in\mk^{\times n}$.
Assume that the pentad $(n,k,p,q,\vv)$ satisfies 
{\it Condition 2}.
Then,
there exist a constant $0<s_{\vv}<1$, a state $\varphi_{\vv}$, 
and a positive element $e_{\vv}\in\mk_{+}$,
such that 
\begin{description}
\item[(1)]$T_{\vv}$ satisfies the Spectral Property II
with respect to $(s_{\vv},e_{\vv},\varphi_{\vv})$, and
\item[(2)]
$s(e_{\vv})=p,\quad s(\varphi_{_\vv})=q$.
\end{description}
\end{lem}
\begin{rem}\label{tavv}
We call $(s_{\vv}, e_\vv,\varphi_{\vv})$ {\it the triple associated with $(n,k,p,q,\vv)$}.
From now on, if we write $(s_{\vv}, e_\vv,\varphi_{\vv})$,
it means this triple.
Furthermore, $\rho_{\vv}$ denotes the density matrix of $\varphi_\vv$.
\end{rem}
As $s(e_{\vv})=p$ and $s(\varphi_{\vv})=q$,
there exist $x_{\vv}\in (p \mk p)_{+}$ and 
$y_{\vv}\in (q \mk q)_{+}$
such that
$e_{\vv}x_{\vv}=x_{\vv}e_{\vv}=p$ and $\rho_{\vv}y_{\vv}=y_{\vv}\rho_{\vv}=q$.
We set $a_{\vv}:=\lV x_\vv\rV^{-1}$ and 
$c_{\vv}:=\lV y_\vv\rV^{-1}$.

The following numbers will be used.
\begin{align}\label{eq:efdef}
&\tilde E_{\vv}(N):=k^2\lV T_{\vv}^N\lmk 1- P^{T_{\vv}}_{\{1\}}\rmk\rV,\quad
E_{\vv}(N):=\lmk a_{\vv}c_{\vv}\rmk^{-1}\tilde E_{\vv}(N),\quad N\in\nan,\nonumber\\
&F_{\vv}:=\sup_{N\in\nan}\lV T_{\vv}^N\lmk 1- P^{T_{\vv}}_{\{1\}}\rmk\rV+\lV e_{\vv}\rV,\quad
L_{\vv}:=\inf \left\{L\in\nan\mid
\sup_{N\ge L}E_{\vv}(N)<\frac 12\right\}.
\end{align}
From Lemma \ref{cpmain}, we know that
$\lim_{N\to\infty}\lV T_{\vv}^N\lmk 1- P^{T_{\vv}}_{\{1\}}\rmk\rV=0$.
Hence we have $\lim_{N\to\infty} \tilde E_{\vv}(N)=
\lim_{N\to\infty} E_{\vv}(N)=0$,
$F_{\vv}<\infty$
and $L_{\vv}\in\nan$.
By the definition, we have
\begin{align}\label{eq:tnb}
\sup_{N\in \nan}\lV T_{\vv}^N\rV\le
F_{\vv}.
\end{align}
With $e_{\vv}$ and $\varphi_{\vv}$, 
we define a quasi-linear form $\braket{\cdot}{\cdot}_{\vv}$
on $\mk$ by
\[
\braket{ A}{B}_{\vv}
:=\varphi_{\vv}\lmk A^* e_{\vv}B \rmk,\quad
A,B\in\mk.
\]
Note that for any $X\in p\mk q$,
\begin{align}\label{eq:ac}
\lV X\rV_2^2\le
(a_{\vv}c_{\vv})^{-1}
\braket{X}{X}_{\vv}.
\end{align}
From this, we see that $\braket{\cdot}{\cdot}_{\vv}$
is an inner product on $p \mk q$.
For any $Z\in \mk $, we have
\begin{align}\label{eq:iv}
\braket{Z}{Z}_{\vv}^{\frac 12}
=\braket{p Z q}{pZq}_{\vv}^{\frac 12}
=\sup\left\{
\lv
\braket{pZq}{Y}_{\vv}
\rv\mid
Y\in p\mk q,\;\;
\braket{Y}{Y}_{\vv}=1
\right\}.
\end{align}
For the first equality, we used the fact that $e_\vv\in p\mk p$ and $\rho_{\vv}\in q\mk q$.

To prove Lemma \ref{cpmain},
first we note the following basic properties.
\begin{lem}\label{basiccp}
Let $n,k\in\nan$, $p\in\caP(\mk)$ and $\vv\in\mk^{\times n}$ such that
$v_{\mu}p=pv_{\mu}p$, $\mu=1,\ldots,n$.
Then,
\begin{enumerate}
\item for all $N\in\nan$, $\mu^{(N)}\in\{1,\ldots,n\}^{\times N}$ and $A\in\mk$,we have
 \begin{align*}
 p\lmk\widehat {v_{\mu^{(N)}}}\rmk^*\bar p=0,\quad
 \bar p\lmk\widehat{ v_{\mu^{(N)}}}\rmk^*\bar p=\lmk\widehat v_{\mu^{(N)}}\rmk^*\bar p,\\
 \sum_{ \mu^{(N)}\in \{1,\ldots, n \}^{\times N}}
 \bar p\lmk\widehat{ v_{\mu^{(N)}}}\rmk A
 \lmk\widehat{ v_{\mu^{(N)}}}\rmk^*\bar p
 =T_{\vv_{\bar p}}^N(A),\\
\sum_{ \mu^{(N)}\in \{1,\ldots, n \}^{\times N}}
 \lmk\widehat{ v_{\mu^{(N)}}}\rmk pAp
 \lmk\widehat{ v_{\mu^{(N)}}}\rmk^*
 =T_{\vv_{p}}^N(A),
 \end{align*}
\item for any $\eta\in\cc^k$ and $N\in\nan$, we have
\begin{align*}
\sum_{\mu^{(N)}\in \{1,\ldots,n\}^{\times N}}\lV\bar p\lmk\widehat{ v_{\mu^{(N)}}}\rmk^*p  \eta\rV^2
\le
\lmk \sum_{m=1}^N\lmk \Tr T_{\vv_{\bar p}}^{N-m}(1)\rmk^{\frac 12}
\braket{\eta}{T_{\vv_p}^{m-1}(1)\eta}^{\frac 12}\rmk^{2}\sum_{\mu=1}^n\lV v_{\mu}\rV^2
\end{align*}
\item
for any $A\in\mk$ and $N\in\nan$,
\begin{align*}
&\lV T^N_{\vv}(A)-pT^N_{\vv}(A)p\rV\\
&\le2\lV A \rV
 \lV T_{\vv_{\bar p}}^N(1)\rV^{\frac 12}
\\
&\lmk \sup_{M\in\nan}\lV T_{\vv_p}^M\rV^{\frac 12} +\sup_{M\in\nan}\lV T_{\vv_{\bar p}}^M\rV^{\frac 12}
+\lmk \sum_{m=1}^N\lmk \Tr T_{\vv_{\bar p}}^{N-m}(1)\rmk^{\frac 12}
\lV T_{\vv_p}^{m-1}(1)\rV^{\frac 12}\rmk\lmk \sum_{\mu=1}^n\lV v_{\mu}\rV^2\rmk^{\frac 12}
\rmk,
\end{align*}
\item
\begin{align*}
\sup_{M\in\nan}\lV T_{\vv}^M\rV
\le
&\sup_{M\in\nan}\lV T^M_{\vv_p}\rV+\sup_{M\in\nan}\lV \ T^M_{\vv_{\bar p}}\rV
+2\sup_{M\in\nan}\lmk \sum_{m=1}^M\lmk \Tr T_{\vv_{\bar p}}^{M-m}(1)\rmk^{\frac 12}
{\lV T_{\vv_p}^{m-1}(1)\rV}^{\frac 12}\rmk
\lmk
\sum_{\mu=1}^n\lV v_{\mu}\rV^2\rmk^{\frac 12}
\lV T^M_{\vv_{\bar p}}(\bar p)\rV^{\frac 12}\\
&+\sup_{M\in\nan}\lmk \sum_{m=1}^M\lmk \Tr T_{\vv_{\bar p}}^{M-m}(1)\rmk^{\frac 12}
{\lV T_{\vv_p}^{m-1}(1)\rV}^{\frac 12}\rmk^{2}\sum_{\mu=1}^n\lV v_{\mu}\rV^2.
\end{align*}
\end{enumerate}
\end{lem}
See Section \ref{apbasiccp} for the proof.
\begin{lem}\label{pcut}
Let $n,k\in\nan$, $p\in\caP(\mk)$ and $\vv\in\mk^{\times n}$,
and suppose that (ii),(vi) of {\it Condition 2} is satisfied for this $(n,k,p,\vv)$.
Furthermore, assume that
there exist constants $c_{p}>0$, $0<s_{p}<1$, 
a positive linear functional $\varphi_{p}$ on $\mk$
and a positive  element $e_{p}\in \mk$
such that $s(\varphi_p), s(e_{p})\le p$, $\varphi_p(e_p)\neq 0$ and
\begin{align}\label{tps}
\lV T^N_{\vv_{p}}\lmk A\rmk-\varphi_{p}(A)e_{p}\rV\le
c_{p}s_{p}^N\lV A\rV,\quad
\text{for all} \; A\in \mk,\text{and}\;N\in\nan.
\end{align}
Then, there exist a positive linear functional $\varphi_{\vv,p}^{(r)}$ and a constant
$0<s_{\vv,p}^{(r)}<1$
such that $\varphi_{\vv,p}^{(r)}(e_p)=1$, and
 $T_{\vv}$ satisfies the Spectral Property II with respect to $(s_{\vv,p}^{(r)}, \varphi_{\vv,p}^{(r)}(1)e_p,\lmk \varphi_{\vv,p}^{(r)}(1)\rmk^{-1}\varphi_{\vv,p}^{(r)})$.
If furthermore 
for any $\eta\in \cc^k$ with $p\eta\neq 0$,
there exists an $l_{\eta}\in\nan$ such that $\lmk \caK_{l_{\eta}}(\vv)\rmk ^*\eta=\cc^k$.
then
$\varphi_{\vv,p}^{(r)}$ is faithful.
\end{lem}
\begin{proof}
By (ii) of the {Condition 2} , we can apply Lemma \ref{basiccp}.
By (vi) of the {Condition 2} and (\ref{tps}),
we have 
$\sup_{M\in\nan}\lV T^M_{\vv_{p}}\rV<\infty$
and $\sup_{M\in\nan}\lV T^M_{\vv_{\bar p}}\rV<\infty$.
Furthermore, we have
\begin{align*}\sup_{N\in\nan}\lmk
 \sum_{m=1}^N\lmk \Tr T_{\vv_{\bar p}}^{N-m}(1)\rmk^{\frac 12}
\lV T_{\vv_p}^{m-1}(1)\rV^{\frac 12}\rmk
\le \sup_{M\in\nan}\lV T^M_{\vv_{p}}\rV^{\frac12}
 \sup_{N\in\nan}\lmk \sum_{m=1}^N\ \sqrt kc_{\bar p}^{\frac 12}s_{\bar p}^{\frac12(N-m)}\rmk<\infty.
\end{align*}
Therefore, by {\it 4} of Lemma \ref{basiccp},
$\sup_{M\in\nan}\lV T^M_{\vv}\rV<\infty$.
Furthermore, by  (vi) of the {Condition 2}  and Lemma \ref{basiccp} {\it 3}, we have
\begin{align}\label{pcut2}
\lim_{l\to\infty}
\lV T^l_{\vv}(\cdot)-pT^l_{\vv}(\cdot )p\rV=0.
\end{align}
For any $\varepsilon>0$, choose $l_0\in\nan$ such that
$\lmk
\lV T^{l_0}_{\vv}(\cdot)-pT^{l_0}_{\vv}(\cdot )p\rV\sup_{L\in\nan}\lV T^L_{\vv}\rV\rmk<\frac \varepsilon 4$
and $c_p s_p^{l_0}\sup_{L\in\nan}\lV T^L_{\vv}\rV<\frac\varepsilon 4$.
For any $N,M\ge 2l_0$, using {\it 1},  {\it 3} of Lemma \ref{basiccp} and (\ref{tps})
we have
\begin{align*}
&\lV T^N_{\vv}(A)-T^M_{\vv}(A)\rV\\
&\le
 \lV T^{N-l_0}_{\vv}\lmk T^{l_0}_{\vv}(A)-pT^{l_0}_{\vv}(A )p\rmk\rV+
 \lV T^{M-l_0}_{\vv}\lmk T^{l_0}_{\vv}(A)-pT^{l_0}_{\vv}(A )p\rmk\rV\\
& +\lV
 T_{\vv_p}^{N-l_0}\lmk pT^{l_0}_{\vv}(A )p\rmk-\varphi_p\lmk pT^{l_0}_{\vv}(A )p\rmk e_p
 \rV
 +\lV
 T_{\vv_p}^{M-l_0}\lmk pT^{l_0}_{\vv}(A )p\rmk-\varphi_p\lmk pT^{l_0}_{\vv}(A )p\rmk e_p
 \rV<\varepsilon\lV A\rV.
 \end{align*}
Hence $\{T^N_{\vv}\}$ is a Cauchy sequence and has a limit $T_{\vv}^{\infty}$.

Clearly we have $T_{\vv}^{\infty}\circ T_{\vv}=T_{\vv}^{\infty}$ and $T_{\vv}^{\infty}$ is positive.
Furthermore, for any $A\in\mk$, 
\[
T_{\vv}^{\infty}\lmk A\rmk=\lim_{N\to\infty}T^{2N}_{\vv}(A)
=\lim_{N\to\infty}T^{N}_{\vv}\lmk p T^{N}_{\vv}(A)p \rmk
=\lim_{N\to\infty}T^{N}_{\vv_p}\lmk p T^{N}_{\vv}(A)p \rmk
=\lim_{N\to\infty}\varphi_p\lmk p T^{N}_{\vv}(A)p \rmk e_p\in\cc e_p.
\]
Therefore, there exists a positive linear functional $\varphi_{\vv,p}^{(r)}$ such that
$T_{\vv}^{\infty}(\cdot)=\varphi_{\vv,p}^{(r)}(\cdot)e_p$.
 By the $T_{\vv}$-invariance of  
$T_{\vv}^{\infty}$,
 $\varphi_{\vv,p}^{(r)}$ is $T_{\vv}$-invariant.
Note that
\[
T_{\vv}(e_p)=T_{\vv_p}(e_p)=\lmk\varphi_p(e_p)\rmk^{-1}\lim_{N\to \infty}T_{\vv_p}^{N+1}(e_p)
=e_p.
\]
From $\lim_{N\to\infty} T_{\vv}^N(\cdot )=\varphi_{\vv,p}^{(r)}(\cdot) e_p$,
$r_{T_\vv}=1$, and
$1$ is a non-degenerate eigenvalue of $T_{\vv}$.
This equality also implies 
$P_{\{1\}}^{T_{\vv}}=\varphi_{\vv,p}^{(r)}(\cdot )e_p$.
The rest of the spectrum of $T_{\vv}$ is in a disk $\caB_{s_{\vv,p}^{(r)}}(0)$,
for some $0<s_{\vv,p}^{(r)}<1$.
Furthermore, we have $\varphi_{\vv,p}^{(r)}(e_p)e_p=\lim T_{\vv}^N(e_p)=e_p$.
Hence $\varphi_{\vv,p}^{(r)}(e_p)=1$.(Note that $e_p\neq 0$ because $\varphi_p(e_p)\neq 0$.)
By this, we have $\varphi_{\vv,p}^{(r)}(1)>0$ and
$\varphi_{\vv,p}^{(r)}(1) e_p=P_{\{1\}}^{T_{\vv}}(1)$.
Hence $T_{\vv}$ satisfies the Spectral Property II with respect to $(s_{\vv,p}^{(r)}, \varphi_{\vv,p}^{(r)}(1)e_p,\lmk \varphi_{\vv,p}^{(r)}(1)\rmk^{-1}\varphi_{\vv,p}^{(r)})$.

To prove the last statement, assume that for any $\eta\in \cc^k$ with $p\eta\neq 0$,
there exists an $l_{\eta}\in\nan$ such that $\lmk \caK_{l_{\eta}}(\vv)\rmk ^*\eta=\cc^k$.

Let $\rho\in \mk_{+}$ be the density matrix of $\varphi_{\vv,p}^{(r)}$, i.e., $\varphi_{\vv,p}^{(r)}=\Tr\rho(\cdot)$.
By the $T_{\vv}$-invariance of $\varphi_{\vv,p}^{(r)}$, we have 
\begin{align}\label{ts}
\sum_{\mu=1}^nv_{\mu}^*\rho v_{\mu}=\rho.
\end{align}

We claim that for any $\xi\in \overline{s(\rho)}\cc^k$, and $\eta\in s(\rho)\cc^k$,
we have
$\braket{\eta}{\caK_l(\vv)\xi}=0$ for all $l\in\nan$.
As $\eta\in s(\rho)\cc^k$, there exists $c_\eta>0$ such that
$\rho\ge c_\eta\ket{\eta}\bra{\eta}$.
 By the repeated use of (\ref{ts}),
 for all $l\in\nan$
 we have
 \[
 0=\braket{\xi}{\rho\xi}=\sum_{\mu^{(l)}\in\{1,\ldots,n\}^{\times l}}
 \braket{\xi}{\lmk \widehat{v_{\mu^{(l)}}}\rmk^*\rho  \widehat{v_{\mu^{(l)}}}\xi}
 \ge c_{\eta}\sum_{\mu^{(l)}\in\{1,\ldots,n\}^{\times l}}
 \lv
 \braket{\eta}{\widehat{v_{\mu^{(l)}}}\xi}
 \rv^2.
 \]
 Therefore we have $\braket{\eta}{\widehat{v_{\mu^{(l)}}}\xi}=$
 for any $\mu^{(l)}$ and this proves the claim.
 
 Next, there exists an $\eta\in s(\rho)\cc^k$ such that $p\eta\ne 0$.
 To see this,
 note that $\varphi_{\vv,p}^{(r)}(p)\neq 0$, for $\varphi_{\vv,p}^{(r)}(e_p)\neq 0$.
 If any $\eta\in s(\rho)\cc^k$ satisfies $p\eta=0$, then
 $s(\rho)\le \bar p$ and we have $\varphi_{\vv,p}^{(r)}(p)=0$ which is a contradiction.
 
Let us fix some $\eta\in s(\rho)\cc^k$ with $p\eta\ne 0$.
By the assumption, there exists an $l_{\eta}\in\nan$ such that $\lmk \caK_{l_{\eta}}(\vv)\rmk ^*\eta=\cc^k$.
Then for any $\xi\in \overline s(\rho)\cc^k$, we have
\[
\braket{\cc^k}{\xi}=\braket{\lmk\caK_{l_{\eta}}(\vv)\rmk ^*\eta}{\xi}=\braket{\eta}{\caK_{l_{\eta}}(\vv)\xi}=0,
\]
by the claim.
 This means $\xi=0$. Hence we have $s(\rho)=1$, i.e., $\varphi_{\vv,p}^{(r)}$ is faithful.
\end{proof}
By taking the adjoint of the previous Lemma, we obtain the following Lemma.
\begin{lem}\label{qcut}
Let $n,k\in\nan$, $q\in\caP(\mk)$ and $\vv\in\mk^{\times n}$,
and suppose that (iii),(v) of {\it Condition 2} is satisfied for this $(n,k,q,\vv)$.
Furthermore, assume that
there exist constants $c_{q}>0$, $0<s_{q}<1$, 
a positive linear functional $\varphi_{q}$ on $\mk$
and a positive  element $e_{q}\in \mk$
such that $s(\varphi_q), s(e_{q})\le q$, $\varphi_q(e_q)\neq 0$ and
\begin{align}\label{tps2}
\lV T^N_{\vv_{q}}\lmk A\rmk-\varphi_{q}(A)e_{q}\rV\le
c_{q}s_{q}^N\lV A\rV,\quad
\text{for all} \; A\in \mk,\text{and}\;N\in\nan.
\end{align}
Then, there exist a positive element $e_{\vv,q}^{(l)}\in \mk$ and a constant
$0<s_{\vv,q}^{(l)}<1$
such that $\varphi_q(e_{\vv,q}^{(l)})=1$, and
 $T_{\vv}$ satisfies the Spectral Property II with respect to 
$(s_{\vv,q}^{(l)}, \varphi_q(1)e_{\vv,q}^{(l)},\lmk \varphi_q(1)\rmk^{-1}\varphi_q)$.
If furthermore 
for any $\eta\in \cc^k$ with $q\eta\neq 0$,
there exists an $l_{\eta}\in\nan$ such that $\lmk \caK_{l_{\eta}}(\vv)\rmk \eta=\cc^k$.
then
$e_{\vv,q}^{(l)}$ is strictly positive.
\end{lem}
Now we are ready to prove
Lemma \ref{cpmain}.
\begin{proofof}[Lemma \ref{cpmain}]
Let $k_R:=\rank p$ and  $k_L:=\rank q$.
Note that $q \mk q\simeq \Mat_{k_L}$.
Under this identification, we apply
Lemma \ref{pcut} to
$(n,k_L,pq, \vv_{q})$.
The first condition of Lemma \ref{pcut} ((ii) of {\it Condition 2})
can be checked by
(i),(ii)  of {\it Condition 2}.
The second condition ((vi) of {\it Condition 2}) flows from
\[
\lV T^N_{\vv_{(q-pq)}}(A)\rV
=\lV q T^N_{\vv_{\bar p}}(A) q\rV\le
c_{\bar p}s_{\bar p }^N\lV A\rV,\quad A\in\mk,\quad
N\in\nan.
\]
Here we used (i) (iii) of {\it Condition 2} for the first equality and
(vi) for the last inequality.
The third condition of Lemma \ref{pcut} (i.e., (\ref{tps}))
is now (iv) of {\it Condition 2} itself.
Hence,  Lemma \ref{pcut} is applicable
to $(n,k_L,pq, \vv_{q})$ and we obtain
a nonzero positive linear functional $\varphi_{\vv_q,pq}^{(r)}$ on $q \mk q$
and 
$0<s_{\vv_q,pq}^{(r)}<1$. With respect to $
\lmk s_{\vv_q,pq}^{(r)}, \varphi_{\vv_q,pq}^{(r)}(q)e_{pq},
\lmk  \varphi_{\vv_q,pq}^{(r)}(q)\rmk^{-1}\varphi_{\vv_q,pq}^{(r)}
\rmk$,
 $T_{\vv_q}$ satisfies the
 {\it Spectral Property II}.
Define a positive linear functional
$\tilde\varphi_{\vv_q,pq}^{(r)}$ by $\tilde\varphi_{\vv_q,pq}^{(r)}(A)=\varphi_{\vv_q,pq}^{(r)}(qAq)$
on $\mk$.
Set $\varphi_{\vv}:=\lmk \tilde \varphi_{\vv_q,pq}^{(r)}(1)\rmk^{-1}\tilde\varphi_{\vv_q,pq}^{(r)}$.
Note that $\varphi_{\vv}(1)=1$.
From (vii) of {\it Condition 2},
we have $s(\tilde\varphi_{\vv_q,pq}^{(r)})=q$,
by the last statement of Lemma \ref{pcut}.

Next we apply Lemma \ref{qcut} to
$(n,k,q,\vv)$.
The first condition of Lemma \ref{qcut} ((iii) of {\it Condition 2})
is
(iii)  of {\it Condition 2}.
The second condition ((v) of {\it Condition 2})
is (v) of {\it Condition 2}.
The third condition (\ref{tps2}) follows from  (6) of the 
 {\it Spectral Property II} of  $T_{\vv_q}$.

Hence,  Lemma \ref{qcut} is applicable and we obtain
a positive element $e_{\vv,q}^{(l)}$ in $\mk $,
and $0<s_{\vv,q}^{(l)}<1$. 
Set $e_{\vv}=e_{\vv,q}^{(l)}$ and $s_{\vv}= s_{\vv,q}^{(l)}$.
With respect to $\lmk s_{\vv},e_{\vv} ,\varphi_{\vv}\rmk$,
 $T_{\vv}$ satisfies the
 {\it Spectral Property II}.

We know $s(\varphi_{\vv})=s(\tilde\varphi_{\vv_q,pq}^{(r)})=q$.
We show $s(e_{\vv})=p$ in the rest of the proof.

To see this, we apply Lemma \ref{qcut}
to $(n, k_R,pq, \vv_p)$, under the identification
$p\mk p\simeq \Mat_{k_R}$.

The first condition of Lemma \ref{qcut} ((iii) of {\it Condition 2})
can be checked by
(i),(iii)  of {\it Condition 2}.
The second condition ((v) of {\it Condition 2}) follows from
\[
\lV T^N_{\vv_{p-pq}}(A)\rV
=\lV T^N_{\vv_{\bar q}}(pAp) \rV\le
c_{\bar q}s_{\bar q}^N\lV A\rV,\quad A\in\mk,\quad
N\in\nan.
\]
Here we used (i) (ii) of {\it Condition 2} for the first equality and
(v) for the last inequality.
The third condition (\ref{tps2}) is (iv) of {\it Condition 2} itself.

Hence,  Lemma \ref{qcut} is applicable and we obtain
a positive element $e_{\vv_p,pq}^{(l)}$ in $p \mk p $,
and $0<s_{\vv_p,pq}^{(l)}<1$. 
With respect to $\lmk s_{\vv_p,pq}^{(l)},\varphi_{pq}(1)e_{\vv_p,pq}^{(l)} ,\varphi_{pq}(1)^{-1}\varphi_{pq}\rmk$,
 $T_{\vv_p}$ satisfies the
 {\it Spectral Property II}.
 By (viii) of {\it Condition 2}, $s(e_{\vv_p,pq}^{(l)})=p$.
 
As $\varphi_{\vv}$, $\varphi_{pq}$ are faithful on $pq \mk pq$,
we have $\varphi_{\vv}(pq),\varphi_{pq}(pq)>0$.
By
\[
\varphi_{\vv}(pq)e_{\vv}=\lim_{N\to\infty}T_{\vv}^N(pq)
=\lim_{N\to\infty}T_{\vv_p}^N(pq)=
\varphi_{pq}(pq) e_{\vv_p,pq}^{(l)},
\] 
we obtain $s(e_{\vv})=s(e_{\vv_p,pq}^{(l)})=p$.
\end{proofof}
\subsection{The intersection property}\label{intergen}
In this subsection we prove the following Lemma.
\begin{lem}\label{lem:inter}
Let $m_2,m_3\in\nan$.
Let $n,k\in\nan$ and $\vv\in\mk^{\times n}$.
Assume that the triple $(n,k,\vv)$ satisfies
the Condition 4 for $(m_2,m_3)$
with respect to  $X_{m_2}\in\caK_{m_2}(\vv)$.
Then,
\begin{description}
\item[(i)]for all $N\ge m_3$,
\[
(X_{m_2}^{-1})^*\ker\gv{N}\subset \ker\gv{N+m_2}, and
\]
\item[(ii)]for all $N\ge m_2+m_3+1$, 
\begin{align*}
\cgv{N}=\lmk\cgv{N-1}\otimes\cc^n\rmk\cap\lmk\cc^n\otimes\cgv{N-1}\rmk.
\end{align*}
\end{description}
\end{lem}
\begin{proof}
First we prove (i).
For any $N\ge m_3$ and $Y\in \ker \gv{N}$,
by the definition of $\gv{N}$, we have
$\Tr YZ^*=0$ for any $Z\in\caK_{N}(\vv)$.
By the assumption, for any $\mu^{(N+m_2)}\in\{1,\ldots,n\}^{\times (N+m_2)}$,
we have $X_{m_2}^{-1}\widehat{v_{\mu^{(N+m_2)}}}\in \caK_{N}(\vv)$.
From these two observations,
we have
\begin{align*}
\gv{N+m_2}\lmk (X_{m_2}^{-1})^*Y \rmk
&=
\sum_{\mu^{(N+m_2)}\in\{1,\ldots,n\}^{\times (N+m_2)}}
\lmk
\Tr \lmk (X_{m_2}^{-1})^*Y \lmk \widehat{v_{\mu^{(N+m_2)}}}\rmk^*\rmk\rmk
\widehat{\psi_{\mu^{(N+m_2)}}}\\
&=\sum_{\mu^{(N+m_2)}\in\{1,\ldots,n\}^{\times (N+m_2)}}
\lmk
\Tr \lmk Y \lmk (X_{m_2}^{-1}\widehat{v_{\mu^{(N+m_2)}}}\rmk^*\rmk\rmk
\widehat{\psi_{\mu^{(N+m_2)}}}=0.
\end{align*}
Hence we have proven (i).

The inclusion $\cgv{N}\subset \lmk\cgv{N-1}\otimes\cc^n\rmk\cap\lmk\cc^n\otimes\cgv{N-1}\rmk$
for any $2\le N\in\nan$
can be seen by
\begin{align*}
\gv{N}(Y)&=\sum_{\nu\in\{1,\ldots,n\}}
\sum_{\mu^{(N-1)\in\{1,\ldots,n\}^{\times (N-1)}}}
\lmk \Tr\lmk Y v_\nu^* \lmk \widehat{v_{\mu^{(N-1)}}})^* \rmk\rmk\rmk
 \widehat{\psi_{\mu^{(N-1)}}}\otimes \psi_{\nu}\\
 &=\sum_{\nu\in\{1,\ldots,n\}}
\gv{N-1}(Y v_\nu^*)\otimes  \psi_{\nu}\\
&=\sum_{\nu\in\{1,\ldots,n\}}
\sum_{\mu^{(N-1)\in\{1,\ldots,n\}^{\times (N-1)}}}
\lmk \Tr\lmk Y \lmk \widehat{v_{\mu^{(N-1)}}})^*v_\nu^*  \rmk\rmk\rmk
 \psi_{\nu}\otimes\widehat{\psi_{\mu^{(N-1)}}} \\
 &=\sum_{\nu\in\{1,\ldots,n\}}
\psi_{\nu}\otimes\gv{N-1}(v_\nu^*Y ).
\end{align*}

To see the opposite inclusion,
let $N\ge m_2+m_3+1$
 and $\Phi\in \lmk\cgv{N-1}\otimes\cc^n\rmk\cap\lmk\cc^n\otimes\cgv{N-1}\rmk$.
 Then by some sets of $k\times k$ matrices, $\{C_{\mu}\}_{\mu=1}^n$ and $\{D_{\nu}\}_{\nu=1}^n$,
 we can write $\Phi$ as
 \begin{align}\label{pcd}
\Phi=
\sum_{\mu=1}^n\psi_{\mu}\otimes \gv{N-1}(C_{\mu})
=\sum_{\nu=1}^n\gv{N-1}(D_{\nu})\otimes\psi_{\nu}.
\end{align}
From this relation and the definition of $\gv{N}$,
we have for all $\mu,\nu,\mu_{j}\in\{1,\ldots,n\}$, 
$2\le j\le m_2$
and  $\mu_{i}\in\{1,\ldots,n\}$, $m_2+1\le i\le N-1$,
\[
\Tr\lmk\lmk
v_{\mu_{m_2}}^*\cdots v_{\mu_2}^*C_{\mu}v_{{\nu}}^*-
v_{\mu_{m_2}}^*\cdots v_{\mu_2}^*v_{\mu}^*D_{\nu}
\rmk
v_{\mu_{N-1}}^*\cdots v_{\mu_{m_2+1}}^*\rmk
=0.
\]
 Therefore,  for all $\mu,\nu,\mu_{j}\in\{1,\ldots,n\}$, 
$2\le j\le  m_2$,
 \begin{align}\label{cd}
 v_{\mu_{m_2}}^*\cdots v_{\mu_2}^*C_{\mu}v_{{\nu}}^*-
v_{\mu_{m_2}}^*\cdots. v_{\mu_2}^*v_{\mu}^*D_{\nu}
\in \ker\gv{N-1-m_2}.
\end{align}
(If $m_2=1$, replace $v_{\mu_{m_2}}^*\cdots v_{\mu_2}^*$ by $1$, and understand (\ref{cd}) as
$C_{\mu}v_{{\nu}}^*-
v_{\mu}^*D_{\nu}
\in \ker\gv{N-2}$.
)

We claim that there exists $\tilde C\in\mk$ such that
$\tilde C v_{\nu}^*-D_{\nu}\in\ker\gv{N-1}$ for all $\nu\in\{1,\ldots,n\}$.
 As $X_{m_2}\in\caK_{m_2}(\vv)$,
 there exists a set of coefficients $\{\alpha_{\mu^{(m_2)}}\}_{\mu^{(m_2)}\in\{1,\ldots n\}^{\times m_2}}\subset \cc$
 such that
 $X_{m_2}=\sum_{\mu^{(m_2)}}\alpha_{\mu^{(m_2)}}\widehat{v_{\mu^{(m_2)}}}$.
 On the other hand, as $N-1-m_2\ge m_3$,
(\ref{cd}) implies
 \begin{align}\label{xcd}
(X_{m_2}^{-1})^{*}\lmk
v_{\mu_{m_2}}^*\cdots v_{\mu_2}^*C_{\mu_1}v_{{\nu}}^*-
v_{\mu_{m_2}}^*\cdots. v_{\mu_2}^*v_{\mu_{1}}^*D_{\nu}\rmk
\in (X_{m_2}^{-1})^{*}\ker\gv{N-1-m_2}
\subset\ker\gv{N-1},
\end{align}
for all $\nu,\mu_{j}\in\{1,\ldots,n\}$, $1\le j\le m_2$.
We used (i) for the last inclusion.
Set 
\[\tilde C:= \sum_{\mu^{(m_2)}=(\mu_1,\ldots,\mu_{m_2})}
\overline{\alpha_{\mu^{(m_2)}}}(X_{m_2}^{-1})^{*}
v_{\mu_{m_2}}^*\cdots v_{\mu_2}^*C_{\mu_1}\in\mk.
\]
From (\ref{xcd}), we obtain
\begin{align*}
&\tilde C v_{\nu}^*-D_{\nu}=\tilde C v_{\nu}^*-(X_{m_2}^{-1})^*X_{m_2}^*D_{\nu}\\
&=\sum_{\mu^{(m_2)}=(\mu_1,\ldots,\mu_{m_2})}
\overline{\alpha_{\mu^{(m_2)}}}(X_{m_2}^{-1})^{*}\lmk
v_{\mu_{m_2}}^*\cdots v_{\mu_2}^*C_{\mu_1}v_{{\nu}}^*-
v_{\mu_{m_2}}^*\cdots. v_{\mu_2}^*v_{\mu_{1}}^*D_{\nu}\rmk
\in
\ker\gv{N-1},
\end{align*}
for all $\nu\in\{1,\ldots,n\}$,
proving the claim.

Substituting this to (\ref{pcd}), we have
\[
\Phi
=\sum_{\mu_{N}=1}^n\gv{N-1}(D_{\mu_N})\otimes\psi_{\mu_{N}}
=\sum_{\mu_{N}=1}^n\gv{N-1}(\tilde C v_{\mu_N}^*)\otimes\psi_{\mu_{N}}
=\gv{N}(\tilde C)\in\cgv{N}.
\]
\end{proof}
\subsection{Bijectivity of $\gv{N}\vert_{p\mk q}$}\label{bijecgen}
In this section, we prove the bijectivity of 
$\gv{N}$ on ${p\mk q}$ for $N$ large enough,
under {\it Condition 2} and {\it Condition 3}.
Recall the notations (\ref{eq:efdef}) and (\ref{eq:mvpq})
\begin{lem}\label{lem:bijec}
Let $n,k\in\nan$, $p,q\in\caP(\mk)$, and $\vv\in\mk^{\times n}$.
Suppose that the pentad $(n,k,p,q,\vv)$ satisfies
the {\it Condition 2}, and the {\it Condition 3}
for some $m_1\in\nan $.
Then 
\begin{enumerate}
\item 
$M_{\vv,p,q}\le L_\vv<\infty$, i.e.,
$\gv{N}$ is injective on $p\mk q$ for large $N$,
\item
for
all $N\ge \max\{M_{\vv,p,q}, m_1\}$,
the map
$\gv{N}\vert_{p\mk q}:p\mk q\to \cgv{N}$
is bijective.
\end{enumerate}
\end{lem}

We first introduce the Lemma which we will use repeatedly.
\begin{lem}\label{lem:ge}
Let $n,k\in\nan$, $p,q\in\caP(\mk)$, and $\vv\in\mk^{\times n}$
satisfying the {\it Condition 2}.
Then for any $X,Y\in p\mk q$ and $N\in\nan$,
we have
\begin{align}\label{eq:xyg}
\lv
\braket{\gv{N}(X)}{\gv{N}(Y)}-\braket{X}{Y}_{\vv}
\rv
\le E_{\vv}(N)\braket{X}{X}_{\vv}^{\frac 12}
\braket{Y}{Y}_{\vv}^{\frac 12},
\end{align}
and
\begin{align}\label{eq:xxg}
\lmk 1-E_{\vv}(N)\rmk
\braket{X}{X}_{\vv}
\le
\lV{\gv{N}(X)}\rV^2
\le
\lmk 1+E_{\vv}(N)\rmk
\braket{X}{X}_{\vv}.
\end{align}
Furthermore, for 
$X\in p\mk q$ and $N\in\nan$ with $N\ge L_{\vv}$,
we have
\begin{align}\label{eq:gtn}
\lV
X
\rV_2
\le
\sqrt{\frac{2}{a_{\vv}c_{\vv}}}\lV \gv{N}(X)\rV,
\end{align}
and $\gv{N}$ is injective on $p\mk q$.
\end{lem}
\begin{proof}
The proof is basically Section 5 of \cite{Fannes:1992vq},
but the fact that our $e_{\vv},\rho_{\vv}$
are not strictly positive requires an additional argument.
As in Lemma 5.2 of \cite{Fannes:1992vq},
we have for any $X,Y\in p \mk q$ and $N\in\nan$,
\begin{align}\label{eq:gg}
&\braket{\gv{N}(X)}{\gv{N}(Y)}\nonumber\\
&=\braket{X}{Y}_{\vv}
+\sum_{i=1}^{k}\sum_{j=1}^{k}
\braket{\chi_i^{(k)}}{T_{\vv}^N\circ\lmk \unit-P_{\{1\}}^{T_{\vv}}\rmk
\lmk
X^*\ket{\chi_i^{(k)}}\bra{\chi_j^{(k)}}Y\rmk
\chi_j^{(k)}},
\end{align}
and 
\begin{align*}
&\lv
\braket{\gv{N}(X)}{\gv{N}(Y)}-\braket{X}{Y}_{\vv}
\rv
\le
k^2\lV T_{\vv}^N\lmk 1- P^{T_{\vv}}_{\{1\}}\rmk\rV\lV X\rV\lV Y\rV\\
&=\tilde E_{\vv}(N)\lV X\rV\lV Y\rV
\le E_{\vv}(N)\braket{X}{X}_{\vv}^{\frac 12}
\braket{Y}{Y}_{\vv}^{\frac 12}.
\end{align*}
Here we used (\ref{eq:ac}) for the last inequality.
The inequality (\ref{eq:xxg}) is clear from (\ref{eq:xyg}).
Furthermore, if $N\ge L_{\vv}$,
we have $E_{\vv}(N)<\frac 12$. Substituting this to
(\ref{eq:xxg}) and using (\ref{eq:ac}), we obtain
(\ref{eq:gtn}). The injectivity is clear from this in equality.
\end{proof}
Now we prove Lemma \ref{lem:bijec}
\begin{proofof}[Lemma \ref{lem:bijec}]
{\it 1.} was already proven in Lemma \ref{lem:ge}.
To see {\it 2},
let $N\ge \max\{M_{\vv,p,q},m_1\}$.
As $\gv{N}$ is injective on $p\mk q$,
we have
\[
(\rank p)(\rank q)=
\dim p\mk q=\dim\gv{N}\lmk p\mk q\rmk
\le
\dim\gv{N}(\mk)=(\rank p)(\rank q),
\]
and $\gv{N}\lmk p\mk q\rmk\subset \gv{N}(\mk)$.
Hence we have $\gv{N}\lmk p\mk q\rmk= \gv{N}(\mk)$,
proving {\it 2}.
\end{proofof}
\subsection{Estimation of the overlaps of projections}\label{gapgen}
In this section we prove the following Lemma.
\begin{lem}\label{lem:ov}
Let $n,k\in\nan$, $p,q\in\caP(\mk)$, and $\vv\in\mk^{\times n}$.
Suppose that the pentad $(n,k,p,q,\vv)$ satisfies 
the {\it Condition 2} and the {\it Condition 3}
for some $m_1\in\nan $.
Then for all $l,m,r\in\nan$ with $m\ge \max\{m_1, L_{\vv}\}$,
we have
\begin{align}
\lV
\lmk \unit_{[0,l-1]}\otimes \gvp{m+r}\rmk\lmk \gvp{l+m}\otimes
\unit_{[l+m,l+m+r-1]}-\gvp{l+m+r}\rmk
\rV\le
2 F_{\vv}E_{\vv}(m)\lmk
F_{\vv}^2E_{\vv}(m)+1
\rmk.
\end{align}
\end{lem}
The proof goes parallel to Lemma 6.2 of \cite{Fannes:1992vq}.
However, again the fact that our $e_{\vv},\rho_{\vv}$
are not strictly positive requires additional arguments.
\begin{lem}\label{lem:df}
Let $n,k\in\nan$, $p,q\in\caP(\mk)$, and $\vv\in\mk^{\times n}$
satisfying {\it Condition 2}.
Let $l,m,r\in\nan$,
$\Phi\in \cgv{l+m}\otimes\lmk \cc^n\rmk^{\otimes r}$
and $\Psi\in \lmk \cc^n\rmk^{\otimes l} \otimes \cgv{m+r}$.
Let ${\mathfrak P}^{\Phi}:=\{\tilde \Phi_{\mu^{(r)}}\}_{\mu^{(r)}\in\{1,\ldots,n\}^{\times r}}\subset \mk$,
${\mathfrak P}_{\Psi}:=\{\tilde \Psi_{\mu^{(l)}}\}_{\mu^{(l)}\in\{1,\ldots,n\}^{\times l}}\subset \mk$
be sets of matrices
such that
\begin{align*}
\Phi=\sum_{\mu^{(r)}} \gv{l+m}\lmk \tilde \Phi_{\mu^{(r)}}\rmk\otimes\widehat{\psi_{\mu^{(r)}}}
,\quad
\Psi=\sum_{\mu^{(l)}} \widehat{\psi_{\mu^{(l)}}}
\otimes\gv{m+r}\lmk \tilde \Psi_{\mu^{(l)}}\rmk.
\end{align*}
Define
\begin{align}\label{eq:vp}
V^{{\mathfrak P}^{\Phi}}:=\sum_{\mu^{(r)}}
\tilde \Phi_{\mu^{(r)}}\rho_{\vv}\widehat{v_{\mu^{(r)}}}y_{\vv},\quad
V_{{\mathfrak P}_{\Psi}}:=\sum_{\mu^{(l)}}
x_{\vv}\widehat{v_{\mu^{(l)}}}e_{\vv}\tilde \Psi_{\mu^{(l)}}\in\mk.
\end{align}
Then we have
\begin{align}\label{eq:pp}
\lv
\braket{\Phi}{\Psi}
-\braket{V^{{\mathfrak P}^{\Phi}}}{V_{{\mathfrak P}_{\Psi}}}_{\vv}
\rv
\le
k\lV T^m_{\vv}\lmk\unit-P^{T_{\vv}}_{\{1\}}\rmk\rV
\lmk
\Tr T^l_{\vv}(1)\sum_{\mu^{(r)}}\tilde \Phi_{\mu^{(r)}}\tilde \Phi_{\mu^{(r)}}^*\rmk^{\frac 12}
\lmk
\Tr T^r_{\vv}\lmk
\sum_{\mu^{(l)}}\tilde \Psi_{\mu^{(l)}}^*\tilde 
\Psi_{\mu^{(l)}}
\rmk
\rmk^{\frac 12}
\end{align}
\end{lem}
\begin{rem}
Note that such 
${\mathfrak P}^{\Phi}$, ${\mathfrak P}_{\Psi}$ are not necessarily
unique.
\end{rem}
\begin{proof}
As in the proof of Lemma \ref{lem:ge}, it is straightforward to show that
\begin{align}\label{eq:gg}
\braket{\Phi}{\Psi}-
\sum_{\mu^{(r)}\mu^{(l)}}
\varphi_{\vv}\lmk
\tilde \Phi_{\mu^{(r)}}^*
\widehat{v_{\mu^{(l)}}}e_{\vv}\tilde \Psi_{\mu^{(l)}}
\lmk \widehat{v_{\mu^{(r)}}}\rmk^*\rmk
\end{align}
is bounded by the right hand side of (\ref{eq:pp}).
We rewrite the second term.
Note that $\lmk \widehat{v_{\mu^{(l)}}}\rmk p=p\lmk \widehat{v_{\mu^{(l)}}}\rmk p$ and
$\lmk \widehat{v_{\mu^{(r)}}}\rmk^* q=q\lmk \widehat{v_{\mu^{(r)}}}\rmk^* q$.
From this and the relations $e_\vv x_\vv=p$, $\rho_{\vv}y_{\vv}=q$, 
$s(e_\vv)=p$ and $s(\rho_{\vv})=q$, we have
\begin{align*}
&\sum_{\mu^{(r)}\mu^{(l)}}
\varphi_{\vv}\lmk
\tilde \Phi_{\mu^{(r)}}^*
\widehat{v_{\mu^{(l)}}}e_{\vv}\tilde \Psi_{\mu^{(l)}}
\lmk \widehat{v_{\mu^{(r)}}}\rmk^*\rmk
=\sum_{\mu^{(r)}\mu^{(l)}}
\Tr\lmk\rho_{\vv} y_{\vv}\lmk \widehat{v_{\mu^{(r)}}}\rmk^*\rho_{\vv}
\tilde \Phi_{\mu^{(r)}}^*e_{\vv}x_{\vv}
\widehat{v_{\mu^{(l)}}}e_{\vv}\tilde \Psi_{\mu^{(l)}}
\rmk\\
&=\varphi_{\vv}\lmk
\lmk{V^{{\mathfrak P}^{\Phi}}}\rmk^*
e_\vv
\lmk
{V_{{\mathfrak P}_{\Psi}}}\rmk
\rmk=\braket{V^{{\mathfrak P}^{\Phi}}}{V_{{\mathfrak P}_{\Psi}}}_{\vv}.
\end{align*}
\end{proof}
\begin{lem}\label{lem:qpp}
Let $n,k\in\nan$, $p,q\in\caP(\mk)$, and $\vv\in\mk^{\times n}$.
Suppose that the pentad $(n,k,p,q,\vv)$ satisfies the
{\it Condition 2} and the {\it Condition 3}
for some $m_1\in\nan $.
Let $l,m,r\in\nan$ with $m\ge \max\{m_1, L_{\vv}\}$,
$\Phi\in \cgv{l+m}\otimes\lmk \cc^n\rmk^{\otimes r}$
and $\Psi\in \lmk \cc^n\rmk^{\otimes l} \otimes \cgv{m+r}$.
Then 
\begin{description}
\item[(1)]
there exist unique 
${\mathfrak Q}^{\Phi}:=\{ \Phi_{\mu^{(r)}}\}_{\mu^{(r)}\in\{1,\ldots,n\}^{\times r}}\subset p \mk q$,
${\mathfrak Q}_{\Psi}:=\{ \Psi_{\mu^{(l)}}\}_{\mu^{(l)}\in\{1,\ldots,n\}^{\times l}}\subset p\mk q$
such that
\begin{align*}
\Phi=\sum_{\mu^{(r)}} \gv{l+m}\lmk \Phi_{\mu^{(r)}}\rmk\otimes\widehat{\psi_{\mu^{(r)}}}
,\quad
\Psi=\sum_{\mu^{(l)}} \widehat{\psi_{\mu^{(l)}}}
\otimes\gv{m+r}\lmk \Psi_{\mu^{(l)}}\rmk.
\end{align*}
\item[(2)]
For ${\mathfrak Q}^{\Phi}$, ${\mathfrak Q}_{\Psi}$ of
(1), we have
\begin{align*}
\Tr\lmk
T^l_\vv(1)\lmk
\sum_{\mu^{(r)}}\Phi_{\mu^{(r)}}
\Phi_{\mu^{(r)}}^*
\rmk
\rmk
\le
\frac{2F_{\vv}}{a_{\vv}c_{\vv}}
\lV \Phi\rV^2,\quad
\Tr\lmk
T^r_\vv\lmk
\sum_{\mu^{(l)}}\Psi_{\mu^{(l)}}^*
\Psi_{\mu^{(l)}}
\rmk
\rmk
\le
\frac{2kF_{\vv}}{a_{\vv}c_{\vv}}
\lV \Psi\rV^2.
\end{align*}
\item[(3)]
For any $X\in p\mk q$,
we have
\begin{align*}
\Tr\lmk
T^l_\vv(1)
\sum_{\mu^{(r)}}
\lmk X\lmk\widehat{v_{\mu^{(r)}}}\rmk^*\rmk
\lmk
\lmk X\lmk\widehat{v_{\mu^{(r)}}}\rmk^*\rmk^*
\rmk
\rmk
\le
\frac{kF_{\vv}^2}{a_{\vv}c_{\vv}}
\braket{X}{X}_{\vv}\\
\Tr\lmk
T^r_\vv\lmk
\sum_{\mu^{(l)}}
\lmk \lmk\widehat{v_{\mu^{(l)}}}\rmk^*X\rmk^*
\lmk \lmk\widehat{v_{\mu^{(l)}}}\rmk^*X\rmk
\rmk
\rmk
\le
\frac{kF_{\vv}^2}{a_{\vv}c_{\vv}}
\braket{X}{X}_{\vv}
\end{align*}\end{description}
\end{lem}
\begin{proof}
{\it (1)} follows from Lemma \ref{lem:bijec}.
The first inequality of {\it (2)} can be checked as
\begin{align*}
&\Tr\lmk
T^l_\vv(1)\lmk
\sum_{\mu^{(r)}}\Phi_{\mu^{(r)}}
\Phi_{\mu^{(r)}}^*
\rmk
\rmk
\le
\lV
T^l_\vv(1)
\rV
\Tr\lmk
\sum_{\mu^{(r)}}\Phi_{\mu^{(r)}}
\Phi_{\mu^{(r)}}^*\rmk
\\
&\le
F_{\vv}
\Tr\lmk
\sum_{\mu^{(r)}}\Phi_{\mu^{(r)}}
\Phi_{\mu^{(r)}}^*
\rmk
\le
\frac{2F_{\vv}}{a_\vv c_\vv} 
\sum_{\mu^{(r)}}\lV
\gv{l+m}\lmk \Phi_{\mu^{(r)}}\rmk
\rV^2
=
\frac{2F_{\vv}}{a_{\vv}c_{\vv}}
\lV \Phi\rV^2
\end{align*}
Here we used (\ref{eq:tnb}) for the second inequality
and (\ref{eq:gtn}) for the third inequality
with $m\ge L_{\vv}$.
The second one can be obtained similarly.
The first inequality of {\it (3)}
can be seen for $X\in p \mk q$, by
\begin{align*}
\Tr\lmk
T^l_\vv(1)
\sum_{\mu^{(r)}}
\lmk X\lmk\widehat{v_{\mu^{(r)}}}\rmk^*\rmk
\lmk
\lmk X\lmk\widehat{v_{\mu^{(r)}}}\rmk^*\rmk^*
\rmk
\rmk
=\Tr\lmk T^{r}_\vv\lmk X^*T^l_\vv\lmk 1\rmk X\rmk\rmk
\le
kF_{\vv}^2\lV X\rV^2
\le
\frac{kF_{\vv}^2}{a_{\vv}c_{\vv}}
\braket{X}{X}_{\vv}.
\end{align*}
Here we used (\ref{eq:tnb}) for the first inequality
and (\ref{eq:ac}) for the second inequality.
The second inequality of {\it (3)} can be obtained similarly.
\end{proof}
\begin{lem}\label{lem:vv}
Let $n,k\in\nan$, $p,q\in\caP(\mk)$, and $\vv\in\mk^{\times n}$.
Suppose that the pentad $(n,k,p,q,\vv)$ satisfies 
the  {\it Condition 2} and the {\it Condition 3}
for some $m_1\in\nan $.
Let $l,m,r\in\nan$ with $m\ge \max\{m_1, L_{\vv}\}$,
$\Phi\in \lmk
\cgv{l+m}\otimes\lmk \cc^n\rmk^{\otimes r}\rmk\cap \cgv{l+m+r}^{\perp}$
and $\Psi\in \lmk
\lmk \cc^n\rmk^{\otimes l} \otimes \cgv{m+r}\rmk
\cap \cgv{l+m+r}^{\perp}$.
Let 
${\mathfrak Q}^{\Phi}:=\{ \Phi_{\mu^{(r)}}\}_{\mu^{(r)}\in\{1,\ldots,n\}^{\times r}}\subset p \mk q$,
${\mathfrak Q}_{\Psi}:=\{ \Psi_{\mu^{(l)}}\}_{\mu^{(l)}\in\{1,\ldots,n\}^{\times l}}\subset p\mk q$
as in {\it (1)} of Lemma \ref{lem:qpp}.
Then, for $V^{{\mathfrak Q}^{\Phi}}$, $V_{{\mathfrak Q}_{\Psi}}$
defined as in (\ref{eq:vp}),
we have
\begin{align*}
\braket{V^{{\mathfrak Q}^{\Phi}}}{V^{{\mathfrak Q}^{\Phi}}}_{\vv}^{\frac 12}
\le\sqrt 2 F_{\vv}^{\frac 32}E_\vv(m)\lV\Phi\rV,\quad
\braket{V_{{\mathfrak Q}_{\Psi}}}{V_{{\mathfrak Q}_{\Psi}}}_\vv^{\frac 12}
\le\sqrt 2 F_{\vv}^{\frac 32}E_\vv(m)\lV\Psi\rV.
\end{align*}
\end{lem}
\begin{proof}
Take an arbitrary $X\in p \mk q$.
We show
\begin{align}\label{eq:xv}
\lv
\braket{V^{{\mathfrak Q}^{\Phi}}}{X}_{\vv}\rv
\le\sqrt 2 F_{\vv}^{\frac 32}E_\vv(m)\lV\Phi\rV\braket{X}{X}_{\vv}^{\frac 12},
\quad
\lv
\braket{X}{V_{{\mathfrak Q}_{\Psi}}}_\vv\rv
\le\sqrt 2 F_{\vv}^{\frac 32}E_\vv(m)\lV\Psi\rV\braket{X}{X}_{\vv}^{\frac 12}.
\end{align}

To see this, note that
\begin{align*}
\gv{l+m+r}(X)
=\sum_{\mu^{(r)}}\gv{l+m}
\lmk X\lmk \widehat{v_{\mu^{(r)}}}\rmk^*\rmk
\otimes\widehat{\psi_{\mu^{(r)}}}
=\sum_{\mu^{(l)}}
\widehat\psi_{\mu^{(l)}}\otimes
\gv{m+r}\lmk \lmk
\widehat v_{\mu^{(l)}}\rmk^*X \rmk.
\end{align*}
Set ${\mathfrak P}^{\gv{l+m+r}(X)}:=\{
X\lmk \widehat{v_{\mu^{(r)}}}\rmk^*
\}_{\mu^{(r)}}$,
${\mathfrak P}_{\gv{l+m+r}(X)}:=\{ \lmk
\widehat v_{\mu^{(l)}}\rmk^*X \}_{\mu^{(l)}}$
and consider the $V^{{\mathfrak P}^{\gv{l+m+r}(X)}}$,
$V_{{\mathfrak P}_{\gv{l+m+r}(X)}}$ given by the formula
(\ref{eq:vp}).
Then we find
$
V^{{\mathfrak P}^{\gv{l+m+r}(X)}}=V_{{\mathfrak P}_{\gv{l+m+r}(X)}}=X
$.

For $\Phi$ and $\Psi$, consider the unique
${\mathfrak Q}^{\Phi}:=\{ \Phi_{\mu^{(r)}}\}_{\mu^{(r)}\in\{1,\ldots,n\}^{\times r}}\subset p \mk q$,
${\mathfrak Q}_{\Psi}:=\{ \Psi_{\mu^{(l)}}\}_{\mu^{(l)}\in\{1,\ldots,n\}^{\times l}}\subset p\mk q$
given in Lemma \ref{lem:qpp}.
Then, by Lemma \ref{lem:df} and Lemma \ref{lem:qpp} and $\Phi\in\caG_{l+m+r}^{\perp}$,
we have
\begin{align*}
&\lv
\braket{V^{{\mathfrak Q}^{\Phi}}}{X}_{\vv}\rv
=\lv
\braket{V^{{\mathfrak Q}^{\Phi}}}{
V_{{\mathfrak P}_{\gv{l+m+r}(X)}}
}_{\vv}
-\braket{\Phi}{\gv{l+m+r}(X)}
\rv\\
&\le
k\lV T^m_{\vv}\lmk\unit-P_{\{1\}}^{T_{\vv}}\rmk\rV
\sqrt{\frac{2F_\vv}{a_\vv c_\vv}}\lV \Phi\rV
\sqrt{\frac{kF_\vv^2}{a_\vv c_\vv}}
\braket{X}{X}_\vv^{\frac 12}
\le\sqrt 2 F_{\vv}^{\frac 32}E_\vv(m)\lV\Phi\rV\braket{X}{X}_{\vv}^{\frac 12},
\end{align*}
proving the first inequality of (\ref{eq:xv}).
From (\ref{eq:iv}) and (\ref{eq:xv}), we obtain
the first inequality of Lemma \ref{lem:vv}.
The second inequality can be seen in the same manner.
\end{proof}
Now we are ready to prove Lemma \ref{lem:ov}.
\begin{proofof}[Lemma \ref{lem:ov}]
In the setting of Lemma \ref{lem:ov}, we fix arbitrary
 $\Phi\in \lmk
\cgv{l+m}\otimes\lmk \cc^n\rmk^{\otimes r}\rmk\cap \cgv{l+m+r}^{\perp}$
and $\Psi\in \lmk
\lmk \cc^n\rmk^{\otimes l} \otimes \cgv{m+r}\rmk
\cap \cgv{l+m+r}^{\perp}$.
From Lemma \ref{lem:vv}, we have
\begin{align*}
\lv
\braket{V^{{\mathfrak Q}^\Phi}}{V_{{\mathfrak Q}_\Psi}}_\vv
\rv
\le
\braket{V^{{\mathfrak Q}^\Phi}}{V^{{\mathfrak Q}^\Phi}}_\vv^{\frac 12}
\braket{V_{{\mathfrak Q}_\Psi}}{V_{{\mathfrak Q}_\Psi}}_\vv^{\frac12}
\le 2 F_\vv^3 E_\vv(m)^2\lV\Phi\rV\lV\Psi\rV.
\end{align*}
Combining this with Lemma \ref{lem:qpp} and Lemma \ref{lem:df},
we obtain
\begin{align*}
\lv
\braket{\Phi}{\Psi}
\rv
\le
2E_{\vv}(m)F_{\vv}\lmk E_{\vv}(m)F_\vv^2+1\rmk\lV\Phi\rV\lV\Psi\rV.
\end{align*}
This completes the proof.
\end{proofof}
\subsection{Proof of Proposition \ref{prop:maingen}}

\begin{proofof}[Proposition \ref{prop:maingen}]
As $r_{T_\vv}=1$, $\vv$ is nonzero and $\caG_{l,\vv}$ is a nonzero
subspace of $\bigotimes_{i=0}^{l-1}\cc^n$.
That $M_{\vv,p,q}<\infty$ is Lemma \ref{lem:bijec}.
Set $m_0=m_2+m_3$.
Then (ii) of Lemma \ref{lem:inter} implies Property~(I,$m_0$),
i.e., (i) of {\it Condition 1} holds.

To prove (ii) of  {\it Condition 1}, 
choose $L_1\in\nan$ such that
\begin{align*}
4 \sqrt{m}F_{\vv}E_{\vv}(m)\lmk
F_{\vv}^2E_{\vv}(m)+1
\rmk
<1,
\end{align*}
for all $m\ge L_1$.
This is possible because of the spectral property of $T_\vv$.
Set $l_0:=\max\{m_1,L_1,L_\vv\}+1\in\nan$.
For any $l\ge l_0$ and $n\ge 2l$,
we use 
Lemma \ref{lem:ov}
replacing $(l,m,r)$ in Lemma \ref{lem:ov}
by $(n-l+1,l-1,1)$.
Then we obtain (ii) for this $l_0$.
\end{proofof}
\section{Properties of the ground state structure of $H_{\Phi_{m,\bb}}$ 
for $\bb\in \ClassA$}\label{sec:gss}
In this section we prove Theorem \ref{thm:asymmetric}  of the Hamiltonian 
$H_{\Phi_{m,\bb}}$, given by 
$\bb\in \ClassA$.
In subsection \ref{subsec:sgc},
we prove the spectral gap and 
in subsection \ref{subsec:es},
we investigate the ground state structure 
of the Hamiltonians in this class.
\subsection{The spectral gap of $H_{\Phi_{m,\bb}}$ for $\bb\in \ClassA$}\label{subsec:sgc}
In this subsection, we prove 
the following proposition which includes (i),(ii) of Theorem \ref{thm:asymmetric}:
\begin{prop}\label{prop:bbspec}
Let $n,n_0\in\nan$ with $n\ge 2$, $k_R,k_L\in\nan\cup \{0\}$. 
Let $(\lal,\bbD,\bbG,Y)\in\caT(k_R,k_L)$ and 
 $\bb\in\ClassA$
with respect to $(n_0,k_R,k_L,\lal,\bbD,\bbG,Y)$.
Then,
there exist a constant $0<s_{\bb}<1$, a state $\varphi_{\bb}$, 
and a positive element $e_{\bb}\in \mnz\otimes\mkk$,
such that 
\begin{description}
\item[(1)]$T_{\bb}$ satisfies the Spectral Property II
with respect to $(s_{\bb},e_{\bb},\varphi_{\bb})$
\item[(2)]
$s(e_{\bb})=\hpu,\quad s(\varphi_{\bb})=\hpd$.
\end{description}
For
all $N\ge l_\bb(n,n_0,k_R,k_L,\lal,\bbD,\bbG,Y)$,
the map
\[
\gb{N}\vert_{\hpu(\mnz\otimes\mkk)  \hpd}:\hpu
(\mnz\otimes\mkk) \hpd\to \cgb{N}
\]
is bijective and 
$\bbm_\bb\le 2l_\bb(n,n_0,k_R,k_L,\lal,\bbD,\bbG,Y)$.
In particular, $\ker\lmk H_{\Phi_{m,\bb}}\rmk_{[0,N-1]}=\Ran\Gamma^{(R)}_{N,\bb}= \cgb{N}$
is $n_0^2(k_L+1)(k_R+1)$ dimensional for $m\ge2l_\bb$,
 and $N\ge m$.
Furthermore,
$H_{\Phi_{m,\bb}}$ is gapped with respect to the open boundary conditions
for all $m\ge \max\{ 2l_\bb, \frac{\log\lmk n_0^2(k_L+1)(k_R+1)+1\rmk}{\log n}\}$.
\end{prop}
\begin{rem}
We use the notation $\rho_\bbB$, $x_\bbB$, $y_\bbB$ etc.
from Remark \ref{tavv}.
\end{rem}

\begin{lem}\label{lem:oob}
Let $n\in\nan$ with $n\ge 2$, and $\bb\in\ClassA$ with respect to $(n_0,k_R,k_L,\lal,\bbD,\bbG,Y)$. Define 
$\oo_{\bb}=(\omega_{1,\bb},\ldots,\omega_{n,\bb})\in \mnz^{\times n}$ by
\[
\omega_{\mu,\bb}\otimes \eij{00}=\lmk\unit\otimes \eij{00}\rmk
B_{\mu}\lmk\unit\otimes \eij{00}\rmk,\quad \mu=1,\ldots,n.
\]
Then $\oo_{\bb}\in\Prim(n,n_0)$.
\end{lem}
\begin{proof}
As elements in $\caD(k_R,k_L,\bbD,\bbG)$ and $Y$ are all upper triangular matrices,
for any $l\in\nan$ and $\mu^{(l)}\in\{1,\ldots,n\}^{\times l}$, we have
\[
\widehat{\omega_{\mu^{(l)},\bb}}\otimes \eij{00}=\lmk\unit\otimes \eij{00}\rmk
\widehat{B_{\mu^{(l)}}}\lmk\unit\otimes \eij{00}\rmk,\quad \mu=1,\ldots,n.
\]
Therefore, we have
\[
\caK_{l}(\oo_{\bb})\otimes \eij{00}=
\lmk\unit\otimes \eij{00}\rmk
\caK_{l}(\bb)\lmk\unit\otimes \eij{00}\rmk
=\mnz\otimes \eij{00},
\]
for any $l\ge \lb$.
This means $\oo_{\bb}$ is primitive.
We have to show that $r_{T_{{\oo_{\bb}}}}=1$.
First, we have
\[
r_{T_{{\oo_{\bb}}}}=\lim_{N\to\infty}
\lV T_{{\oo_{\bb}}}^N\rV^{\frac 1N}
=\lim_{N\to\infty}
\lV
(\unit \otimes\eij{00})
\lmk
T_{\bb}^N\lmk\unit \otimes\eij{00}\rmk\rmk
(\unit \otimes\eij{00})
\rV^{\frac 1N}
\le r_{T_{\bb}}=1.
\]
Define $\bb'$ by
$B_{\mu}':=\omega_{\mu,\bb}\otimes \Lambda_{\lal}$,
 $\mu=1,\ldots,n$.
As $B_{\mu}-B_{\mu}'$ is in $\mnz\otimes \UT_{0,k_1+k_2+1}$
and $B_{\mu}'\in \mnz\otimes \Lambda_{\lal}$,
$T_{\bb}-T_{\bb'}$ is nilpotent.
This and $B_\mu'\in \mnz\otimes \Lambda_\lal$ implies  
$\sigma(T_{\bb'})^c\subset \sigma(T_{\bb})^c$, i.e.,
we have $\sigma(T_{\bb})\subset \sigma(T_{\bb'})$.
Therefore, we have 
$1=r_{T_{\bb}}\le r_{T_{\bb'}}=r_{T_{{\oo_{\bb}}}}$,
proving $r_{T_{{\oo_{\bb}}}}=1$.
\end{proof}
\begin{lem}\label{lem:nijyu}
For 
$(\lal,\bbD,\bbG,Y)\in\caT(k_R,k_L)$,  we have
\[
I_R(D_a)\overline{\pd}=\qu{-(a+1)}I_R(D_a)\overline{\pd},\quad
\overline{\pu}I_L(G_b)=\overline{\pu}I_L(G_b)\qd{b+1}.
\]
\end{lem}
\begin{proof}
Suppose $\eij{ii}I_R(D_a)\overline{\pd}\neq 0$ for $i\le 0$.
As the left hand side is equal to $I_R(E_{ii}^{(k_R,0)}D_a\sum_{j=-k_R}^{-1}E_{jj}^{(k_R,0)})$,
this means there is $j\in\{-k_R,\ldots,-1\}$ such that
$\lambda_{R,i}=\lambda_{R,-a}\lambda_{R,j}$. Note that we have $|\lambda_{R,j}|<1$ because
$j\le -1$. Therefore, we get
$\lv \lambda_{R,i}\rv=\lv \lambda_{R,-a}\rv\lv\lambda_{R,j}\rv
<\lv \lambda_{R,-a}\rv$.
This means $i\le -a-1$, proving the first equality.
The second one can be proven similarly.
\end{proof}
We also use the following facts repeatedly.
\begin{lem}\label{lem:ldgy}
For any $l\in\nan$,
\begin{align*}
\pu(1+Y)^l\pd=\eij{00},\quad
I_L(G_b)=I_L(G_b)\overline{\pu},\quad
I_R(D_a)=\overline{\pd}I_R(D_a).
\end{align*}
\end{lem}
\begin{proof}
The first equality follows from (\ref{eq:ly}). The second and the third equations follows from 
$G_b\in \UT_{0,k_L+1}$ and $D_a\in\UT_{0,k_R+1}$.
\end{proof}
\begin{lem}\label{lem:but}
For any $l\in\nan$ and $\mu^{(l)}\in\{1,\ldots,n\}^{\times l}$, we have
\begin{align}\label{eq:wbp}
&\wb{l}\hpu=\hpu \wb{l}\hpu,\quad
\hpd \wb{l}=\hpd \wb{l}\hpd.
\end{align}
\end{lem}
\begin{proof}
This is because $B_\mu\in\mnz\otimes \UT_{k_L+k_R+1}$.
\end{proof}
\begin{lem}\label{lem:sn}
Let $n\in\nan$ with $n\ge 2$, and $\bb\in\ClassA$ with respect to $(n_0,k_R,k_L,\lal,\bbD,\bbG,Y)$.
Then the pentad $(n,n_0(k_L+k_R+1),\hpu,\hpd,\bb)$ satisfies the {\it Condition 2}.
\end{lem}
\begin{proof}
There exists $\oo_{\bb}\in\Prim(n,n_0)$ given by Lemma \ref{lem:oob}.
We would like to show (i)-(viii) of {\it Condition 2}
for the pentad $(n,n_0(k_R+k_L+1),\hpu,\hpd,\bb)$.
(i): From the definition of
$\hpu,\hpd$, we have $\hpu\hpd=\hpd\hpu=\heij{00}\neq 0$.\\
(ii) and (iii) are Lemma \ref{lem:but}.
(iv):
From $\oo_{\bb}\in\Prim(n,n_0)$ and Lemma \ref{lem:ps},
(iv) can be checked.\\
(v),(vi): 
Note, as in the proof of Lemma \ref{lem:oob},
that $r_{T_{(\bb_{\overline{\hpd}})}}
\le r_{T_{\bb'}}\le |\lambda_{-1}|^2$,
where $\bb'=(B_\mu')_{\mu}$
is given by $B_{\mu}'=\omega_{\mu,\bb}\otimes\Lambda_{\lal}\overline{\pd}$.
This implies (v). (vi) can be shown similarly.\\
(vii),(viii) : We prove (viii). The proof for (vii) is the same.
Assume that $k_R\ge 1$.
Let $\l\ge l_\bb(n,n_0,k_R,k_L,\lal,\bbD,\bbG,Y)$ and $\eta\in \hpu\lmk \cc^{n_0}\otimes \cc^{k_R+k_L+1}\rmk$
with $\hpu\hpd\eta=\heij{00}\eta\neq 0$. Fix some $\alpha\in\{1,\ldots, n_0\}$ such that
$\lmk \zeij{\alpha,\alpha}\otimes \eij{00}\rmk\eta\neq 0$. Set $\eta':=\lmk \unit\otimes\Lambda_{\lal}^l(1+Y)^l\rmk\eta$.
Note that
\[
\caK_l(\bb_{\hpu})\eta=\caK_l(\bb)\hpu\eta
=
\lmk
\mnz\otimes\spa\left\{ \pu,I_R(D_a)\right\}\rmk\eta',
\]
and $\braket{\chi_{\alpha}^{(n_0)}\otimes \fii{0}}{\eta'}\neq 0$.
Here we used Lemma \ref{lem:ldgy} and Lemma \ref{lem:but},
for the first equality, and (\ref{eq:lblb}) for the second one.

We consider the following proposition for $i=-k_R,\ldots,0$:
\begin{center}
$(P_i)$: 
$\cc^{n_0}\otimes \lmk\qu{i}\cc^{k_R+k_L+1}\rmk\subset \kl{l}(\bb_{\hpu})\eta$.
 \end{center}

To see that $(P_{-k_R})$ holds,
note that for any $\beta\in\{1,\ldots, n_0\}$
we have $\lmk
\zeij{\beta\alpha}\otimes I_R(D_{k_R})\rmk\eta'
\in\kl{l}(\bb_{\hpu})\eta$.
By Lemma \ref{lem:nijyu}, 
we have $I_R(D_{k_R})=I_R(D_{k_R})\eij{00}=\eij{-k_R,0}$.
Therefore, we have
\begin{align*}
\lmk
\zeij{\beta\alpha}\otimes I_R(D_{k_R})\rmk\eta'
=\braket{\cnz{\alpha}\otimes \fii{0}}{\eta'}
\cnz{\beta}\otimes \fii{-k_R},
\end{align*}
with $\braket{\cnz{\alpha}\otimes \fii{0}}{\eta'}\neq 0$.
Hence we have
$\cnz{\beta}\otimes \fii{-k_R}\in\kl{l}(\bb_{\hpu})\eta$, proving $(P_{-k_R})$.

Suppose that $(P_{i-1})$ holds for some $i\le 0$. 
We show that $(P_{i})$
holds.
To see this, 
note that for any $\beta\in\{1,\ldots, n_0\}$
we have $\lmk
\zeij{\beta\alpha}\otimes I_R(D_{-i})\rmk\eta'
\in\kl{l}(\bb_{\hpu})\eta$,
if $i\le -1$, and 
$\lmk
\zeij{\beta\alpha}\otimes I_R(\unit)\rmk\eta'
\in\kl{l}(\bb_{\hpu})\eta$.
Now if $i\le -1$, by Lemma \ref{lem:nijyu}, 
we have 
\[
I_R(D_{-i})=
I_R(D_{-i})\eij{00}+I_R(D_{-i})\qu{-1}
=\eij{i0}+\qu{i-1}I_R(D_{-i})\qu{-1}.
\]
This relation also holds for $i=0$ case, if we set $D_0=\unit$.
Therefore, we have
 \begin{align*}\lmk
 \zeij{\beta\alpha}\otimes I_R(D_{-i})\rmk\eta'
 =\braket{\cnz{\alpha}\otimes \fii{0}}{\eta'}
\cnz{\beta}\otimes \fii{i}
 +\text{an element of }
 \cc^{n_0}\otimes \qu{i-1}\cc^{k_R+k_L+1}
 \end{align*}
with $\braket{\cnz{\alpha}\otimes \fii{0}}{\eta'}\neq 0$.
From $(P_{i-1})$, we have
$\cc^{n_0}\otimes \qu{i-1}\cc^{k_R+k_L+1}\subset \kl{l}(\bb_{\hpu})\eta$.
Hence we have
$\cnz{\beta}\otimes \fii{i}\in  \kl{l}(\bb_{\hpu})\eta$, proving $(P_{i})$.

Inductively, we obtain $(P_0)$, proving (viii) for $k_R\ge 1$ case.
The $k_R=0$ case is much simpler. We just need to note 
$\zeij{\beta\alpha}\eta\in\kl{l}(\bb)\eta$.
\end{proof}

\begin{lem}\label{lem:cbt}
Let $n,n_0\in\nan$ with $n\ge 2$, $k_R,k_L\in\nan\cup \{0\}$
and $(\lal,\bbD,\bbG,Y)\in\caT(k_R,k_L)$.
Let $\bb\in\ClassA$
with respect to $(n_0,k_R,k_L,\lal,\bbD,\bbG,Y)$.
Then the pentad $(n,n_0(k_R+k_L+1),\hpu,\hpd,\bb)$ satisfies
the Condition 3 for $l_\bb(n,n_0,k_R,k_L,\lal,\bbD,\bbG,Y)$.
\end{lem}
\begin{proof}
Note that the matrices $\unit$, $\{I_R^{(k_R,k_L)}(D_a)\}_{a=1}^{k_R} $, $\{I_L^{(k_R,k_L)}(G_b)\}_{b=1}^{k_L}$
$\left\{E_{-a,b}^{(k_R,k_L)}\right\}_{a=1,\ldots,k_R,b=1,\ldots,k_L}$ are linearly independent.
By this fact,
for $ l\ge l_\bb(n,n_0,k_R,k_L,\lal,\bbD,\bbG,Y)$, we have
\begin{align*}
&\dim \kl{l}(\bb)=
\dim\lmk
\lmk\mnz\otimes\lmk
 \caD(k_R,k_L,\bbD,\bbG)\lmk \Lambda_{\lal}\lmk 1+Y\rmk\rmk^{l}
\rmk\rmk\rmk\\
&=n_0^2(k_R+1)(k_L+1)=\lmk \rank \hpu\rmk\cdot\lmk \rank\hpd\rmk.
\end{align*}
\end{proof}
\begin{lem}\label{lem:cbf}
Let $n,n_0\in\nan$ with $n\ge 2$, $k_R,k_L\in\nan\cup \{0\}$
and $(\lal,\bbD,\bbG,Y)\in \caT(k_R,k_L)$.
Let $\bb\in\ClassA$ with respect to 
$(\lal,\bbD,\bbG,Y)$.
Then the triple $(n,n_0(k_R+k_L+1),\bb)$ satisfies
the {\it Condition 4} for $(l_\bb,l_\bb)$.
\end{lem}
\begin{proof}
Set $l_\bb:=l_\bb(n,n_0,k_R,k_L,\lal,\bbD,\bbG,Y)$.
By the definition of $l_\bb$
the invertible element $\unit\otimes\Lambda_{\lal}^{l_\bb}(1+Y)^{l_{\bb}}$
belongs to $\lmk\mnz\otimes\lmk
 \caD(k_R,k_L,\bbD,\bbG)\lmk \Lambda_{\lal}\lmk 1+Y\rmk\rmk^{l_{\bb}}
\rmk\rmk=\kl{l_\bb}(\bb)$.
We also note by (\ref{eq:adl}),
\begin{align}
&\lmk
 \unit\otimes\lmk\Lambda_{\lal}(1+Y)\rmk^{l_\bb}
 \rmk^{-1}\kl{N+l_\bb}(\bb)
=\lmk
 \unit\otimes\lmk \Lambda_{\lal}(1+Y)\rmk^{l_\bb}
 \rmk^{-1}
\lmk\mnz\otimes\lmk
 \caD(k_R,k_L,\bbD,\bbG)\lmk \Lambda_{\lal}\lmk 1+Y\rmk\rmk^{N+l_{\bb}}
\rmk\rmk\nonumber\\
&\subset \lmk\mnz\otimes\lmk
 \caD(k_R,k_L,\bbD,\bbG)\lmk \Lambda_{\lal}\lmk 1+Y\rmk\rmk^{N}
\rmk\rmk
=\kl{N}(\bb),\quad N\ge l_\bb.
\end{align}
\end{proof}

\begin{lem}\label{lem:ml}
Let $n,n_0\in\nan$ with $n\ge 2$, $k_R,k_L\in\nan\cup \{0\}$
and $(\lal,\bbD,\bbG,Y)\in\caT(k_R,k_L)$.
Let $\bb\in\ClassA$
with respect to $(n_0,k_R,k_L,\lal,\bbD,\bbG,Y)$. 
Then we have $M_{\bb,\hpu,\hpd}\le l_\bb(n,n_0,k_R,k_L,\lal,\bbD,\bbG,Y)$.
\end{lem}
\begin{proof}
Let $l\ge l_\bb(n,n_0,k_R,k_L,\lal,\bbD,\bbG,Y)$.
Then, we have
\begin{align}\label{eq:ukd}
\hpu \caK_l(\bb)\hpd=
\mnz\otimes\lmk\pu\mkk\pd\rmk.
\end{align}
To see this, note that from Lemma \ref{lem:ldgy},
we have
\begin{align*}
&\pu\lmk \Lambda_{\lal}\lmk 1+Y\rmk\rmk^l \pd=\eij{00},\\
&\pu I_R^{(k_R,k_L)}(D_a)\lmk \Lambda_{\lal}\lmk 1+Y\rmk\rmk^l \pd
=I_R^{(k_R,k_L)}(D_a E_{00}^{(k_R,0)})=I_R^{(k_R,k_L)}(E_{-a0}^{(k_R,0)})=\eij{-a,0}.
\end{align*}
Furthermore, for any $i\in\{-k_R,\ldots,0\}$, we have
\begin{align*}
&\eij{ib}\lmk \Lambda_{\lal}\lmk 1+Y\rmk\rmk^l 
=\sum_{j'=1}^{k_L}\braket{\fii{b}}{\lmk \Lambda_{\lal}\lmk 1+Y\rmk\rmk^l \fii{j'}}\eij{ij'},\quad b\in \{1,\ldots,k_L\},\\
&\eij{ij}=\sum_{b'=1}^{k_L}\braket{\fii{j}}{\lmk \Lambda_{\lal}\lmk 1+Y\rmk\rmk^{-l} \fii{b'}}\eij{ib'}\lmk \Lambda_{\lal}\lmk 1+Y\rmk\rmk^l ,\quad
j\in \{1,\ldots,k_L\}.
\end{align*}
In particular, we have
$\spa\{\eij{ib}\lmk \Lambda_{\lal}\lmk 1+Y\rmk\rmk^l \}_{b=1}^{k_L}=\spa\{\eij{ij}\}_{j=1}^{k_L}$, 
and we obtain
\begin{align*}
&\spa\left \{\pu I_L^{(k_R,k_L)}(G_b)\lmk \Lambda_{\lal}\lmk 1+Y\rmk\rmk^l  \right\}_{b=1}^{k_L}
=\spa\left \{\eij{0b}\lmk \Lambda_{\lal}\lmk 1+Y\rmk\rmk^l  \right\}_{b=1}^{k_L}
=\spa\left\{\eij{0j} \right\}_{j=1}^{k_L},\\
&\spa\left \{E_{-a,b}^{(k_R,k_L)}\lmk \Lambda_{\lal}\lmk 1+Y\rmk\rmk^l\right\}_{b=1}^{k_L}
=\spa\left \{E_{-a,j}^{(k_R,k_L)}\right\}_{j=1}^{k_L},\quad a=1,\ldots,k_R.
\end{align*}
These proves (\ref{eq:ukd}).
If $X\in\lmk
\hpu(\mnz\otimes\mkk)\hpd\rmk\cap\ker\gb{l}$,
then (\ref{eq:ukd}) implies
$X=0$, proving the claim.
\end{proof}
\begin{proofof}[Proposition \ref{prop:bbspec}]
The first statement follows from Lemma \ref{lem:sn} and 
Lemma \ref{cpmain}.
The bijectivity of $\gb{N}$ for $N\ge l_\bb(n,n_0,k_R,k_L,\lal,\bbD,\bbG,Y)$
can be see from Lemma \ref{lem:ml}.
From Lemma \ref{lem:sn},  \ref{lem:cbt} and Lemma \ref{lem:cbf},
applying Proposition \ref{prop:maingen}
we get $m_\bb\le 2l_\bb$ and $\{\caG_{l,\bb}\}$ satisfies {\it Condition 1}.
From Theorem \ref{bruno}, for all $m\ge 2l_\bb$ and $N\ge m$, 
we have $\ker\lmk H_{\Phi_{m,\bb}}\rmk_{[0,N-1]}=\caG_{N,\bb}=\Ran\Gamma^{(R)}_{N,\bb}$
and from the bijectivity proven above, its dimension is 
$n_0^2(k_L+1)(k_R+1)$.
Furthermore, if $m\ge \max\{ 2l_\bb, \frac{\log\lmk n_0^2(k_L+1)(k_R+1)+1\rmk}{\log n}\}$,
then $G_m\ne \unit $ and $\gamma_{l,m}>0$ for any $l\ge \max\{ l_0,m\}$.
Therefore, from (2) of Theorem \ref{bruno},
$H_{\Phi_{m,\bb}}$ is gapped with respect to the open boundary conditions.
\end{proofof}

\subsection{Edge states of $H_{\Phi_{m,\bb}}$}\label{subsec:es}
In this subsection, we prove (iii),(iv),(v),(vi),(vii), (viii) of
 Theorem \ref{thm:asymmetric}.
\begin{lem}\label{lem:fs}
Let $n,n_0\in\nan$ with $n\ge 2$, $k_R,k_L\in\nan\cup \{0\}$, $\bb\in\ClassA$, and
$m_\bb\le m\in\nan$.
Let $\Gamma=\bbZ$, $\Gamma=(-\infty,-1]$ or $\Gamma=[0,\infty)$.
For a state $\omega$ on ${\mathcal A}_\Gamma$, 
$\omega\in\caS_{\Gamma}(H_{\Phi_{m,\bb}})$
if and only if
$\omega(\tau_{i}(1-\gbp{m}))=0$
for any $i\in\bbZ$ with
$[i,i+m-1]\subset \Gamma$.
In particular, $\caS_{\Gamma}(H_{\Phi_{m,\bb}})$ is a convex set.
\end{lem}
\begin{proof}
Suppose that $\omega$ is a state on ${\mathcal A}_{\Gamma}$ 
such that $\omega(\tau_{i}(1-\gbp{m}))=0$,
for any $i\in\bbZ$ with
$[i,i+m-1]\subset \Gamma$.
Then its restriction $\omega\vert_{\caA_{I}}$ 
to each interval $I\subset \Gamma$
is a ground state of $\lmk H_{\Phi_{m,\bb}}\rmk_{I}$.
Hence $\omega$ is a $wk*$-accumulation point of extensions of 
$\omega\vert_{\caA_{I}}\in {\caS}_{I}\lmk
 H_{\Phi_{m,\bb}}\rmk$, hence
$\omega\in {\caS}_{\Gamma}(H_{\Phi_{m,\bb}})$, 
by definition.

If $\omega\in {\caS}_{\Gamma}\lmk H_{\Phi_{m,\bb}}\rmk$, 
then
there exits a subnet $\{I'\}$ of 
intervals in $\Gamma$ associated with states 
$\omega_{I'}$
on ${\mathcal A}_{\Gamma}$ 
such that $\omega_{I'}\vert_{\caA_{I'}}
\in \caS_{I'}(H_{\Phi_{m,\bb}})$,
and $\omega={\rm wk}*-\lim_{I'} \omega_{I'}$.
Because $m\ge m_\bb$, we have $\omega\circ\tau_{i}(1-\gbp{m})=
\lim_{I'} 
\omega_{I'}\lmk \tau_{i}(1-\gbp{m})\rmk=0$,
for any $i\in\bbZ$ with
$[i,i+m-1]\subset \Gamma$.
\end{proof}

\begin{lem}\label{lem:left}
Let $n\in\nan$ with $n\ge 2$, and $\bb\in\ClassA$ with respect to $(n_0,k_R,k_L,\lal,\bbD,\bbG,Y)$.
Set \begin{align*}
\bbL_\bb(A):=\sum_{\mu^{(l)},\nu^{(l)}\in \{1,\ldots,n\}^{\times l}}
\braket{\ws{l}}{A\wsn{l}}
\lmk \wbn{l}\rmk^{*}\rho_{\bb} \wb{l},
\end{align*} 
 if $A\in {\mathcal A}_{[-l,-1]}\simeq \otimes_{i=0}^{l-1}\mn$ for $l\in\nan$.
 Then $\bbL_\bb$ defines a well-defined 
completely positive map on ${\mathcal A}^{\rm loc}_{(-\infty,-1]}$. 
This $L_{\bb}$ extends to a completely positive map from the half-infinite chain ${\mathcal A}_{(-\infty,-1]}$
 to $\mnz\otimes \mkk$, which we will denote by the same symbol $\bbL_\bb$.
We have
\begin{align}\label{eq:lbran}
\Ran \bbL_\bb=\hpd \lmk\mnz\otimes \mkk\rmk\hpd.
\end{align}
\end{lem}

\begin{proof}
Note that from the relation $\sum_{\mu=1}^n B_{\mu}^*\rho_\bb B_{\mu}=\rho_\bb$,
$\bbL_\bb$ is well-defined.
It can be checked directly that $\bbL_\bb\vert_{{\caA}_{[-l,-1]}}$ defines
a completely positive map on ${\caA}_{[-l,-1]}$ 
with norm $\lV \bbL_\bb\vert_{{\caA}_{[-l,-1]}}\rV=\lV \bbL_\bb(1)\rV=\lV \rho_\bb\rV$, for any $l\in\nan$.
Therefore, we can extend it to 
a completely positive map on ${\mathcal A}_{(-\infty,-1]}$.

Next we check (\ref{eq:lbran}).
The inclusion $\Ran \bbL_\bb\subset \hpd \lmk\mnz\otimes 
\mkk\rmk\hpd$ can be seen
from $\lmk \wbn{l}\rmk^{*}\rho_{\bb} \wb{l}\in \hpd \lmk\mnz\otimes \mkk\rmk\hpd$
which follows from (\ref{eq:wbp}) and $s(\rho_\bb)=\hpd$.

To see the opposite inclusion, note that for any
$l\ge l_{\bb}(n,n_0,k_R,k_L,\lal,\bbD,\bbG,Y)$,
$\alpha,\beta\in\{1,\ldots,n_0\}$, and
$b,b'=0,\ldots,k_L$, we have
\begin{align*}
&\lmk \zeij{1,\alpha}\otimes \lmk\Lambda_{\lal}\lmk 1+Y\rmk\rmk^lI_L(G_b)\rmk^*
\rho_\bb
\lmk \zeij{1,\beta}\otimes \lmk\Lambda_{\lal}\lmk 1+Y\rmk\rmk^lI_L(G_{b'})\rmk\\
&=
\braket{\cnz{1}\otimes \fii{0}}
{\rho_\bb\lmk\cnz{1}\otimes \fii{0}\rmk}
\zeij{\alpha,\beta}\otimes \eij{b,b'}\\
&+
\text{an element of}
\mnz\otimes \lmk
\qd{b+1}{\mkk}\qd{b'}+
\qd{b}{\mkk}\qd{b'+1}\rmk,
\end{align*}
by Lemma \ref{lem:nijyu}.
Here, $G_0$ should be understood as $G_0:=\unit$.
As $l\ge l_\bb(n,n_0,k_R,k_L,\lal,\bbD,\bbG,Y)$, the left hand side is in
$\kl{l}(\bb)^*\rho_{\bb}\kl{l}(\bb)$.
Because $s(\rho_\bb)=\hpd$, $
\braket{\cnz{1}\otimes \fii{0}}
{\rho_\bb\lmk\cnz{1}\otimes \fii{0}\rmk}$ is not zero.
Therefore, we have
\begin{align*}
\zeij{\alpha,\beta}\otimes \eij{b,b'}
\in \kl{l}(\bb)^*\rho_{\bb}\kl{l}(\bb)
+
\mnz\otimes \lmk
\qd{b+1}{\mkk}\qd{b'}+
\qd{b}{\mkk}\qd{b'+1}\rmk
\end{align*}
for any $l\ge l_\bb(n,n_0,k_R,k_L,\lal,\bbD,\bbG,Y)$,
$\alpha,\beta\in\{1,\ldots,n_0\}$, and
$b,b'=0,\ldots,k_L$.
In particular, we have 
$\zeij{\alpha,\beta}\otimes \eij{k_L,k_L}
\in \kl{l}(\bb)^*\rho_{\bb}\kl{l}(\bb)$.
Starting from this, by induction with respect to $b,b'$, we conclude
\[
\mnz\otimes \hpd(\mkk)\hpd\subset \spa(\kl{l}(\bb)^*\rho_{\bb}\kl{l}(\bb))\subset\Ran \bbL_\bb.
\]
%
%
%
\end{proof}

Similarly, we obtain the following.
\begin{lem}\label{lem:right}
Let $n\in\nan$ with $n\ge 2$, and $\bb\in\ClassA$ with respect to $(n_0,k_R,k_L,\lal,\bbD,\bbG,Y)$.
Set \begin{align*}
\bbR_\bb(A):=\sum_{\mu^{(l)},\nu^{(l)}\in \{1,\ldots,n\}^{\times l}}
\braket{\ws{l}}{A\wsn{l}}
\wb{l}e_{\bb}\lmk \wbn{l}\rmk^*,
\end{align*} 
 if $A\in {\mathcal A}_{[0,l-1]}\simeq \otimes_{i=0}^{l-1}\mn$ for $l\in\nan$.
 Then $\bbR_\bb$ defines a well-defined completely positive map 
on ${\mathcal A}^{\rm loc}_{[0,\infty)}$. This $R_{\bb}$ extends to a completely positive map from the half-infinite chain ${\mathcal A}_{[0,\infty)}$
 to $\mnz\otimes \mkk$, which we will denote by the same symbol $\bbR_\bb$.
We have
\[
\Ran \bbR_\bb=\hpu \lmk\mnz\otimes \mkk\rmk\hpu.
\]
\end{lem}
\begin{lem}\label{lem:gee}
Let $n\in\nan$ with $n\ge 2$, and $\bb\in\ClassA$ with respect to $(n_0,k_R,k_L,\lal,\bbD,\bbG,Y)$.
For each $\sigma_L\in {\mathfrak E}_{n_0(k_L+1)}$,
under the identification 
$\Mat_{n_0(k_L+1)}\simeq \mnz\otimes \pd\mkk\pd$,
define
$\Xi_L(\sigma_L):\caA_{(-\infty,-1]}\to \cc$
by
\[
\Xi_L(\sigma_L)(A)
:=\sigma_{L}(y_{\bb}^{\frac 12} \bbL_\bb(A)y_{\bb}^{\frac 12}),\quad
A\in\caA_{(-\infty,-1]}.
\]
For each $\sigma_R\in {\mathfrak E}_{n_0(k_R+1)}$,
under the identification 
$\Mat_{n_0(k_R+1)}\simeq \mnz\otimes \pu\mkk\pu$,
define
$\Xi_R(\sigma_R):\caA_{[0,\infty)}\to \cc$
by
\[
\Xi_R(\sigma_R)(A)
:=\sigma_{R}(x_{\bb}^{\frac 12} \bbR_\bb(A)x_{\bb}^{\frac 12}),\quad
A\in\caA_{[0,\infty)}.
\]
For each nonzero $X\in\mnz\otimes\pu\mkk\pd$,
set
\begin{align}
\sigma_{L,X}:=\frac{\Tr\lmk \rho_{\bb}^{\frac12}X^*e_{\bb}X
\rho_{\bb}^{\frac12}\lmk \cdot\rmk \rmk}{\varphi_{\bb}(X^*e_{\bb}X)}\in\caE_{n_0(k_L+1)},\quad
\sigma_{R,X}:=\frac{\varphi_{\bb}\lmk X^* e_{\bb}^{\frac12}\lmk\cdot\rmk
e_{\bb}^{\frac12}X\rmk}{\varphi_{\bb}(X^*e_{\bb}X)}\in\caE_{n_0(k_R+1)}.
\end{align}
Then
we have 
\begin{align}
&\lv\braket{\gb{N}(X)}{\tau_{N-l}(A)\gb{N}(X)}
-\varphi_{\bb}(X^*e_{\bb}X)
\Xi_L\lmk \sigma_{L,X}\rmk\circ\tau_{-l}(A)\rv
\le 
\tilde E_{\bb}(N-l)F_{\bb}\lV X\rV^2\lV A\rV,\nonumber\\
&\lv\braket{\gb{N}(X)}{A\gb{N}(X)}
-\varphi_{\bb}(X^*e_{\bb}X)
\Xi_R\lmk \sigma_{R,X}\rmk(A)
\rv
\le 
\tilde E_{\bb}(N-l)F_{\bb}\lV X\rV^2\lV A\rV
\end{align}
for all nonzero $X\in\mnz\otimes \pu \mkk\pd$, $l\in\nan$, $A\in{\mathcal A}_{[0,l-1]}$, and
$l\le N$.
\end{lem}
\begin{rem}
Recall the definitions (\ref{eq:efdef}).
\end{rem}
\begin{proof}
Let $A\in {\mathcal A}_{[0,l-1]}$ for some $l\in \nan$ and $l\le N$.
Define  a linear map $V:\lmk \bigotimes_{i=0}^{l-1}\cc^n\rmk\otimes \lmk\cc^{n_0}\otimes\cc^{k_L+k_R+1}\rmk
\to \cc^{n_0}\otimes\cc^{k_L+k_R+1}$ by
\begin{align*}
V(\xi\otimes \eta):=\sum_{\mu^{(l)}\in\{1,\ldots,n\}^{l}}\braket{\ws{l}}{\xi}B_{\mu^{(l)}}\eta,\quad
\xi\in \bigotimes_{i=0}^{l-1}\cc^n,\;\;\eta\in \cc^{n_0}\otimes\cc^{k_L+k_R+1}.
\end{align*}
We define a linear map
$\Theta_A:\mnz\otimes\mkk \to\mnz\otimes\mkk$ by
\begin{align*}
\Theta_A(X)=V \lmk A\otimes X\rmk V^*,\quad X\in  \mnz\otimes\mkk.
\end{align*}
Note that 
\[
\varphi_\bb\circ \Theta_A=\Tr\lmk  L_\bb\circ \tau_{-l}\lmk A\rmk\lmk\cdot\rmk\rmk,\quad
\Theta_A\lmk e_\bb\rmk =R_\bb\lmk A\rmk.
\]
Then we have
\[
\lV \Theta_A\rV\le \lV A\rV\lV V\rV\lV V^*\rV=\lV A\rV\lV V V^*\rV
=\lV A\rV\lV T_{\bb}^{l}(1)\rV\le \lV A\rV F_{\bb}.
\]
(Recall the estimate (\ref{eq:tnb}).)
Using these notations, for nonzero $X\in\mnz\otimes\pu\mkk\pd$, we have
\begin{align*}
&\braket{\gb{N}(X)}{\tau_{N-l}(A)\gb{N}(X)}\\
&=\sum_{\mu^{(l)},\nu^{(l)}} \sum_{\alpha,\beta=1}^{n_0}\sum_{i,j=-k_R}^{k_L}
\braket{\ws{l}}{A\wsn{l}}\\
&\braket{\cnz{\alpha}\otimes \fii{i}}{
T_{\bb}^{N-l}\lmk \wb{l} X^* \lmk
\ket{\cnz{\alpha}\otimes \fii{i}}
\bra{\cnz{\beta}\otimes \fii{j}} \rmk  X\lmk
\wbn{l}\rmk^*\rmk\lmk \cnz{\beta}\otimes \fii{j}\rmk}\\
&=\varphi_{\bb}(X^*e_{\bb}X)
\Xi_L\lmk \sigma_{L,X}\rmk\circ\tau_{-l}(A)\\
&+ \sum_{\alpha,\beta=1}^{n_0}\sum_{i,j=-k_R}^{k_L}\\
&\braket{\cnz{\alpha}\otimes \fii{i}}{
T_{\bb}^{N-l}\circ\lmk\unit-P_{\{1\}}^{T_{\bb}}\rmk
\circ\Theta_A\lmk X^* \lmk
\ket{\cnz{\alpha}\otimes \fii{i}}
\bra{\cnz{\beta}\otimes \fii{j}} \rmk  X
\rmk\lmk \cnz{\beta}\otimes \fii{j}\rmk}\\
\end{align*}
Hence we have
\begin{align*}
\lv\braket{\gb{N}(X)}{\tau_{N-l}(A)\gb{N}(X)}-\varphi_{\bb}(X^*e_{\bb}X)
\Xi_L\lmk \sigma_{L,X}\rmk\circ\tau_{-l}(A)\rv
\le
\tilde E_{\bb}(N-l)F_{\bb}\lV X\rV^2\lV A\rV
\end{align*}
The second inequality can be proven similarly.
\end{proof}
\begin{lem}\label{lem:leftgs}
Let $n\in\nan$ with $n\ge 2$, and $\bb\in\ClassA$ with respect to $(n_0,k_R,k_L,\lal,\bbD,\bbG,Y)$.
Let  $\sigma_L\in {\mathfrak E}_{n_0(k_L+1)}$, $\sigma_R\in {\mathfrak E}_{n_0(k_R+1)}$, and consider
$\Xi_L(\sigma_L)$ and $\Xi_R(\sigma_R)$ defined in
Lemma \ref{lem:gee}.
Then for $m\ge {\bbm}_{\bb}$,  we have
$\Xi_L(\sigma_L)\in\caS_{(-\infty,-1]}
\lmk H_{\Phi_{m,\bb}}\rmk$, and $\Xi_R(\sigma_R)\in\caS_{[0,\infty)}
\lmk H_{\Phi_{m,\bb}}\rmk$.
For $m\ge {\bbm}_{\bb}$, the maps
$\Xi_L:{\mathfrak E}_{n_0(k_L+1)}\to 
\caS_{(-\infty,-1]}(H_{\Phi_{m,\bb}})$ and
$\Xi_R:{\mathfrak E}_{n_0(k_R+1)}\to 
\caS_{[0,\infty)}(H_{\Phi_{m,\bb}})$
are affine bijections.
\end{lem}
\begin{proof}
Let $\sigma_L\in{\mathfrak E}_{n_0(k_L+1)} $.
As $\bbL_\bb$ is positive, $\Xi_L(\sigma_L)$
is positive. Furthermore, we have
\[
\Xi_L(\sigma_L)(1)
=\sigma_{L}(y_{\bb}^{\frac 12} \bbL_\bb(1)y_{\bb}^{\frac 12})
=\sigma_{L}(y_{\bb}^{\frac 12} \rho_\bb y_{\bb}^{\frac 12})
=\sigma_{L}(1)=1.
\]
Hence $\Xi_L(\sigma_L)$ is a state on $A\in\caA_{(-\infty,-1]}$.

Next we show $\Xi_L(\sigma_L)\in\caS_{(-\infty,-1]}(H_{\Phi_{m,\bb}})$,
for $m\ge \bbm_{\bb}$.
From Lemma \ref{lem:fs}, it suffices to show 
$\Xi_L(\sigma_L)(\tau_{i}(1-\gbp{m}))=0$
for any $i\in\bbZ$ with
$[i,i+m-1]\subset (-\infty,-1]$.
To do this, note that there exists a set $\{Z_j\}_{j=1}^{n_0(k_L+1)}\subset
\mnz\otimes\pu \mkk\pd$ such that
$\sum_{j}\Tr\lmk e_{\bb}\lmk  Z_j \cdot  Z_j^*\rmk\rmk
=\sigma_{L}\lmk y_\bb^{\frac 12}\cdot y_\bb^{\frac 12}\rmk $.
From Lemma \ref{lem:gee}, we have
\begin{align*}
&\Xi_L(\sigma_L)\lmk
\tau_{i}(1-\gbp{m})\rmk
=\sum_{j}\Tr\lmk e_{\bb}\lmk  Z_j \lmk \bbL_\bb\lmk
\tau_{i}(1-\gbp{m})\rmk\rmk  Z_j^*\rmk\rmk\\
&=\sum_{j}
\lim_{N\to\infty}
\braket{\gb{N}(Z_j)}{\tau_{N}\lmk\tau_{i}(1-\gbp{m})\rmk
\gb{N}(Z_j)}
=0,
\end{align*}
because
$m\ge \bbm_\bb$.

It is clear from the definition that 
$\Xi_L$ is affine.
To show that $\Xi_L$ is injective, 
assume that $\Xi_L(\sigma_1)=\Xi_L(\sigma_2)$
for some $\sigma_1,\sigma_2\in {\mathfrak E}_{n_0(k_L+1)}$.
Then 
we have
\begin{align*}
\sigma_{1}(y_{\bb}^{\frac 12} \bbL_\bb(A)y_{\bb}^{\frac 12})
=\Xi_{L}(\sigma_1)(A)
=\Xi_L(\sigma_2)(A)
=\sigma_{2}(y_{\bb}^{\frac 12} \bbL_\bb(A)y_{\bb}^{\frac 12}),
\end{align*}
for all $A\in\caA_{(-\infty,-1]}$.
As $\bbL_\bb$ is onto $\hpd\lmk
\mnz\otimes \mkk\rmk \hpd\simeq\Mat_{n_0(k_L+1)}$ 
and $y_{\bb}^{\frac 12}$
is an invertible element of $\hpd\lmk
\mnz\otimes \mkk\rmk \hpd\simeq\Mat_{n_0(k_L+1)}$,
 we have
$\sigma_1=\sigma_2$.

To complete the proof, we prove that $\Xi_L$ is onto
$\caS_{(-\infty,-1]}(H_{\Pbm{m}})$ for $m\ge \bbm_\bb$.
Let $\omega\in\caS_{(-\infty,-1]}(H_{\Pbm{m}}) $.
For each $N\in\nan$, let $D_N$ be the density matrix of  
the restriction of $\omega$ to ${\mathcal A}_{[-N,-1]}$, namely 
$\omega(A) = \Tr_{[-N,-1]}(D_N A)$ for any $A\in {\mathcal A}_{[-N,-1]}$.
By Lemma \ref{lem:fs}, we have
$\omega\lmk \tau_{i}(1-\gbp{m})\rmk=0$
for any $i\in\bbZ$ with
$[i,i+m-1]\subset (-\infty,-1]$.
Therefore, from the intersection property,
we have that $\Ran(\tau_N(D_N))\subset \cgb{N}$
for all $N\ge\bbm_{\bb}$.

From Proposition \ref{prop:bbspec}, this means for 
$N\ge \max\{l_\bb(n,n_0,k_R,k_L,\lal,\bbD,\bbG),\bbm_{\bb}\}$,
that there exist 
\[
X_{i,N}\in \hpu\lmk\mnz\otimes\mkk\rmk\hpd,\quad
i=1,\ldots, n_0^2(k_R+1)(k_L+1),\] such that 
\begin{align*}
\tau_N(D_N)=\sum_i\ket{\gb{N}(X_{i,N})}\bra{\gb{N}(X_{i,N})}.
\end{align*}
Furthermore, for  $N\ge L_\bb$,
we have
\begin{align}\label{eq:acb}
\sum_{i=1}^{n_0^2(k_R+1)(k_L+1)}\lV X_{i,N}\rV^2
\le\frac{2}{a_{\bb}c_{\bb}}
\sum_{i=1}^{n_0^2(k_R+1)(k_L+1)} 
\Vert \gb{N}(X_{i,N})\Vert^2 =\frac{2}{a_{\bb}c_{\bb}},
\end{align}
by Lemma \ref{lem:ge}.
Hence, by the compactness, there is a subsequence $\{N_m\}_m$ such that
$\lim_{m\to\infty}X_{i,N_m} = X_{i,\infty}\in \hpu(\mnz\otimes\mkk)\hpd$,
for all $i=1,\ldots,n_0^2(k_R+1)(k_L+1)$. 

By Lemma \ref{lem:gee}, we have
\begin{align*}
&\omega(A)
=\lim_{m\to\infty} \Tr D_{N_m}A=
\lim_{m\to\infty} \Tr\lmk \tau_{N_m}\lmk D_{N_m}\rmk \tau_{N_m}\lmk A\rmk\rmk
=\lim_{m\to\infty}
\sum_i\braket{\gb{N_m}(X_{i,N_m})}{\tau_{N_m}\lmk A\rmk\gb{N_m}(X_{i,N_m})}\\
&
=\sum_i\Tr\lmk e_\bb\lmk X_{i,\infty}\bbL_\bb\lmk A\rmk X_{i,\infty}^*\rmk\rmk
,\quad A\in\caA_{(-\infty,-1]}^{\rm loc}.
\end{align*}
From this, we have
\[
\sum_{i}\Tr\lmk e_\bb\lmk X_{i,\infty}\rho_\bb^{\frac 12}(1 )\rho_\bb^{\frac 12} X_{i,\infty}^* \rmk\rmk
=1,
\]
and
\[
\sigma_L(\cdot):=\sum_{i}\Tr\lmk e_\bb\lmk X_{i,\infty}\rho_\bb^{\frac 12}(\cdot )\rho_\bb^{\frac 12} X_{i,\infty}^* \rmk\rmk.
\]
defines a state $\sigma_L$ on $\Mat_{n_0(k_L+1)}\simeq \mnz\otimes \hpd\mkk\hpd$.
We have $\Xi_L(\sigma_L)(A)=\omega$.
\end{proof}
Similarly, we have the following.
\begin{lem}\label{lem:infty}
Let $n\in\nan$ with $n\ge 2$, and $\bb\in\ClassA$ with respect to $(n_0,k_R,k_L,\lal,\bbD,\bbG,Y)$.
Then for any $m\ge \bbm_{\bb}$, 
$\caS(H_{\Pbm{m}})$ consists of
a unique state
$\omega_{\bb,\infty}$ such that
\[
\omega_{\bb,\infty}\lmk
A
\rmk=
\sum_{\mu^{(l)},\nu^{(l)}\in \{1,\ldots,n\}^{\times l}}
\braket{\ws{l}}{A\wsn{l}}
\varphi_\bb\lmk \wb{l}e_{\bb}\lmk \wbn{l}\rmk^*\rmk,
A\in\caA_{[i,i+l-1]},\;i\in \bbZ,\;l\in \nan .
\]
\end{lem}
The latter property immediately means that 
$\omega_{\bb,\infty}$ is translation invariant.
\begin{rem}
The last representation can be written
\[
\omega_{\bb,\infty}\lmk
A
\rmk=
\sum_{\mu^{(l)},\nu^{(l)}\in \{1,\ldots,n\}^{\times l}}
\braket{\ws{l}}{A\wsn{l}}
\varphi_\bb\lmk \lmk\widehat{\omega_{\mu^{(l)}}}\otimes \eij{00}\rmk e_{\bb}\lmk \widehat{\omega_{\nu^{(l)}}}\otimes \eij{00} \rmk^*\rmk,
A\in\caA_{[i,i+l-1]},\;i\in \bbZ,\;l\in \nan .
\]
\end{rem}

Next we study how the information of support of
 $\sigma_L\in\caE_{n_0(k_L+1)}$ is reflected to state
 $\Xi_L(\sigma_L)$.
\begin{lem}\label{lem:ffs}
If $\sigma_L\in\caE_{n_0(k_L+1)}$ (resp. $\sigma_R\in\caE_{n_0(k_R+1)}$) is faithful, then
\begin{align*}
\inf\left\{ \sigma\lmk\Xi_L(\sigma_L)\vert_{\caA_{[-l,-1]}}\rmk\setminus \{0\}\mid l\in\nan\right\}>0,\quad(resp.
\inf\left\{ \sigma\lmk\Xi_R(\sigma_R)\vert_{\caA_{[0,l-1]}}\rmk\setminus \{0\}
\mid l\in\nan\right\}>0.)
\end{align*}
\end{lem}
\begin{proof}
Let $\sigma_L\in \caE_{n_0(k_L+1)}$ be a faithful state.
We denote the density matrix of $\sigma_L$ by 
$\tilde\sigma_L$.
Let $l_\bb\le l\in\nan$ and $\zeta\in\bigotimes_{i=0}^{l-1}\cc^n$.
By the definition, we have
\begin{align*}
&\Xi_L(\sigma_L)\lmk\tau_{-l} \lmk\ket{\zeta}\bra{\zeta}\rmk\rmk
=
\sigma_L\lmk y_{\bb}^{\frac 12}L_\bb\lmk\tau_{-l} \lmk\ket{\zeta}\bra{\zeta}\rmk \rmk y_{\bb}^{\frac 12}\rmk
=
\sum_{\mu^{(l)},\nu^{(l)}}\braket{\ws{l}}{\zeta}\braket{\zeta}{\wsn{l}}
\sigma_L\lmk y_{\bb}^{\frac12}\lmk\wbn{l}\rmk^*
\rho_{\bb}\wb{l}y_{\bb}^{\frac12}\rmk\\
&=\sum_{\alpha,\beta=1}^{n_0}\sum_{i,j=0}^{k_L}
\braket{\cnz{\alpha}\otimes \fii{i}}{\tilde\sigma_L^{\frac12}y_{\bb}^{\frac12}\sum_{\nu^{(l)}}\braket{\zeta}{\wsn{l}}\lmk\wbn{l}\rmk^*
\rho_{\bb}^{\frac12}\lmk\cnz{\beta}\otimes \fii{j}\rmk}\\
&\quad\quad\braket{\cnz{\beta}\otimes \fii{j}}
{\rho_{\bb}^{\frac12}\sum_{\mu^{(l)}}\braket{\ws{l}}{\zeta}\wb{l}y_{\bb}^{\frac12}\tilde\sigma_L^{\frac12}
\lmk\cnz{\alpha}\otimes \fii{i}\rmk}\\
&=\sum_{\alpha,\beta=1}^{n_0}\sum_{i,j=0}^{k_L}
\lv
\braket{\zeta}{\gb{l}\lmk\rho_{\bb}^{\frac12}\lmk\zeij{\beta\alpha}\otimes
\eij{ji}\rmk\tilde\sigma_{L}^{\frac12}y_{\bb}^{\frac12}\rmk}
\rv^2.
\end{align*}
From this equality, we see $\zeta\in\overline{\tau_{l} \lmk s\lmk \Xi_L(\sigma_L)\vert_{\caA_{[-l,-1]}}\rmk\rmk}
\bigotimes_{i=0}^{l-1}\cc^n$
if and only if 
\[
\zeta\in \lmk\gb{l}\lmk\rho_{\bb}^{\frac12}\lmk\mnz\otimes
\pd\mkk\pd\rmk\tilde\sigma_{L}^{\frac12}y_{\bb}^{\frac12}\rmk\rmk^{\perp}
=\lmk\gb{l}\lmk\lmk\mnz\otimes\pd\mkk\pd\rmk\rmk\rmk^{\perp}.
\]
Here for the equality, we used the fact that $\rho_\bb$, $y_\bb$ and $\tilde\sigma_L$
are invertible in $\mnz\otimes\pd\mkk\pd$.
In other words, we have
\begin{align}\label{eq:sp}
\gb{l}\lmk\lmk\mnz\otimes\pd\mkk\pd\rmk\rmk= {\tau_{l} \lmk s\lmk \Xi_L(\sigma_L)\vert_{\caA_{[-l,-1]}}\rmk\rmk}
\bigotimes_{i=0}^{l-1}\cc^n.
\end{align}

We claim that 
\[
\caW=\{\gb{l}\lmk\lmk\zeij{\beta\alpha}\otimes \eij{0j}\rmk
\tilde\sigma_L^{\frac12}y_\bb^{\frac12}\rmk\mid \alpha,\beta=1,\ldots,n_0,j=0,\ldots,k_L\}
\]
is a basis of $\gb{l}\lmk\lmk\mnz\otimes\pd\mkk\pd\rmk\rmk$.
To see this, note that $\caW$ is linearly independent 
because $M_{\bb,\hpu,\hpd}\le l_\bb\le l$.
Therefore, $\spa\caW$ is  an $n_0^2(k_L+1)$ dimensional subspace of
$\gb{l}\lmk\lmk\mnz\otimes\pd\mkk\pd\rmk\rmk$.
On the other hand, for $X\in \mnz\otimes\pd\mkk\pd$, 
$X\in \ker\gb{l}$ if and only if
\[
0=\Tr\lmk X\lmk \caK_{l}(\bb)\rmk^*\rmk
=\Tr\lmk X\lmk
\mnz\otimes\spa\left\{
\pd\Lambda_{\lal}^l(1+Y)^l,
\pd\Lambda_{\lal}^l(1+Y)^lI_L(G_b)
\right\}\rmk^*
\rmk
\]
Therefore, we have $\dim\lmk\gb{l}\lmk\lmk\mnz\otimes\pd\mkk\pd\rmk\rmk\rmk
=n_0^2(k_L+1)$.
Hence we have $\gb{l}\lmk\lmk\mnz\otimes\pd\mkk\pd\rmk\rmk=\spa\caW$.

Define a bounded operator $A_l$ on $\cc^{n_0}\otimes \cc^{n_0}\otimes \cc^{k_L+1}$
by
\begin{align*}
&\braket{\cnz{\beta}\otimes\cnz{\alpha}\otimes \fiil{j}}{A_l\lmk \cnz{\beta'}\otimes\cnz{\alpha'}\otimes \fiil{j'}\rmk}\\
&:=\braket{\gb{l}\lmk\lmk\zeij{\beta\alpha}\otimes \eij{0j}\rmk
\tilde\sigma_L^{\frac12}y_\bb^{\frac12}\rmk}
{\gb{l}\lmk\lmk\zeij{\beta'\alpha'}\otimes \eij{0j'}\rmk
\tilde\sigma_L^{\frac12}y_\bb^{\frac12}\rmk},
\\
&\alpha,\beta,\alpha',\beta'=1,\ldots,n_0,\quad j,j'=0,\ldots,k_L.
\end{align*}
We would like to bound this operator from below for large $l$.\\
Let $J$ be the antilinear operator on $\cc^{n_0}\otimes \cc^{k_L+1}$ given as the complex conjugation
with respect to the basis $\{\cnz{\alpha}\otimes \fiil{i}\}_{\alpha=1,\ldots,n_0,i=0,\ldots, k_L}$. 
Set $O_L:=J\tilde\sigma_L^{\frac 12}J$.
As we have $s(e_\bb)=\hpu\ge \unit\otimes \eij{00}$, 
there exists a strictly positive element $b_\bb\in\mnz$ such that
$\heij{00}e_{\bb}\heij{00}=b_\bb\otimes\eij{00}$.
Note that
\begin{align*}
&\braket{\lmk\zeij{\beta\alpha}\otimes \eij{0j}\rmk
\tilde\sigma_L^{\frac12}y_\bb^{\frac12}}{\lmk\zeij{\beta'\alpha'}\otimes \eij{0j'}\rmk
\tilde\sigma_L^{\frac12}y_\bb^{\frac12}}_{\bb}
=\sigma_L\lmk\zeij{\alpha\alpha'}\otimes\eij{jj'}\rmk\braket{\cnz{\beta}}{b_\bb\cnz{\beta'}}\\
&=\braket{\cnz{\beta}\otimes\cnz{\alpha}\otimes \fiil{j}}{\lmk b_\bb\otimes O_L^*O_L
\rmk\lmk \cnz{\beta'}\otimes\cnz{\alpha'}\otimes \fiil{j'}\rmk}.
\end{align*}
There exist numbers $c_1,c_2>0$ such that $ c_1\unit\le \lmk b_\bb\otimes O_L^*O_L\rmk\le c_2\unit$.\\
By Lemma \ref{lem:ge}, we have
\begin{align*}
&\lv\braket{\gb{l}\lmk\lmk\zeij{\beta\alpha}\otimes \eij{0j}\rmk
\tilde\sigma_L^{\frac12}y_\bb^{\frac12}\rmk}
{\gb{l}\lmk\lmk\zeij{\beta'\alpha'}\otimes \eij{0j'}\rmk
\tilde\sigma_L^{\frac12}y_\bb^{\frac12}\rmk}\right.\\
&\quad\quad\quad\quad\quad\quad\quad\left.-\braket{\lmk\zeij{\beta\alpha}\otimes \eij{0j}\rmk
\tilde\sigma_L^{\frac12}y_\bb^{\frac12}}{\lmk\zeij{\beta'\alpha'}\otimes \eij{0j'}\rmk
\tilde\sigma_L^{\frac12}y_\bb^{\frac12}}_{\bb}
\rv\\
&\le E_\bb(l)\lV\lmk b_\bb^{\frac12}\otimes O_L
\rmk\cnz{\beta}\otimes\cnz{\alpha}\otimes \fiil{j}\rV\lV\lmk b_\bb^{\frac 12}\otimes O_L
\rmk\lmk \cnz{\beta'}\otimes\cnz{\alpha'}\otimes \fiil{j'}\rmk\rV\le c_2E_\bb(l).
\end{align*}
Hence we have
\begin{align*}
\lmk c_1-c_2(n_0^4(k_L+1)^2)E_\bb(l)\rmk \unit\le A_l.
\end{align*}
Set $2l_\bb\le l_1\in\nan$ so that 
$c_2(n_0^4(k_L+1)^2)E_\bb(l)<\frac{c_1}{2}$ for all $l\ge l_1$.
Then we have  
\[
\frac {c_1}{2}\unit \le A_l, \quad l_1\le l.
\]
Applying Lemma \ref{lem:lb},  and using (\ref{eq:sp}), we obtain
\begin{align}
\sum_{\alpha,\beta=1}^{n_0}\sum_{j=0}^{k_L}
\ket{\gb{l}\lmk\lmk\zeij{\beta\alpha}\otimes \eij{0j}\rmk
\tilde\sigma_L^{\frac12}y_\bb^{\frac12}\rmk}
\bra{\gb{l}\lmk\lmk\zeij{\beta\alpha}\otimes \eij{0j}\rmk
\tilde\sigma_L^{\frac12}y_\bb^{\frac12}\rmk}
\ge \frac {c_1}{2}\tau_l\lmk s\lmk \Xi_L(\sigma_L)\vert_{\caA_{[-l,-1]}}\rmk\rmk,
\end{align}
for $l\ge l_1$.

To complete the proof, let $l\ge l_1$. Then for any $\zeta\in\bigotimes_{i=0}^{l-1}\cc^n$, we have
\begin{align*}
&\Xi_L(\sigma_L)\lmk\tau_{-l} \lmk\ket{\zeta}\bra{\zeta}\rmk\rmk
=
\sigma_L\lmk y_{\bb}^{\frac 12}L_\bb\lmk\tau_{-l}\lmk \ket{\zeta}\bra{\zeta}\rmk \rmk y_{\bb}^{\frac 12}\rmk
\ge
c_\bb
\sum_{\mu^{(l)},\nu^{(l)}}\braket{\ws{l}}{\zeta}\braket{\zeta}{\wsn{l}}
\sigma_L\lmk y_{\bb}^{\frac12}\lmk\wbn{l}\rmk^*
\hpd\wb{l}y_{\bb}^{\frac12}\rmk\\
&=c_\bb\sum_{\alpha,\beta=1}^{n_0}\sum_{i,j=0}^{k_L}
\braket{\cnz{\alpha}\otimes \fii{i}}{\tilde\sigma_L^{\frac12}y_{\bb}^{\frac12}\sum_{\nu^{(l)}}\braket{\zeta}{\wsn{l}}\lmk\wbn{l}\rmk^*
\lmk\cnz{\beta}\otimes \fii{j}\rmk}\\
&\quad\quad\braket{\cnz{\beta}\otimes \fii{j}}
{\sum_{\mu^{(l)}}\braket{\ws{l}}{\zeta}\wb{l}y_{\bb}^{\frac12}\tilde\sigma_L^{\frac12}
\lmk\cnz{\alpha}\otimes \fii{i}\rmk}\\
&\ge c_\bb\sum_{\alpha,\beta=1}^{n_0}\sum_{i=0}^{k_L}
\lv
\braket{\zeta}{\gb{l}\lmk\lmk\zeij{\beta\alpha}\otimes
\eij{0i}\rmk\tilde\sigma_{L}^{\frac12}y_{\bb}^{\frac12}\rmk}
\rv^2\\
&\ge  c_\bb\frac{c_1}{2}
\braket{\zeta}{\tau_l\lmk s\lmk \Xi_L(\sigma_L)\vert_{\caA_{[-l,-1]]}}\rmk\rmk\zeta}.
\end{align*}
This proves the claim of the Lemma.
\end{proof}
\begin{lem}\label{lem:spts}
If the supports of $\sigma_L,\sigma_{L}'\in\caE_{n_0(k_L+1)}$  
(resp. $\sigma_R,\sigma_{R}'\in\caE_{n_0(k_R+1)}$) are not orthogonal,
then 
\[
\lV \Xi_L(\sigma_L)-\Xi_L(\sigma_L')\rV<2,\quad
(resp. \lV \Xi_R(\sigma_R)-\Xi_R(\sigma_R')\rV<2.)
\]
\end{lem}
\begin{proof}
Note that $\hat L_\bb : \caA_{(-\infty,-1]}\to \mnz\otimes\pd\mkk\pd\simeq\Mat_{n_0(k_L+1)}$
given by $\hat L_\bb(\cdot):=y_\bb^{\frac 12}L_\bb(\cdot)y_{\bb}^{\frac 12}$
is a unital CP map.
We have $\Xi_L(\sigma)=\sigma\circ\hat L_\bb$.

Let
$\sigma_L,\sigma_{L}'\in\caE_{n_0(k_L+1)}$   be states such that
\[
\lV \Xi_L(\sigma_L)-\Xi_L(\sigma_L')\rV=2.
\]
Then there exists a sequence of self-adjoint elements $\{a_n\}$ in the unit ball of
$\caA_{(-\infty,-1]}$. such that
\begin{align*}
\lv \Xi_L(\sigma_L)(a_n)-\Xi_L(\sigma_L')(a_n)\rv \to 2. 
\end{align*}
As $\hat L_\bb$ is a unital CP map, $\{\hat L_\bb(a_n)\}_n$ is a sequence of
self-adjoint operators in the unit ball of $\Mat_{n_0(k_L+1)}$.
As the unit ball of $\Mat_{n_0(k_L+1)}$ is compact, there exists a subsequence 
$\{a_n'\}$ such that  $\hat L_\bb(a_n')$ converges to some self-adjoint element $x$ in the unit ball of 
$ \Mat_{n_0(k_L+1)}$.
For this $x$, we have
\begin{align*}
\lv\sigma_L(x)-\sigma_L'(x)\rv
=\lim_n \lv
\sigma_L\circ\hat L_\bb(a_n')-\sigma_L'\circ\hat L_\bb(a_n')
\rv
=
\lim_{n}\lv \Xi_L(\sigma_L)(a_n')-\Xi_L(\sigma_L')(a_n')\rv =2
\end{align*}
As $-1\le x\le 1$, this means that $s(\sigma_L)$ and $s(\sigma_L')$ are orthogonal.
\end{proof}

\begin{lem}\label{lem:nr}Let $n\in\nan$ with $n\ge 2$, and $\bb\in\ClassA$ with respect to $(n_0,k_R,k_L,\lal,\bbD,\bbG,Y)$.
Let $m\ge \bbm_\bb$.
For any $\psi\in\caS_{[0,\infty)}(H_{\Phi_{m,\bb}})$ (resp. 
$\psi\in\caS_{(-\infty,-1]}(H_{\Phi_{m,\bb}})$), there exists an  $l_{\psi}\in \nan$
such that $\lV\psi-\psi\circ\tau_{l_\psi}\rV<2$
(resp.  $\lV\psi-\psi\circ\tau_{-l_\psi}\rV<2$).
\end{lem}
\begin{proof}
Let $\psi\in\caS_{(-\infty,-1]}(H_{\Phi_{m,\bb}})$.
By Lemma \ref{lem:leftgs},
there exists a state $\sigma_L$ on 
$\Mat_{n_0(k_R+k_L+1)}\simeq\mnz\otimes\pd\mkk\pd$
such that $\Xi_L(\sigma_L)=\psi$. Let $\hat\sigma_L$ be the density matrix of 
$\sigma_L$.
For each $N\in\nan$, define a state $\kappa_N\in \caE_{n_0(k_L+1)}$ given by the density matrix
\[
\hat \kappa_N:=\rho_\bb^{\frac 12}T_\bb^N\lmk
y_\bb^{\frac 12}\hat \sigma_L  y_\bb^{\frac 12}\rmk \rho_\bb^{\frac 12}.
\]
By a straight forward calculation, one can check
\begin{align*}
\psi\circ\tau_{-N}=\Xi_L(\kappa_N),\quad N\in\nan.
\end{align*}
Hence
the existence of $N\in\nan$ such that
\[
\Tr \hat\sigma_L\hat \kappa_N\neq 0
\]
implies the existence of  $N\in\nan$ such that
$
\lV\psi-\psi\circ\tau_{-N}\rV<2
$
from Lemma \ref{lem:spts}.

Therefore, to prove the Lemma, it suffices to show that for any  nonzero $\eta\in\cc^{n_0}\otimes \cc^{k_L+1}$,
there exists an $N_\eta\in\nan$ such that 
\[
\braket{\eta}{\rho_\bb^{\frac 12}\caK_{N_\eta}(\bb)y_\bb^{\frac 12}\eta}\neq 0.
\]
We prove this by contradiction.
Suppose for some nonzero 
$\eta\in\cc^{n_0}\otimes \cc^{k_L+1}$,
we have
\[
\braket{\eta}{\rho_\bb^{\frac 12}\caK_{N}(\bb)y_\bb^{\frac 12}\eta}=0,\quad N\in\nan.
\]
In particular, we have 
\[
0=\braket{\eta}{\rho_\bb^{\frac 12}\lmk \unit\otimes \lmk \Lambda_{\lal}(1+Y)\rmk^{N} \rmk y_\bb^{\frac 12}\eta}
=\sum_{k=0}^{k_L+k_R+1} {}_NC_{k}\braket{\eta}{\rho_\bb^{\frac 12}\lmk \unit\otimes \Lambda_{\lal}^N Y^k\rmk y_\bb^{\frac 12}\eta},\quad  l_\bb(n,n_0,k_R,k_L,\lal,\bbD,\bbG,Y)\le N\in\nan.
\]
By Lemma \ref{lem:lsl}, for any $\lambda\in\{\lambda_i\}_{i=0}^{k_L}$, we obtain 
\[
\braket{\eta}{\rho_\bb^{\frac 12}\lmk \unit\otimes \sum_{i:\lambda_i=\lambda}\eij{ii}\rmk y_\bb^{\frac 12}\eta}=0.
\]
Summing up for all distinct $\lambda\in\{\lambda_i\}_{i=0}^{k_L}$,
we obtain 
\[
\braket{\eta}{\eta}=\braket{\eta}{\rho_\bb^{\frac 12}y_\bb^{\frac 12}\eta}=0.
\]
This is a contradiction.
\end{proof}
The boundary effect decays exponentially fast in these models.
\begin{lem}\label{lem:vii}
Let $n\in\nan$ with $n\ge 2$, and $\bb\in\ClassA$ with respect to $(n_0,k_R,k_L,\lal,\bbD,\bbG,Y)$.
Let $m\ge \bbm_\bb$.
Then there exist constants $0<s'_\bb<1$ and $C_{\bb}'>0$ such that
for any
$N\in \nan$, $\varphi_L\in\caS_{(-\infty,-1]}(H_{\Phi_{m,\bb}})$, 
and $\varphi_R\in\caS_{[0,\infty)}(H_{\Phi_{m,\bb}})$,
\begin{align*}
\lv \varphi_L\lmk \tau_{-N}(A)\rmk
-\omega_{\bb,\infty}\lmk A\rmk\rv
\le
C_\bb' \lmk s_\bb'\rmk^N\lV A\rV,\quad A\in \caA_{[-\infty,-1]},\\
\lv \varphi_R\lmk \tau_{N}(A)\rmk
-\omega_{\bb,\infty}\lmk A\rmk\rv
\le
C_\bb' \lmk s_\bb'\rmk^N\lV A\rV,\quad  A\in \caA_{[0,\infty)}.
\end{align*}
\end{lem}
\begin{proof}
We prove the first inequality. The second one can be proven in the same way.
Let $\varphi_L\in\caS_{(-\infty,-1]}(H_{\Phi_{m,\bb}})$.
By Lemma \ref{lem:leftgs},
there exists a state $\sigma_L$ on 
$\Mat_{n_0(k_R+k_L+1)}\simeq\mnz\otimes\pd\mkk\pd$
such that $\Xi_L(\sigma_L)=\varphi_L$.
Let $\hat\sigma_L$ be the density matrix of 
$\sigma_L$.
For  $l,N\in\nan$ and 
$A\in \caA_{[-l,-1]}$, we have
$\tau_{-N}(A)\in \caA_{[-N-l,-N-1]}$.
Therefore we have
\begin{align*}
&\varphi_L(\tau_{-N}(A))
=\Xi_L(\sigma_L)(\tau_{-N}(A))
=\sigma_{L}(y_{\bb}^{\frac 12} \bbL_\bb(\tau_{-N}(A))y_{\bb}^{\frac 12})\\
&=\sum_{\mu^{(l)},\nu^{(l)}\in \{1,\ldots,n\}^{\times l}}
\sum_{\mu^{(N)},\nu^{(N)}\in \{1,\ldots,n\}^{\times N}}
\braket{\lmk \ws{l}\otimes\ws{N}\rmk}
{\tau_{-N}(A)\lmk \wsn{l}\otimes\wsn{N}\rmk}
\sigma_{L}\lmk
 y_{\bb}^{\frac 12}
\lmk \wbn{l} \wbn{N}\rmk^{*}\rho_{\bb} \wb{l}\wb{N}y_{\bb}^{\frac 12}\rmk\\
&=\sum_{\mu^{(l)},\nu^{(l)}\in \{1,\ldots,n\}^{\times l}}
\braket{\ws{l}}
{A\wsn{l}}
\sigma_{L}\lmk
 y_{\bb}^{\frac 12}
(T_\bb^*)^N\lmk \lmk \wbn{l} \rmk^{*}\rho_{\bb} \wb{l}\rmk y_{\bb}^{\frac 12}\rmk\\
&=\sigma_{L}(y_{\bb}^{\frac 12} (T_\bb^*)^N\lmk
\bbL_\bb(A)\rmk
y_{\bb}^{\frac 12})
=\Tr \lmk T_\bb^N\lmk y_{\bb}^{\frac 12}\hat\sigma_{L}y_{\bb}^{\frac 12}\rmk\lmk
\bbL_\bb(A)\rmk\rmk.
\end{align*}
Note from Lemma \ref{lem:infty} that 
$\omega_{\bb,\infty}(A)=\Tr e_{\bb}L_{\bb}(A)
=\Tr\lmk P_{\{1\}}^{T_{\bb}}\lmk y_{\bb}^{\frac 12}
\hat\sigma_Ly_{\bb}^{\frac 12}\rmk L_{\bb}(A)\rmk$.
Hence we obtain
\begin{align}
\lv\omega_{\bb,\infty}\lmk A\rmk
-\varphi_L(\tau_{-N}(A))\rv
\le \lV
T_{\bb}^N\lmk\unit-P^{T_{\bb}}_{\{1\}}\rmk\rV
\lV y_\bb\rV \lV A\rV\lV\bbL_\bb\rV.
\end{align}
By the density of $\cup_{l}\caA_{[-l,-1]}$ in $\caA_{(-\infty,-1]}$, this proves the claim.
\end{proof}
This corresponds to (vii) of Theorem \ref{thm:asymmetric}.
Let us use this to prove (viii).
\begin{lem}\label{lem:viii}Let $n\in\nan$ with $n\ge 2$, and $\bb\in\ClassA$ with respect to $(n_0,k_R,k_L,\lal,\bbD,\bbG,Y)$.
Let $m\ge \bbm_\bb$.
Then any element in $\caS_{(-\infty,-1]}(H_{\Phi_{m,\bb}})$ or $\caS_{[0,\infty)}(H_{\Phi_{m,\bb}})$ is a factor
state.
\end{lem}
\begin{proof}
We consider $\varphi\in \caS_{[0,\infty)}(H_{\Phi_{m,\bb}})$.
By Theorem 2.6.10 of \cite{BR1}, it suffices to show that for any $\varepsilon>0$ and $l\in\nan$,
there exists an $L\in\nan$ such that
\begin{align}\label{eq:br}
\lv
\varphi\lmk AB\rmk-\varphi\lmk B\rmk\varphi\lmk A\rmk
\rv
\le
{\varepsilon}\lV A\rV\lV B\rV,\quad
A\in \caA_{[0,l-1]},\;B\in\caA_{[L,\infty)}.
\end{align}

Fix an arbitrary $\varepsilon>0$ and $l\in\nan$.
Choose $l\le L\in\nan$ so that $C_\bb' \lmk s_\bb'\rmk^{L-l}<\frac{\varepsilon}{8}$ and fix.
(Here, $C_\bb' $ and $ s_\bb'$ are given in the previous Lemma.)

We claim
\begin{align}
\lv
\varphi\lmk AB\rmk-\varphi\lmk B\rmk\varphi\lmk A\rmk
\rv
\le
\frac{\varepsilon}4\lV A\rV\lV B\rV,\quad
A\in \caA_{[0,l-1],+},\;B\in\caA_{[L,\infty)}.
\end{align}
If $\varphi(A)=0$, then by the Cauchy-Schwartz inequality, the left hand side is $0$ and the inequality holds.
If  $\varphi(A)\neq0$, then 
\[
\caA_{[0,\infty)}\ni B\mapsto \frac{\varphi\lmk A\tau_l(B)\rmk}{\varphi(A)}
\]
defines a state on $\caA_{[0,\infty)}$. By the Cauchy-Schwartz inequality
and Lemma \ref{lem:fs}, we can check that the state belongs
to $\caS_{[0,\infty)}(H_{\Phi_{m,\bb}})$. We apply the previous Lemma to this state and obtain
\begin{align*}
\lv \frac{\varphi\lmk AB\rmk}{\varphi(A)}
-\omega_{\bb,\infty}\lmk B\rmk\rv
\le
C_\bb' \lmk s_\bb'\rmk^{L-l}\lV B\rV,\quad\;B\in\caA_{[L,\infty)}.
\end{align*}
Considering $A=1$ case, we obtain
\begin{align*}
\lv {\varphi\lmk B\rmk}
-\omega_{\bb,\infty}\lmk B\rmk\rv
\le
C_\bb' \lmk s_\bb'\rmk^{L-l}\lV B\rV,\quad B\in\caA_{[L,\infty)}.
\end{align*}
From these inequalities, we obtain the claim.
That (\ref{eq:br}) follows from the claim is trivial.
\end{proof}

\begin{lem}\label{lem:iv}Let $n\in\nan$ with $n\ge 2$, and $\bb\in\ClassA$ with respect to $(n_0,k_R,k_L,\lal,\bbD,\bbG,Y)$.
Let $m\ge \bbm_\bb$.
There exist $0<C_{\bb}^{''}$, $N_{\bb}\in\nan$, $\omega_{R,\bb}\in \caS_{[0,\infty)}(H_{\Phi_{m,\bb}})$, and 
$\omega_{L,\bb}\in \caS_{(-\infty,-1]}(H_{\Phi_{m,\bb}})$, such that
\begin{align}
&\lv
\frac{\Tr_{[0,N-1]}\lmk G_{N,\bb} A\rmk}{\Tr_{[0,N-1]}\lmk G_{N,\bb}\rmk}-
\omega_{R,\bb}(A)
\rv\le C_{\bb}^{''} s_{\bb}^{N-l}\lV A\rV,\notag\\
&\lv
\frac{\Tr_{[0,N-1]}\lmk G_{N,\bb}\tau_{N-l}\lmk A\rmk\rmk}{\Tr_{[0,N-1]}\lmk G_{N,\bb}\rmk}-
\omega_{L,\bb}\circ\tau_{-l}(A)
\rv\le C_{\bb}^{''} s_{\bb}^{N-l}\lV A\rV,
\end{align}
for all $l\in\nan$, $A\in\caA_{[0,l-1]}$, and $N\ge \max\{l, N_{\bb}\}$,
and 
\begin{align}
\inf\left\{ \sigma\lmk\omega_{R,\bb}\vert_{\caA_{[0,l-1]}}\rmk\setminus \{0\}\mid l\in\nan\right\}>0,\notag\\
\inf\left\{ \sigma\lmk\omega_{L,\bb}\vert_{\caA_{[-l,-1]}}\rmk\setminus \{0\}\mid l\in\nan\right\}>0.
\end{align}
\end{lem}
\begin{proof}
Set \[
N_\bb:=\max\{ l_\bb(n,n_0,k_R,k_L,\lal,\bbD,\bbG,Y),L_\bb\}.
\]
Fix a basis $\{X_i\}_{i=1}^{n_0^2(k_L+1)(k_R+1)}$ of $\mnz\otimes\pu\mkk\pd$
such that
\begin{align}\label{eq:basis}
\braket{X_i}{X_j}_\bb=\delta_{ij}.
\end{align}
Then by Lemma \ref{lem:ge}, we have
\begin{align}
\lv\braket{\gb{N}(X_i)}{\gb{N}(X_j)}-\delta_{ij}\rv\le
 E_{\bb}(N),\quad
N\in\nan.
\end{align}
We claim
 \begin{align}\label{eq:gng}
\lV
G_{N,\bb}-\sum_{i}\ket{\gb{N}(X_i)}\bra{\gb{N}(X_i)}
\rV
\le \lmk
4+\sqrt 3n_0^2(k_L+1)(k_R+1)\rmk
E_{\bb}(N)
,\quad N\ge N_{\bb}.
\end{align}
 To see this, let $N\ge N_{\bb}$ and $\xi\in\bigotimes_{j=0}^{N-1}\cc^n$, and 
 $\xi=\eta_1+\eta_2$, $\eta_1\in \Ran\gb{N}$, $\eta_2\in\lmk \Ran\gb{N}\rmk^{\perp}$
 its orthogonal decomposition.
 Then, 
 as $l_\bb(n,n_0,k_R,k_L,\lal,\bbD,\bbG,Y)\le N_\bb\le  N$, there exists an
 $X\in\mnz\otimes\pu\mkk\pd$ such that
 $\eta_1=\gb{N}({X})$ by Proposition \ref{prop:bbspec}.
 Using Lemma \ref{lem:ge}, and $ \braket{X}{X}_\bb=\sum_i\lv \braket{X}{X_i}_\bb\rv^2$,
 we obtain
 \begin{align*}
&\lv
\braket{\xi}{\lmk G_{N,\bb}-\sum_{i}\ket{\gb{N}(X_i)}\bra{\gb{N}(X_i)}\rmk\xi}
\rv
=\lv\braket{\gb{N}({X})}{\gb{N}({X})}
-\sum_{i}\lv \braket{\gb{N}(X)}{\gb{N}(X_i)}\rv^2
\rv\\
&\le \sqrt 2E_{\bb}(N)
\lmk
\sqrt2\lV \gb{N}({X})\rV^2+\sum_{i}\lmk
\lv \braket{\gb{N}(X)}{\gb{N}(X_i)}\rv\lV \gb{N}({X})\rV
+\lv
 \braket{X}{X_i}_\bb
\rv\lV \gb{N}({X})\rV
\rmk
\rmk\\
&+\lv
\braket{X}{X}_\bb-\sum_i\lv \braket{X}{X_i}_\bb\rv^2
\rv\\
&\le \sqrt 2E_{\bb}(N)
\lmk
\sqrt2\lV \gb{N}({X})\rV^2+\sum_{i}\lmk
\lv \braket{\gb{N}(X)}{\gb{N}(X_i)}\rv\lV \gb{N}({X})\rV
+
 \sqrt 2\lV \gb{N}({X})\rV^2
\rmk
\rmk\\
&\le \sqrt 2E_{\bb}(N)\lV \gb{N}({X})\rV^2
\lmk
2\sqrt2+\sum_{i}\lmk
\lV{\gb{N}(X_i)}\rV
\rmk\rmk\\
&\le  \sqrt 2E_{\bb}(N)\lV \gb{N}({X})\rV^2
\lmk
2\sqrt2+\sqrt{\frac {3}{2}}n_0^2(k_L+1)(k_R+1)\rmk\\
&\le\lmk
4+\sqrt 3n_0^2(k_L+1)(k_R+1)\rmk
E_{\bb}(N)\lV \xi\rV^2
\end{align*}
This proves the claim.
 We set
 \[
 \omega_{L,\bb}:=\frac{1}{n_0^2(k_L+1)(k_R+1)}\sum_{i=1}^{n_0^2(k_L+1)(k_R+1)}\Xi_L(\sigma_{L,X_i})\in\caS_{(-\infty,-1]}(H_{\Phi_{m,\bb}}).
 \]
 
 Fix any $l\in\nan$ and $A\in\caA_{[0,l-1]}$.
 Then for any $N\ge \max\{l,N_{\bb}\}$, we have
 \begin{align*}
& \lv
\frac{\Tr_{[0,N-1]}\lmk G_{N,\bb}\tau_{N-l}\lmk A\rmk\rmk}{\Tr_{[0,N-1]}\lmk G_{N,\bb}\rmk}-
\omega_{L,\bb}\circ\tau_{-l}(A)
\rv\\
&\le 
\lv
\frac{\Tr_{[0,N-1]}\lmk \lmk G_{N,\bb}-\sum_{i=1}^{n_0^2(k_L+1)(k_R+1)}\ket{\gb{N}(X_i)}\bra{\gb{N}(X_i)}\rmk
\tau_{N-l}\lmk A\rmk\rmk}{\Tr_{[0,N-1]}\lmk G_{N,\bb}\rmk}
\rv\\
&+
\lv
\frac{ \sum_{i=1}^{n_0^2(k_L+1)(k_R+1)}\lmk \braket{\gb{N}(X_i)}
{\tau_{N-l}(A)\gb{N}(X_i)}-\Xi_L(\sigma_{L,X_i})(\tau_{-l}(A))\rmk}{\Tr_{[0,N-1]}\lmk G_{N,\bb}\rmk}
\rv\\
&\le\lmk\lmk4+\sqrt 3n_0^2(k_L+1)(k_R+1)
\rmk
E_{\bb}(N)
+ {E_{\bb}(N-l)F_{\bb}}\rmk\lV A\rV
\end{align*}
Here we used (\ref{eq:gng}), Lemma \ref{lem:gee},  and the fact $\Tr_{[0,N-1]}\lmk G_{N,\bb}\rmk=n_0^2(k_L+1)(k_R+1)$
for $N\ge \max\{l,N_\bb\}$.

Set $C:=s_\bb(a_{\bb}c_{\bb})^{-1}n_0^2(k_L+k_R+1)^2\sup_{|z|=s_\bb}\lV(z-T_{\bb})^{-1}\rV$. Then we have
$E_{\bb}(N)\le Cs_\bb^N$.
Set
\[
C_\bb^{''}=C\lmk\lmk 4+\sqrt 3n_0^2(k_L+1)(k_R+1)
\rmk
+ {F_{\bb}}\rmk.
\]
Then we have
\begin{align*}
& \lv
\frac{\Tr_{[0,N-1]}\lmk G_{N,\bb}\tau_{N-l}\lmk A\rmk\rmk}{\Tr_{[0,N-1]}\lmk G_{N,\bb}\rmk}-
\omega_{L,\bb}\circ\tau_{-l}(A)
\rv\le
C_\bb^{''} s_{\bb}^{N-l}\lV A\rV,\quad
N\ge\max\{N_{\bb},l\}.
\end{align*}
From Lemma \ref{lem:ffs}, to show the last statement, it suffices to show that
\[
\kappa:=\frac{1}{n_0^2(k_L+1)(k_R+1)}\sum_{i=1}^{n_0^2(k_L+1)(k_R+1)}\sigma_{L,X_i}
\]
is faithful.
For $\xi\in\cc^{n_0}\otimes \cc^{k_L+1}$, if $\kappa\lmk \ket{\xi}\bra{\xi}\rmk =0$, 
we have
$e_\bb^{\frac 12}X_i\rho_\bb^{\frac 12}\xi=0$, for all $i=1,\ldots, n_0^2(k_L+1)(k_R+1)$.
As $\{X_i\}_{i=1}^{n_0^2(k_L+1)(k_R+1)}$  is a basis of $\mnz\otimes \pu \mkk \pd$,
there exist coefficients $c_i\in\cc$ such that 
\[
\sum_i c_iX_i=\ket{\cnz{1}\otimes \fii{0}}\bra{\xi}\rho_\bb^{-\frac 12}.
\]
Hence we have 
\[
\lV \xi\rV^2e_\bb^{\frac 12}\ket{\cnz{1}\otimes \fii{0}}=0,
\]
and we obtain $\xi=0$.
Therefore, $\kappa$ is faithful.
\end{proof}
\begin{lem}\label{lem:mme}
Let $n\in\nan$ with $n\ge 2$, and $\bb\in\ClassA$ with respect to $(n_0,k_R,k_L,\lal,\bbD,\bbG,Y)$.
Let $m\ge \bbm_\bb$. There exists a constant $C_\bbB^{'''}>0$ satisfying the following. : Let $M\in\nan$ 
and $\varphi$ be a state on $\caA_{\bbZ}$. Assume that we have $\varphi(\tau_{i}(1-\gbp{m}))=0$ for all $i\in\bbZ$ with
$[i,i+m-1]\subset [-M,M]^c$.
Then for any $L\in\nan$ with $M+1\le L$ and $A\in\caA_{[-L+1,L-1]^c}$, we have
\begin{align}
\lv \varphi \lmk A\rmk-\lmk \left. \omega_{\bbB,\infty}\rv_{\caA_{(-\infty,-1]}}\otimes \left. \omega_{\bbB,\infty}\rv_{\caA_{[0,\infty)}}\rmk\lmk A\rmk\rv
\le C_\bbB^{'''} s_\bbB^{L-M}\lV A\rV.
\end{align}
In particular, $\omega_{\bbB,\infty}$ satisfies the exponential decay of correlation functions.
\end{lem}
\begin{proof}
Define  a linear map $\tilde V_l:\lmk \lmk \bigotimes_{i=0}^{l-1}\cc^n\rmk^{\otimes 2}\otimes\lmk\cc^{n_0}\otimes\cc^{k_L+k_R+1}\rmk^{\otimes 2}\rmk
\to \lmk \cc^{n_0}\otimes\cc^{k_L+k_R+1}\rmk^{\otimes 2}$ by
\begin{align*}
\tilde V_l(\xi\otimes \eta):=\sum_{\mu^{(l)},\nu^{(l)}\in\{1,\ldots,n\}^{l}}\braket{\ws{l}\otimes \wsn{l}}{\xi}\lmk \widehat{B_{\mu^{(l)}}}\otimes \widehat{B_{\nu^{(l)}}}\rmk\eta,\quad
\xi\in \lmk \bigotimes_{i=0}^{l-1}\cc^n\rmk^{\otimes 2},\;\;\eta\in \lmk \cc^{n_0}\otimes\cc^{k_L+k_R+1}\rmk^{\otimes 2}.
\end{align*}
For each $L,l\in\nan$ and $A\in \caA_{[-L-l+1,-L]}\otimes \caA_{[L,L+l-1]}$,
we define a linear map
$\tilde \Theta_A:\lmk \mnz\otimes\mkk\rmk^{\otimes 2} \to\lmk \mnz\otimes\mkk\rmk^{\otimes 2}$ by
\begin{align*}
\tilde \Theta_A(X)=\tilde V_l \lmk A\otimes X\rmk \tilde V_l^*,\quad X\in  \lmk \mnz\otimes\mkk\rmk^{\otimes 2}.
\end{align*}Here we identify $\bigotimes_{i=-L-l+1}^{-L}\cc^n$, $\bigotimes_{i=L}^{L+l-1}\cc^n$, with $\bigotimes_{i=0}^{l-1}\cc^n$.
As in the proof of Lemma \ref{lem:gee}, we have $\lV \tilde \Theta_A\rV\le \lV A\rV F_\bbB^2$.
Note that 
\begin{align}\label{eq:tto}
\lmk \varphi_\bbB\otimes \varphi_\bbB\rmk\circ\tilde\Theta_A\lmk e_\bbB\otimes e_\bbB\rmk
=\lmk \left. \omega_{\bbB,\infty}\rv_{\caA_{(-\infty,-1]}}\otimes \left. \omega_{\bbB,\infty}\rv_{\caA_{[0,\infty)}}\rmk\lmk A\rmk.
\end{align}

Fix a basis $\{X_i\}_{i=1}^{n_0^2(k_L+1)(k_R+1)}$ of $\mnz\otimes\pu\mkk\pd$ satisfying (\ref{eq:basis}).
Let $M,N\in\nan$ with $N-M\ge \max\{m_\bbB,l_\bbB\}$. 
Then $\{\gb{N-M}(X_i)\otimes \gb{N-M}(X_j)\}_{i,j=1}^{n_0^2(k_L+1)(k_R+1)}$ is a basis of $\lmk \tau_{-N}\lmk G_{N-M,\bbB}\rmk \bigotimes_{i=-N}^{-M-1}\cc^n\rmk\otimes \lmk \tau_{M+1}\lmk
G_{N-M,\bbB}\rmk\bigotimes_{i=M+1}^{N}\cc^n\rmk$.
(Here we again identify $\bigotimes_{i=-N}^{-M-1}\cc^n$, $\bigotimes_{i=M+1}^{N}\cc^n$, with $\bigotimes_{i=0}^{N-M-1}\cc^n$.)
Let $\xi=\sum_{i,j=1}^{n_0^2(k_L+1)(k_R+1)}c_{i,j}\gb{N-M}(X_i)\otimes \gb{N-M}(X_j)\in \lmk \tau_{-N}\lmk G_{N-M,\bbB}\rmk \bigotimes_{i=-N}^{-M-1}\cc^n\rmk\otimes \lmk \tau_{M+1}\lmk
G_{N-M,\bbB}\rmk\bigotimes_{i=M+1}^{N}\cc^n\rmk$,
with $c_{i,j}\in\bbC$. Let $L,l\in\nan$ with $M+1\le L\le N-l+1$ and $A\in \caA_{[-L-l+1,-L]}\otimes \caA_{[L,L+l-1]}$.
By a straightforward calculation, we obtain
\begin{align*}
&\braket{\xi}{A\xi}\\
&=\sum_{iji'j'}\sum_{\alpha\beta\alpha'\beta'}\sum_{aba'b'}\bar c_{ij}c_{i'j'}\\
&\lmk\Tr\otimes\Tr\rmk\lmk
\begin{gathered}
 \lmk \zeij{\beta\alpha}\otimes \eij{ba}\rmk\otimes   \lmk \zeij{\beta'\alpha'}\otimes \eij{b'a'}\rmk\\
\lmk T_\bbB^{N-L-l+1}\otimes T_\bbB^{L-M-1}\rmk \circ \tilde\Theta_A\circ \lmk T_\bbB^{L-M-1}\otimes T_\bbB^{N-L-l+1}\rmk
\lmk \lmk X_i^*\lmk \zeij{\alpha\beta}\otimes \eij{ab}\rmk X_{i'}\rmk \otimes  \lmk X_j^* \lmk \zeij{\alpha'\beta'}\otimes \eij{a'b'}\rmk X_{j'}\rmk\rmk
\end{gathered}
\rmk.
\end{align*}
Recall the bounds $\lV \tilde \Theta_A\rV\le \lV A\rV F_\bbB^2$ and (\ref{eq:tnb}). 
Using  (\ref{eq:basis}), and the routine argument, we see that there exists a constant $\tilde C_\bbB>0$ such that
\begin{align}\label{eq:xax}
\lv
\braket{\xi}{A\xi}-\sum_{ij}\lv c_{ij}\rv^2\lmk \varphi_\bbB\otimes \varphi_\bbB\rmk\circ\tilde\Theta_A\lmk e_\bbB\otimes e_\bbB\rmk
\rv
\le \tilde C_\bbB \lmk s_\bbB^{N-L-l}+s_{\bbB}^{L-M}\rmk\lV A\rV\sum_{ij}\lv c_{ij}\rv^2.
\end{align}
Substituting (\ref{eq:tto}),
we obtain
\begin{align*}
\lv
\braket{\xi}{A\xi}-\sum_{ij}\lv c_{ij}\rv^2\lmk \left. \omega_{\bbB,\infty}\rv_{\caA_{(-\infty,-1]}}\otimes \left. \omega_{\bbB,\infty}\rv_{\caA_{[0,\infty)}}\rmk\lmk A\rmk
\rv
\le \tilde C_\bbB \lmk s_\bbB^{N-L-l}+s_{\bbB}^{L-M}\rmk\lV A\rV\sum_{ij}\lv c_{ij}\rv^2.
\end{align*}
This holds for any  $M,N\in\nan$ with $N-M\ge \max\{m_\bbB,l_\bbB\}$,
$\xi=\sum_{i,j=1}^{n_0^2(k_L+1)(k_R+1)}c_{i,j}\gb{N-M}(X_i)\otimes \gb{N-M}(X_j)\in \lmk \tau_{-N}\lmk G_{N-M,\bbB}\rmk \bigotimes_{i=-N}^{-M-1}\cc^n\rmk\otimes \lmk \tau_{M+1}\lmk
G_{N-M,\bbB}\rmk\bigotimes_{i=M+1}^{N}\cc^n\rmk$, $L,l\in\nan$ with $M+1\le L\le N-l+1$ and $A\in \caA_{[-L-l+1,-L]}\otimes \caA_{[L,L+l-1]}$.

For each $N\in\nan$, define $L_N:=\max\{M+1,[\frac N2]\}$. Because we have $N-L_N,L_N\to\infty$ as $N\to\infty$, there exists an $\tilde N_\bbB\in\nan$ such that
$\tilde C_\bbB \lmk s_\bbB^{N-L_N-1}+s_{\bbB}^{L_N-M}\rmk<\frac 12$ for all $N\ge \tilde N_\bbB$.
We claim for any
$N\ge\max\{m_\bbB+M,l_\bbB+M,\tilde N_\bbB\}$, $\xi=\sum_{i,j=1}^{n_0^2(k_L+1)(k_R+1)}c_{i,j}\gb{N-M}(X_i)\otimes \gb{N-M}(X_j)\in \lmk \tau_{-N}\lmk G_{N-M,\bbB}\rmk \bigotimes_{i=-N}^{-M-1}\cc^n\rmk\otimes \lmk \tau_{M+1}\lmk
G_{N-M,\bbB}\rmk\bigotimes_{i=M+1}^{N}\cc^n\rmk$, we have $\sum_{ij}\lv c_{ij}\rv^2\le 2\lV \xi\rV^2$.
Applying the above observation with $L:=L_N$, $l=1$, $A=1$, we get this claim
\begin{align*}
\frac 12 \sum_{ij}\lv c_{ij}\rv^2\le \lmk 1-\tilde C_\bbB\lmk s_\bbB^{N-L_N-1}+s_{\bbB}^{L_N-M}\rmk\rmk\sum_{ij}\lv c_{ij}\rv^2\le \lV\xi\rV^2.
\end{align*} 

Substituting this to (\ref{eq:xax}), for all  $N\ge\max\{m_\bbB+M,l_\bbB+M,\tilde N_\bbB\}$, $M+1\le L\le N-l+1$,
$\xi\in \lmk \tau_{-N}\lmk G_{N-M,\bbB}\rmk \bigotimes_{i=-N}^{-M-1}\cc^n\rmk\otimes \lmk \tau_{M+1}\lmk
G_{N-M,\bbB}\rmk\bigotimes_{i=M+1}^{N}\cc^n\rmk$, and $A\in \caA_{[-L-l+1,-L]}\otimes \caA_{[L,L+l-1]}$, we have
\begin{align}\label{eq:xbd}
\lv
\braket{\xi}{A\xi}-\lV\xi\rV^2\lmk \left. \omega_{\bbB,\infty}\rv_{\caA_{(-\infty,-1]}}\otimes \left. \omega_{\bbB,\infty}\rv_{\caA_{[0,\infty)}}\rmk\lmk A\rmk
\rv
\le 4\tilde C_\bbB \lmk s_\bbB^{N-L-l}+s_{\bbB}^{L-M}\rmk\lV A\rV\lV \xi\rV^2.
\end{align}

Let $M\in\nan$ 
and $\varphi$ be a state on $\caA_{\bbZ}$. Assume that we have $\varphi(\tau_{i}(1-\gbp{m}))=0$ for all $i\in\bbZ$ with
$[i,i+m-1]\subset [-M,M]^c$. Let $L\ge M+1$, $l\in\nan$ and $A\in \caA_{[-L-l+1,L+l-1]\setminus [-L+1,L-1]}$.
For any $N\in\nan$ with $N\ge \max\{m_\bbB+M,l_\bbB+M,\tilde N_\bbB,L+l-1\}$, the density matrix of $\rho$ of the restriction of $\varphi$ to
$\caA_{[-N,-M-1]}\otimes \caA_{[M+1,N]}$ can be written as
$\rho=\sum_i\ket{\xi_i}\bra{\xi_i}$ with mutually orthogonal vectors $\xi_i\in \lmk \tau_{-N}\lmk G_{N-M,\bbB}\rmk \bigotimes_{i=-N}^{-M-1}\cc^n\rmk\otimes \lmk \tau_{M+1}\lmk
G_{N-M,\bbB}\rmk\bigotimes_{i=M+1}^{N}\cc^n\rmk$. 
Applying (\ref{eq:xbd}) to $\xi_i$s, we obtain 
\begin{align}\label{eq:xbdd}
\lv
\varphi(A)-\lmk \left. \omega_{\bbB,\infty}\rv_{\caA_{(-\infty,-1]}}\otimes \left. \omega_{\bbB,\infty}\rv_{\caA_{[0,\infty)}}\rmk\lmk A\rmk
\rv
\le 4\tilde C_\bbB \lmk s_\bbB^{N-L-l}+s_{\bbB}^{L-M}\rmk\lV A\rV.
\end{align}
Taking $N\to\infty$ limit, we obtain the result.
\end{proof}

\begin{proofof}[Theorem \ref{thm:asymmetric}]
That $m_\bb\le 2l_\bb(n,n_0,k_R,k_L,\lal,\bbD,\bbG,Y)$ and (i), (ii) are in Proposition \ref{prop:bbspec}.
(iii), (iv), (v), (vii), and (viii) are in Lemma \ref{lem:infty}, Lemma \ref {lem:iv}, 
Lemma \ref{lem:nr}, Lemma \ref{lem:vii} and Lemma \ref{lem:viii},
respectively.
(vi) is from Lemma \ref{lem:fs} and Lemma \ref{lem:leftgs}.
(ix) is Lamma \ref{lem:mme}.
\end{proofof}

\noindent
{\bf Acknowledgment.}\\
{
This work was supported by JSPS KAKENHI Grant Number 25800057 and 16K05171. 
}
\appendix
\section{Notations}\label{sec:nota}
Throughout the article $\inf\emptyset=+\infty$.
We denote the Euclidean distance between a point $x$ and a subset $S$ in $\cc$ (resp. $\rr$) by 
$\dc(x,S)$ (resp. $\drr(x,S)$). 
For a subset $S$ of $\cc$ and $\delta>0$, the $\delta$-neighbourhood of $S$ is denoted by $S_{\delta}$.
We denote the open ball in $\cc$ centered at 
$x\in\cc$ with radius $r$ by $\caB_{r}(x)$.

For a linear space $\caV$, $\dim \caV$ denotes the dimension
of $\caV$. For a vector space $\caV$ and a set of its elements $\{v_{\alpha}\}_{\alpha}\subset\caV$,
$\spa\{v_{\alpha}\}_{\alpha}$ denotes the linear span of $\{v_{\alpha}\}_{\alpha}$ in $\caV$.
We set $\spa\emptyset:=\{0\}$.

For  $k\in\nan$, the set of all $k\times k$ matrices
over $\cc$ is denoted by
$\mk$, while  $\UT_k$ (resp. $\DT_k$) denotes 
the set of all upper (resp. lower) trianguler $k\times k$ matrices.
Furthermore, $\UT_{0,k}$ (resp. $\DT_{0,k}$) denotes the set of elements in $\UT_k$
(resp. $\DT_{k}$)
with $0$ diagonal elements. Let ${\mathfrak E}_k$ denote 
the set of states on ${\Mat}_k$.

For $A\in\mk$, $\lV A\rV$ denotes the uniform norm of $A$, while $\lV A\rV_2$
denotes the Hilbert Schmidt norm.
The set of orthogonal projections in
$\mk$ is denoted by ${\caP}(\mk)$
and the set of positive elements of $\mk$ by $\mk_+$.
Furthermore, we denote the set of unitary elements of $\mk$
by $\caU(\mk)$.
For $A\in\mk_+$, $s(A)$ denotes the support of $A$.
For a positive linear functional $\varphi$ on $\mk$, we
denote its support by $s(\varphi)$ as well.
For $A\in\mk$, $s_l(A)$ (resp. $s_r(A)$) denotes the left (resp. right)
support of $A$.
We write $A>0$ for $A\in\mk$ if $A$ is strictly positive.
The rank of $A\in\mk$ is denoted by $\rank A$.
For a projection $p\in\caP(\mk)$, we denote $1-p$ by $\bar p$.
For subsets $S_1,S_2$ of $\mk$, $S_1\cdot S_2$ denotes the set of matrices
of the form $A_1A_2$ with $A_1\in S_1$, $A_2\in S_2$.
The symbol $\Tr$ denotes the trace of the matrix under consideration.
We denote an inner product given by $\Tr$ by 
$\braket{}{}_{\Tr}$, i.e., $\braket{A}{B}_{\Tr}=\Tr A^*B$
for $A,B\in\mk$.
For $A\in\mk$, we set $\Ad A:\mk\to\mk$ by $\Ad A(B)=A BA^*$, $B\in\mk$.

For $k\in\nan$, we denote the standard basis of $\cc^{k}$
by $\{\chi_{i}^{(k)}\}_{i=1}^{k}$.
We denote the matrix unit $\ket{\chi_{i}^{(k)}}\bra{\chi_{j}^{(k)}}$
by $e_{ij}^{(k)}$.
However, when numbers $k_R,k_L\in\nan\cup\{0\}$ are given
explicitly, we denote the standard basis of $\cc^{k_R+k_L+1}$
by $\{f_{i}^{(k_R,k_L)}\}_{i=-k_R}^{k_L}$, and set $E_{ij}^{(k_R,k_L)}:=\ket{f_i^{(k_R,k_L)}}\bra{f_j^{(k_R,k_L)}}$.
We also use the notation
$\eu:=\sum_{i=-k_R}^{-1}\fii{i}$
and
$\ed:=\sum_{i=1}^{k_L}\fii{i}$.
Furthermore, we define the projections 
$P_R^{(k_R,k_L)}:=\sum_{i=-k_R}^0 E_{ii}^{(k_R,k_L)}$, 
$P_L^{(k_R,k_L)}:=\sum_{i=0}^{k_L} E_{ii}^{(k_R,k_L)}$.
For $n_0\in\nan$,
we also set $\hat P^{(n_0,k_R,k_L)}_R:=\unit_{\mnz}\otimes P_R^{(k_R,k_L)}, \hat P^{(n_0,k_R,k_L)}_L:=\unit_{\mnz}\otimes P_L^{(k_R,k_L)}
$,
$\hat E_{ij}^{(k_R,k_L)}=\unit\otimes E_{ij}^{(k_R,k_L)}$.
For $a=-k_R,\ldots,k_L$ we set
$Q_{R,a}^{(k_R,k_L)}:=
\sum_{i=-k_R}^a E_{ii}^{(k_R,k_L)}$,
$Q_{L,a}^{(k_R,k_L)}:=
\sum_{i=a}^{k_L} E_{ii}^{(k_R,k_L)}$,
$\hat Q_{R,a}^{(n_0,k_R,k_L)}:=\unit_{\mnz}\otimes
\sum_{i=-k_R}^a E_{ii}^{(k_R,k_L)}$
and $\hat Q_{L,a}^{(n_0,k_R,k_L)}:=\unit_{\mnz}\otimes
\sum_{i=a}^{k_L} E_{ii}^{(k_R,k_L)}$.
In this terminology, we have $\hat P^{(n_0,k_R,k_L)}_U=
\hat Q_{R,0}^{(n_0,k_R,k_L)}$
and $\hat P^{(n_0,k_R,k_L)}_L=\hat Q_{L,0}^{(n_0,k_R,k_L)}$.
For $X\in\mnzk$ and $i,j=-k_R,\ldots,k_L$, 
we write the $(i,j)$-element of $X$ by $X_{ij}$, i.e.,
$X=\sum_{-k_R\le i,j\le k_L} X_{ij}\otimes E_{ij}^{(k_R,k_L)}$.
We set $\cn:=\CN$.
For $k_R,k_L\in\nan\cup\{0\}$, we define linear maps $I_R^{(k_R,k_L)}:\Mat_{k_R+1}\to\Mat_{k_L+k_R+1}$, and
$I_L^{(k_R,k_L)}:\Mat_{k_L+1}\to\Mat_{k_L+k_R+1}$
by $I_R^{(k_R,k_L)}(E_{ij}^{(k_R,0)})=E_{ij}^{(k_R,k_L)}$, $i,j=-k_R,\ldots,0$,
$I_L^{(k_R,k_L)}(E_{ij}^{(0,k_L)})=E_{ij}^{(k_R,k_L)}$, $i,j=0,\ldots,k_L$.

For a linear operator $T$ acting on a vector space, we denote the spectrum of $T$ by $\sigma(T)$, and the spectral radius of $T$ by $r_T$.
For an isolated  subset $S$ of $\sigma(T)$, we denote the spectral projection 
of $T$ onto $S$ by $P_{S}^T$.
If $T$ is a self-adjoint operator on a Hilbert space and $S$ a subset of $\rr$, then $\Proj[T\in S]$ also indicates the spectral projection
of $T$ corresponding to $\sigma(T)\cap S$.
For a linear map $\Gamma$, $\ker\Gamma$, and $\Ran\Gamma$ denote the kernel and the range of $\Gamma$ respectively.

For a linear map $T:\mk\to\mk$, $A\in\mk$ is $T$-invariant if $T(A)=A$.
A linear functional $\psi$ is $T$-invariant if $\psi\circ T=\psi$.

For a finite dimensional Hilbert space, braket $\braket{}{}$ denotes the
inner product of
the space under consideration.
We denote the set of all bounded linear operators on 
$\mathcal H$ by
$B({\mathcal H})$.
For  a subspace $\mathfrak K$, ${\mathfrak{K}}^{\perp}$ means the orthogonal complement of $\mathcal H$ .

For an $n$-tuple of matrices $\vv=(v_1,\ldots,v_n)\in\Mat_k^{\times n}$ and $R\in \GL(k,\cc)$,
we denote the $n$-tuple $(R v_1R^{-1},\ldots,R v_nR^{-1})\in\Mat_k^{\times n}$
by $R \vv R^{-1}$.

\section{Proof of Lemma \ref{basiccp}}\label{apbasiccp}
{\it 1.} is obtained by the repeated use of $v_{\mu}p=pv_{\mu}p$.
From the same relation, we see
that
\begin{align}\label{bpvp}
\bar p\lmk\widehat{ v_{\mu^{(N)}}}\rmk^*p 
=\sum_{m=1}^N\lmk\bar p v_{\mu_{N}}^*\bar p\rmk\cdots
\lmk\bar p v_{\mu_{m+1}}^*\bar p\rmk\lmk\bar p v_{\mu_{m}}^* p\rmk\lmk p v_{\mu_{m-1}}^* p\rmk
 \cdots\lmk p v_{\mu_{1}}^* p\rmk.
\end{align}
For $N\in\nan$, $m\in\{1,\ldots,N\}$, $\eta\in\cc^k$, and a CONS $\{\chi_i^{(k)}\}_{i=1}^k$ of $\cc^k$,
we have
\begin{align}
&\sum_{\mu^{(N)}\in \{1,\ldots,n\}^{\times N}}\lV
\lmk\bar p v_{\mu_{N}}^*\bar p\rmk\cdots
\lmk\bar p v_{\mu_{m+1}}^*\bar p\rmk\lmk\bar p v_{\mu_{m}}^* p\rmk\lmk p v_{\mu_{m-1}}^* p\rmk
 \cdots\lmk p v_{\mu_{1}}^* p\rmk\eta
\rV^2\nonumber\\
&= \sum_{\mu^{(N)}\in \{1,\ldots,n\}^{\times N}}\sum_{i=1}^k
\lv 
\braket{\lmk\bar p v_{\mu_{m+1}}
\bar p\rmk\cdots\lmk\bar p v_{\mu_{N}}\bar p\rmk \chi_{i}^{(k)}}
{\lmk\bar p v_{\mu_{m}}^* p\rmk\lmk p v_{\mu_{m-1}}^* p\rmk
 \cdots\lmk p v_{\mu_{1}}^* p\rmk\eta}\rv^2\nonumber\\
&\le\sum_{i=1}^k
 \sum_{\mu^{(N)}\in \{1,\ldots,n\}^{\times N}}
 \lV \lmk\bar p v_{\mu_{m+1}}\bar p\rmk\cdots\lmk\bar p v_{\mu_{N}}\bar p\rmk 
\chi_{i}^{(k)}\rV^2
 \lV v_{\mu_{m}}\rV^2\lV
\lmk p v_{\mu_{m-1}}^* p\rmk
 \cdots\lmk p v_{\mu_{1}}^* p\rmk\eta\rV^2\nonumber\\ 
&=
 \lmk \Tr\lmk T^{N-m}_{\vv_{\bar p}}\lmk \sum_{i=1}^k
e_{ii}^{(k)}\rmk\rmk\rmk
 \braket{\eta}{T^{m-1}_{\vv_p}(1)\eta}
 \lmk \sum_{\mu=1}^n\lV v_{\mu}\rV^2\rmk\label{toch1}.
\end{align}
From (\ref{toch1})  and (\ref{bpvp}), we obtain
{\it 2}:
\begin{align}
&\sum_{\mu^{(N)}\in \{1,\ldots,n\}^{\times N}}\lV\bar p\lmk\widehat{ v_{\mu^{(N)}}}\rmk^*p  \eta\rV^2\nonumber\\
&\le \sum_{\mu^{(N)}\in \{1,\ldots,n\}^{\times N}}
\lmk
\sum_{m=1}^N
\lV
\lmk\bar p v_{\mu_{N}}^*\bar p\rmk\cdots
\lmk\bar p v_{\mu_{m+1}}^*\bar p\rmk\lmk\bar p v_{\mu_{m}}^* p\rmk\lmk p v_{\mu_{m-1}}^* p\rmk
 \cdots\lmk p v_{\mu_{1}}^* p\rmk\eta
\rV\rmk^2\nonumber\\
&=\sum_{m_1=1}^N\sum_{m_2=1}^N\sum_{\mu^{(N)}\in \{1,\ldots,n\}^{\times N}}
\lV
\lmk\bar p v_{\mu_{N}}^*\bar p\rmk\cdots
\lmk\bar p v_{\mu_{m_1+1}}^*\bar p\rmk\lmk\bar p v_{\mu_{m_1}}^* p\rmk\lmk p v_{\mu_{m_1-1}}^* p\rmk
 \cdots\lmk p v_{\mu_{1}}^* p\rmk\eta
\rV\nonumber\\
&\quad\quad\lV
\lmk\bar p v_{\mu_{N}}^*\bar p\rmk\cdots
\lmk\bar p v_{\mu_{m_2+1}}^*\bar p\rmk\lmk\bar p v_{\mu_{m_2}}^* p\rmk\lmk p v_{\mu_{m_2-1}}^* p\rmk
 \cdots\lmk p v_{\mu_{1}}^* p\rmk\eta
\rV\nonumber\\
&\le
\sum_{m_1=1}^N\sum_{m_2=1}^N\lmk
\sum_{\mu^{(N)}\in \{1,\ldots,n\}^{\times N}}
\lV
\lmk\bar p v_{\mu_{N}}^*\bar p\rmk\cdots
\lmk\bar p v_{\mu_{m_1+1}}^*\bar p\rmk\lmk\bar p v_{\mu_{m_1}}^* p\rmk\lmk p v_{\mu_{m_1-1}}^* p\rmk
 \cdots\lmk p v_{\mu_{1}}^* p\rmk\eta
\rV^2\rmk^{\frac 12}\nonumber\\
&\quad\quad\lmk
\sum_{\mu^{(N)}\in \{1,\ldots,n\}^{\times N}}
\lV
\lmk\bar p v_{\mu_{N}}^*\bar p\rmk\cdots
\lmk\bar p v_{\mu_{m_2+1}}^*\bar p\rmk\lmk\bar p v_{\mu_{m_2}}^* p\rmk\lmk p v_{\mu_{m_2-1}}^* p\rmk
 \cdots\lmk p v_{\mu_{1}}^* p\rmk\eta
\rV^2\rmk^{\frac 12}\nonumber\\
&\le
\lmk \sum_{m=1}^N\lmk \Tr T_{\vv_{\bar p}}^{N-m}(1)\rmk^{\frac 12}
\braket{\eta}{T_{\vv_p}^{m-1}(1)\eta}^{\frac 12}\rmk^{2}\sum_{\mu=1}^n\lV v_{\mu}\rV^2.
\end{align}
To see {\it 3}, we first bound $\lV \bar p T^N_{\vv}(A)\rV$ for $A\in\mk$.
By {\it 1} and {\it 2},
\begin{align*}
&\lv
\braket{\bar pT^N_{\vv}(A)\xi}{\eta}
\rv
 =\lv
  \sum_{\mu^{(N)}\in \{1,\ldots,n\}^{\times N}}
\braket{A\lmk \widehat{v_{\mu^{(N)}}}\rmk^* \xi}
 {\bar p \lmk \widehat{v_{\mu^{(N)}}}\rmk^* \bar p \eta}
 \rv\\
 &\le
 \sum_{\mu^{(N)}\in \{1,\ldots,n\}^{\times N}}
 \lV A\rV
\lV
 \lmk \widehat{v_{\mu^{(N)}}}\rmk^* \xi
 \rV
 \lV
 \bar p \lmk \widehat{v_{\mu^{(N)}}}\rmk^* \bar p \eta
 \rV\\
& \le
  \lV A \rV
  \sum_{\mu^{(N)}}\lmk
  \lV p\lmk \widehat{v_{\mu^{(N)}}}\rmk^* p\xi\rV+\lV \bar p\lmk \widehat{v_{\mu^{(N)}}}\rmk^*p\xi\rV
  +\lV \bar p\lmk \widehat{v_{\mu^{(N)}}}\rmk^*\bar p\xi\rV
  \rmk
 \cdot \lV \bar p\lmk \widehat{v_{\mu^{(N)}}}\rmk^*\bar p\eta\rV\\
& \le
  \lV A \rV\cdot 
\lmk
\sum_{\mu^{(N)}}\lV \bar p\lmk \widehat{v_{\mu^{(N)}}}\rmk^*\bar p\eta\rV^2
\rmk^{\frac 12}\\
&\quad \lmk
  \lmk \sum_{\mu^{(N)}}\lV p\lmk \widehat{v_{\mu^{(N)}}}\rmk^* p\xi\rV^2\rmk^{\frac 12} +\lmk\sum_{\mu^{(N)}}\lV \bar p\lmk \widehat{v_{\mu^{(N)}}}\rmk^*\bar p\xi\rV^2\rmk^{\frac 12}
+\lmk\sum_{\mu^{(N)}}\lV \bar p\lmk \widehat{v_{\mu^{(N)}}}\rmk^*p\xi\rV^2\rmk^{\frac 12} 
\rmk \\
&
=\lV A \rV\cdot 
\lmk
\braket{\xi}{T_{\vv_p}^N(1)\xi}^{\frac 12}
+
\braket{\xi}{T_{\vv_{\bar p}}^N(1)\xi}^{\frac 12}
+\lmk\sum_{\mu^{(N)}}\lV \bar p\lmk \widehat{v_{\mu^{(N)}}}\rmk^* p\xi\rV^2\rmk^{\frac 12}
\rmk
\braket{\eta}{T_{\vv_{\bar p}}^N(1)\eta}^{\frac 12}
\\
&\le
\lV A \rV
 \lV T_{\vv_{\bar p}}^N(1)\rV^{\frac 12}\\
& \lmk
 \sup_{M\in\nan}\lV T_{\vv_p}^M\rV^{\frac 12} +\sup_{M\in\nan}\lV T_{\vv_{\bar p}}^M\rV^{\frac 12}
 +\lmk \sum_{m=1}^N\lmk \Tr T_{\vv_{\bar p}}^{N-m}(1)\rmk^{\frac 12}
\lV T_{\vv_p}^{m-1}(1)\rV^{\frac 12}\rmk\lmk \sum_{\mu=1}^n\lV v_{\mu}\rV^2\rmk^{\frac 12}
 \rmk\lV\xi\rV\lV \eta\rV
\end{align*}
for any $\xi,\eta \in\cc^k$.
For the third inequality, we used the Cauchy-Schwarz inequality.
From this, we obtain
\begin{align*}
&\lV T^N_{\vv}(A)-pT^N_{\vv}(A)p\rV
=
\lV \bar pT^N_{\vv}(A)+pT^N_{\vv}(A)\bar p\rV\\
&\le
2\lV A \rV
 \lV T_{\vv_{\bar p}}^N(1)\rV^{\frac 12}\\
&\quad \lmk \sup_{M\in\nan}\lV T_{\vv_p}^M\rV^{\frac 12} +\sup_{M\in\nan}\lV T_{\vv_{\bar p}}^M\rV^{\frac 12}
+\lmk \sum_{m=1}^N\lmk \Tr T_{\vv_{\bar p}}^{N-m}(1)\rmk^{\frac 12}
\lV T_{\vv_p}^{m-1}(1)\rV^{\frac 12}\rmk\lmk \sum_{\mu=1}^n\lV v_{\mu}\rV^2\rmk^{\frac 12}
\rmk
\end{align*}
To see {\it 4}, note that as $T^M_{\vv}$ is positive, we have
\begin{align*}
&\lV T^M_{\vv}\rV =\lV T^M_{\vv}(1)\rV
\le
\lV T^M_{\vv}(p)\rV+\lV T^M_{\vv}(\bar p)\rV
\le \lV T^M_{\vv}(p)\rV
+\lV \bar p T^M_{\vv}(\bar p)\bar p\rV
+\lV \bar p T^M_{\vv}(\bar p) p\rV
+\lV p T^M_{\vv}(\bar p)\bar p\rV+
\lV p T^M_{\vv}(\bar p)p\rV\\
&
\le \lV T^M_{\vv_p}(p)\rV
+\lV T^M_{\vv_{\bar p}}(\bar p)\rV
+
2\lV p T^M_{\vv}(\bar p)p\rV^{\frac 12}
\lV \bar p T^M_{\vv}(\bar p)\bar p\rV^{\frac 12}
+
\lV p T^M_{\vv}(\bar p)p\rV\\
&\le
\sup_{M\in\nan}\lV T^M_{\vv_p}\rV+\sup_{M\in\nan}\lV \ T^M_{\vv_{\bar p}}\rV
+2\lmk \sum_{m=1}^M\lmk \Tr T_{\vv_{\bar p}}^{M-m}(1)\rmk^{\frac 12}
{\lV T_{\vv_p}^{m-1}(1)\rV}^{\frac 12}\rmk
\lmk
\sum_{\mu=1}^n\lV v_{\mu}\rV^2\rmk^{\frac 12}
\lV T^M_{\vv_{\bar p}}(\bar p)\rV^{\frac 12}\\
&+\lmk \sum_{m=1}^M\lmk \Tr T_{\vv_{\bar p}}^{M-m}(1)\rmk^{\frac 12}
{\lV T_{\vv_p}^{m-1}(1)\rV}^{\frac 12}\rmk^{2}\sum_{\mu=1}^n\lV v_{\mu}\rV^2.
\end{align*}
Here we used {\it 2}.
\section{Matrix algebra and CP maps}
In this section we collect useful facts about matrix
algebras. See \cite{Wolf:2012aa}, for example.

\begin{lem}\label{lem:lb}
Let $m,k\in\nan$, with $m\le k$, and
$\{\xi_i\}_{i=1}^m$, a set of 
vectors of $\cc^k$.
Let $A$ be an $m\times m$ matrix
given by
$A=\lmk\braket{\xi_i}{\xi_j}\rmk_{i,j=1}^m$.
Let $X:=\sum_{i=1}^m\ket{\xi_i}\bra{\xi_i}\in\Mat_k$
and $P$ be the support projection of $X$.
Suppose that there  exists a positive constant
$c$ such that
$
c\unit\le A$.
Then we have
$
cP\le X$.In particular, we have
$\sigma(X)\setminus \{0\}\subset [c,\lV X\rV]$.
\end{lem}
\begin{thm}\label{pm1}
Let $T:\Mat_k(\cc)\to \Mat_k(\cc)$ be a
positive linear map. The following properties are equivalent:
\begin{enumerate}
\item There is no nontrivial 
orthogonal projection $P$ such that
$T(P\Mat_k(\cc)P)\subset
P\Mat_k(\cc)P$.
\item For any
nonzero $A\geq0$ and $t>0$,
$\exp(tT)(A)>0$.
\end{enumerate}
\end{thm}
\begin{rem}
A positive map satisfying the above equivalent conditions is said to be irreducible.
\end{rem}
We say that $\lambda$ is a non degenerate eigenvalue of $T$
if the corresponding projection $P_{\{\lambda\}}^T$ is one dimensional. Irreducible positive maps satisfy the following properties.
\begin{thm}\label{pm2}
Let
$T:\Mat_k(\cc)\to \Mat_k(\cc)$ be a nonzero irreducible positive linear 
map.
Then the spectral radius $r_T$ of $T$ is a 
strictly positive, non-degenerate eigenvalue with a strictly positive 
eigenvector $h_T$:
\begin{equation*}
T(h_T) = r_T h_T >0.
\end{equation*}
\end{thm}
\begin{thm}\label{pm3}
Let $T:\Mat_k(\cc)\to \Mat_k(\cc)$ be a unital completely positive map and let
\begin{align*}
T(A)=\sum_{i=1}^{n}v_iAv_i^* 
\end{align*}
be its Kraus decomposition. Let $\vv := (v_1,\ldots,v_n)$. Then the following properties are equivalent:
\begin{enumerate}
\item
There exists $l\in\nan$ such that $T^l(A)>0$
for any nonzero $A\geq 0$.
\item
There exists a unique faithful $T$-invariant state $\varphi$, and it satisfies
 \[
\lim_{l\to\infty} T^l(A)=\varphi(A) 1,\quad
A\in \Mat_k(\cc).
\]
\item
$\sigma(T)\cap\{z\in\cc:|z|\ge 1\}=\{1\}$, $1$ is a non degenerate eigenvalue of $T$, and there exists a faithful $T$-invariant state.
\item
There exists $m\in\nan$
such that 
${\mathcal K}_m(\vv)=\Mat_k(\cc)$.
\item
There exists $m\in\nan$
such that 
${\mathcal K}_l(\vv)=\Mat_k(\cc)$, for all $l\ge m$.
\end{enumerate}
\end{thm}
\begin{lem}\label{lem:ps}
Let $n,n_0\in\nan$ and $\oo\in\Primz(n,n_0)$.
Then $r_{T_{\oo}}>0$ and
there exist constants $c>0$, $0<s<1$, 
a faithful positive linear functional $\varphi$ on $\mnz$
and a strictly
positive element $e\in \mnz$ 
such that
\[
\lV r_{T_\oo}^{-N}T^N_{\oo}\lmk A\rmk-\varphi(A)e\rV\le
cs^N\lV A\rV,\quad
\text{for all} \; A\in \mnz,\text{and}\;N\in\nan.
\]
\end{lem}
\begin{proof}
It is easy to check that for $\oo\in\Primz(n,n_0)$,
$T_{\oo}$ is irreducible,
from Lemma \ref{pm1}.
Hence from Lemma \ref{pm2}, $r_{T_\oo}>0$
is a non degenerate eigenvalue of
$T_{\omega}$ with strictly positive eigenvector
$h_{\oo}$.
Define 
$\hat T_{\oo}:=r_{T_\oo}^{-1}h_{\oo}^{-\frac 12}T_{\oo}\lmk h_{\oo}^{\frac 12}\cdot 
h_{\oo}^{\frac 12}\rmk h_{\oo}^{-\frac 12}$.
Then this $\hat T_{\oo}$ is unital and satisfies
condition {\it 4} of Theorem \ref{pm3}.
From Theorem \ref{pm3}, there exists $0<s_\oo<1$ such that
$\sigma(\hat T_{\oo})\setminus \{1\}\subset \caB_{s_\oo}(0),$ 
$1$ is a non degenerate eigenvalue of $\hat T_{\oo}$, 
and there exists a faithful $\hat T_{\oo}$-invariant state $\varphi_{\oo}$.
We have
\[
\lV
\hat T_{\oo}^N(A)-\varphi_\oo(A) 1
\rV
\le s_\oo^N
C_\oo\lV A\rV,\quad A\in\mnz,\quad N\in\nan,
\]
for some $C_\oo>0$.
The claim of Lemma can be checked immediately from this.
\end{proof}
\begin{lem}\label{lem:lsl}
For $l,k,m\in\nan$ with $k\le l$, 
define
\begin{align*}
v_{a}\lmk s\rmk
=\lmk
\begin{array}{c}
{}_lC_{a}\cdot s^l\\
{}_{l+1}C_{a}\cdot s^{l+1}\\
\vdots\\
{}_{l+km-1}C_{a}\cdot s^{l+km-1}
\end{array}
\rmk,\quad
w_{a}\lmk s\rmk
=\lmk
\begin{array}{c}
l^as^l\\
(l+1)^as^{l+1}\\
\vdots\\
\lmk l+km-1\rmk ^a s^{l+km-1}
\end{array}
\rmk
\in\cc^{km},\quad
a=0,\ldots,k-1,\quad s\in\cc.
\end{align*}
Let $\{s_i\}_{i=1}^m$ be distinct elements in $\cc\setminus\{0\}$.
Then we have the followings.
\begin{enumerate}
\item The vectors
$\{v_a(s_i)\}_{a=0,\ldots,k-1,\;
i=1,\ldots,m}$
are linearly independent. In particular, there exist 
vectors $\zeta_{a,i}=\lmk \zeta_{a,i}(j)\rmk_{j=0}^{km-1}
\in\cc^{km}$, $a=0,\ldots,k-1,\;
i=1,\ldots,m$ such that
\[
\sum_{j=0}^{km-1}{}_{l+j}C_{a'}\cdot s_{i'}^{l+j} \cdot \zeta_{ai}(j)=\delta_{aa'}\delta_{ii'}.
\]
\item The vectors
$\{w_a(s_i)\}_{a=0,\ldots,k-1,\;
i=1,\ldots,m}$
are linearly independent. In particular, there exist 
vectors $\xi_{a,i}=\lmk \xi_{a,i}(j)\rmk_{j=0}^{km-1}\in\cc^{km}$, $a=0,\ldots,k-1,\;
i=1,\ldots,m$ such that
\[
\sum_{j=0}^{km-1}\lmk l+j\rmk^{a'} s_{i'}^{j+l} \xi_{ai}(j)=\delta_{aa'}\delta_{ii'}.
\]
\end{enumerate}
\end{lem}
\begin{proof}
This can be checked by the use of Vandermonde determinant.
\end{proof}

\end{document}